\documentclass[twocolumn,secnumarabic,nobibnotes,superscriptaddress,aps,nofootinbib,pra,10pt]{revtex4-2}
\usepackage{color, xcolor, colortbl}
\usepackage{graphicx}
\usepackage{amsthm,amsmath,amssymb,amsfonts}
\usepackage{dcolumn}
\usepackage{url}
\usepackage{epstopdf}
\usepackage{enumitem}
\usepackage{algorithm}
\usepackage{algpseudocode}
\usepackage{bm}
\usepackage[caption=false]{subfig}
\usepackage{appendix}
\usepackage{multirow}
\usepackage{braket}
\usepackage[english]{babel}
\usepackage{hyperref}
\usepackage[capitalize]{cleveref}
\usepackage{xpatch}
\usepackage{tikz}
\usepackage{adjustbox}
\usepackage{xspace}
\usetikzlibrary{quantikz}

\newcommand{\ud}{\,\mathrm{d}}
\newcommand{\Or}{\mathcal{O}}

\newcommand{\RR}{\mathbb{R}}
\newcommand{\CC}{\mathbb{C}}
\newcommand{\ZZ}{\mathbb{Z}}

\renewcommand{\Re}{\operatorname{Re}}
\renewcommand{\Im}{\operatorname{Im}}

\newcommand{\poly}{\operatorname{poly}}
\newcommand{\Tr}{\operatorname{Tr}}

\newcommand{\mc}[1]{\mathcal{#1}}

\newcommand{\wt}[1]{\widetilde{#1}}

\newcommand{\abs}[1]{\left\lvert#1\right\rvert}
\newcommand{\norm}[1]{\left\lVert#1\right\rVert}

\newtheorem{thm}{\protect\theoremname}
\newtheorem{lem}[thm]{\protect\lemmaname}
\newtheorem{rem}[thm]{\protect\remarkname}
\newtheorem{prop}[thm]{\protect\propositionname}
\newtheorem{cor}[thm]{\protect\corollaryname}
\newtheorem{assumption}[thm]{\protect\assumptionname}
\newtheorem{defn}[thm]{\protect\definitionname}

\providecommand{\definitionname}{Definition}
\providecommand{\assumptionname}{Assumption}
\providecommand{\corollaryname}{Corollary}
\providecommand{\lemmaname}{Lemma}
\providecommand{\propositionname}{Proposition}
\providecommand{\remarkname}{Remark}
\providecommand{\theoremname}{Theorem}
\usepackage{bbm}

\usepackage{silence}


\usetikzlibrary{fit}
\tikzset{%
  highlight/.style={rectangle,rounded corners,fill=blue!15,draw,fill opacity=0.3,thick,inner sep=0pt}
}

\graphicspath{{figures/}}
\newcommand{\leftprod}{\ensuremath{\prod^{\boldsymbol{\leftarrow}}}}
\newcommand{\rightprod}
{\ensuremath{\prod^{\boldsymbol{\rightarrow}}}}

\newcommand{\revise}[1]{\textcolor{black}{#1}}

\newcommand{\tmix}{\ensuremath{t_\mathrm{mix}}}
\newcommand{\DeptMath}{Department of Mathematics, University of California, Berkeley, California 94720, USA}
\newcommand{\LBLMath}{Applied Mathematics and Computational Research Division, Lawrence Berkeley National Laboratory, Berkeley, California 94720, USA}
\newcommand{\Caltech}{
Institute for Quantum Information and Matter, California Institute of Technology, Pasadena, California 91125, USA}
\newcommand{\CIQC}{Challenge Institute of Quantum Computation, University of California, Berkeley, California 94720, USA}
\begin{document}

\title{Single-ancilla ground state preparation via Lindbladians}

\begin{abstract}
We design a quantum algorithm for ground state preparation in the early fault tolerant regime. As a Monte Carlo-style quantum algorithm, our method features a Lindbladian where the target state is stationary. The construction of this Lindbladian is  algorithmic and  should not be seen as a specific approximation to some weakly coupled system-bath dynamics in nature.
Our algorithm can be implemented using just one ancilla qubit and efficiently simulated on a quantum computer. It can prepare the ground state even when the initial state has zero overlap with the ground state, bypassing the most significant limitation of methods like quantum phase estimation. As a variant, we also propose a discrete-time algorithm, demonstrating even better efficiency and providing a near-optimal simulation cost depending on the desired evolution time and precision. Numerical simulation using Ising and Hubbard models demonstrates the efficacy and applicability of our method. 
\end{abstract}
\author{Zhiyan Ding}
\affiliation{\DeptMath}
\author{Chi-Fang Chen}
\affiliation{\Caltech}
\author{Lin Lin}
\email{linlin@math.berkeley.edu}
\affiliation{\DeptMath}
\affiliation{\LBLMath}
\affiliation{\CIQC}
\maketitle


\section{Introduction}
A promising application of quantum computers is to simulate ground state properties of quantum many-body systems~\cite{Alan_2005,Lanyon_2010,Veis_2010,Ge_2019,Brien_2019,LinTong2020a}. To concretely evaluate the end-to-end algorithmic cost, however, the state preparation problem rises as a major conceptual bottleneck~\cite{Malley_2016,LeeLeeZhaiEtAl2023}.

From a complexity theory standpoint, few-body Hamiltonian ground states can be QMA-hard to prepare~\cite{KitaevShenVyalyi2002,AharonovNaveh2002,KempeKitaevRegev2006, OliveiraTerhal2005, AharonovGottesmanEtAl2009}.
Thus we do not expect quantum computers to efficiently prepare ground states for \textit{every} few-body Hamiltonian. While this worst-case hardness may not apply to practically relevant systems, it explains the theoretical obstacles towards \textit{proving} algorithmic guarantees. Indeed, justifying the efficacy of most existing ground state algorithms requires additional assumptions. The most common and transparent assumption is the existence of an easy-to-prepare (pure or mixed) quantum state $\rho_0$ with a good overlap with the ground state $\ket{\psi_0}$, i.e., $p_0=\braket{\psi_{0}|\rho_0|\psi_{0}}=\Omega(1/\mathrm{poly}(n))$
where $n$ is the system size. This assumption allows us to provably solve ground state preparation problems with a cost scaling with $\poly(1/p_0)$ (among other dependences). For instance, the cost associated with quantum phase estimation (QPE) scales as $\Or(1/p^2_0)$ or $\Or(1/p_0)$ depending on the implementation~\cite{NielsenChuang2000,KitaevShenVyalyi2002,LinTong2022}, while nearly optimal post-QPE algorithms exhibit a scaling of $\Or(1/\sqrt{p_0})$~\cite{LinTong2020a,DongLinTong2022}.
Unfortunately, the above strategy demands an instance-dependent choice of good trial states and can be nontrivial to justify in practically relevant systems~\cite{LeeLeeZhaiEtAl2023}.

From a thermodynamics standpoint, however, ground state (and more generally, low-energy Gibbs states $\rho_{\beta}  = e^{-\beta H}/\Tr[e^{-\beta H}]$ with a large inverse temperature $\beta$) should be easy to prepare: just put the sample into
a low-temperature fridge.
Drawing from this intuition, several studies ~\cite{kaplan2017ground,Polla_2021,rost2021demonstrating,Cat_2023,Metcalf_2022,mi2023stable} explore the idea of creating an appropriate cold bath and establishing an effective system-bath coupling to efficiently cool the system to a low-energy state.
Nonetheless, to prepare a low energy state, or ground state with desired precision, it remains unclear how strong the coupling between system-bath should be, or how
big of a bath is required. Guided by thermodynamic intuition, a new wave of \textit{Monte Carlo} style quantum algorithms
have been proposed to extract the precise working principle behind cooling~\cite{Temme_2011,Man_2012,MozgunovLidar2020,shtanko2023preparing,cubitt2023dissipative,Rall_thermal_22,chen2023fast,ChenKastoryanoBrandaoEtAl2023,ChenKastoryanoGilyen2023,ding2024efficient}.

The main tool used in this work is the Lindblad dynamics.  Lindblad dynamics (in particular, when the Lindbladian is a Davies generator~\cite{Davies1974,Davies1976}) is often regarded as an approximation of certain unitary dynamics with weak system-bath interactions, derived through the Born-Markov-Secular\footnote{The secular approximation is also referred to as the rotating wave approximation (RWA).} approximation route~\cite{BreuerPetruccione2002,Lidar2019}. Consequently, the scope of Lindblad dynamics seems to be constrained by the limitation of these approximations. However, as already revealed in the original GKLS formalism~\cite{GoriniKossakowskiSudarshan1976,Lindblad1976}, the applicability range of the Lindblad dynamics is in fact much broader: the generator of \emph{any }quantum Markov semigroup must take the form of a Lindbladian. Several recent works have taken this latter perspective, and have designed Lindblad generators $\mc{L}$ so that $\mc{L}[\rho_{\beta}]=0$ holds approximately~\cite{MozgunovLidar2020,Rall_thermal_22,chen2023fast,ChenKastoryanoBrandaoEtAl2023} or even exactly~\cite{ChenKastoryanoGilyen2023,ding2024efficient}. Then if the Lindblad dynamics $\frac{\ud \rho}{\ud t} = \mc{L}[\rho]$ is relaxing, one may simply evolve the dynamics to obtain the  Gibbs state as the fixed point of the dynamics. Note that this is entirely an algorithmic procedure, and the cooling process may or may not approximate any specific system-bath dynamics occurring in nature.
In this sense, such a quantum algorithm can be regarded as a cousin of classical Markov chain Monte Carlo (MCMC) methods, and in particular,  artificial thermostat techniques~\cite{FrenkelSmit2002}.

 As long as we prescribe the fixed point to be our
target  state, the algorithmic cost per sample is
\[
(\text{simulation cost per unit time}) \times (\text{mixing time}).
\]
Thus, the hardness of low-energy problems, in this framework, reduces to the mixing time of certain quantum Markov chains. The mixing time can be
very much case-dependent for classical Markov chains; we expect similar rich
behavior to transfer to the quantum case, and in particular, we do not expect
this algorithm to efficiently solve QMA-hard problems (just as we do not expect
classical Monte Carlo methods to solve NP-hard problems efficiently). While the
efficacy of the QPE and the Monte Carlo approaches both rely on additional
assumptions, we hope the quantum Markov chain can be constructed more systematically than the ansatz state. Furthermore, based on the experimental success of cooling, we might expect the mixing time for physically relevant systems to be reasonably short, i.e.,  scaling polynomially with the system size.
For a Hamiltonian $H = \sum_{i=0}^{d-1} \lambda_i
\ket{\psi_i}\bra{\psi_i}$ with a spectral gap $\Delta = \lambda_1 - \lambda_0>0$,
the Gibbs state $\rho_{\beta}$ can be very close to the ground state when $\beta \Delta \gg 1$.
However, the preparation of the Gibbs state requires imposing certain \textit{quantum
detailed balance condition}~\cite{KossakowskiFrigerioGoriniEtAl1977,CarlenMaas2017,ChenKastoryanoBrandaoEtAl2023,ChenKastoryanoGilyen2023,ding2024efficient},
a delicate balance of the rates between the heating and cooling transitions. 

 
The approach involving Gibbs state preparation presents another algorithmic challenge. While simulating the Lindblad dynamics up to time
$T$ may achieve a complexity of $\Or( T \log(T/\epsilon)/\log\log(T/\epsilon) )$~\cite{CleveWang2017,LW22} -- achieving near-optimal complexity in $T$ and $\epsilon$ -- these methods rely on yet another layer of intricate block encodings and precise control circuits, assuming fully fault-tolerant quantum computers. In the early fault-tolerant regime with constrained quantum resources, such as a limited number of logical qubits, a significant simplification would be essential.

Returning to our discussion on ground state preparation, it is important to recognize that the ground state represents the lowest-energy state where the detailed balance condition becomes singular (i.e., all heating processes are prohibited). As a result, it might not be necessary to utilize all the previously mentioned techniques to address general detailed balance conditions. This line of reasoning naturally leads us to our guiding question:

\textit{Can we devise a Monte Carlo-style quantum algorithm for ground states in the early fault-tolerant regime while maintaining near-optimal complexity in the Lindblad evolution time $T$ and precision $\epsilon$?}

\section{Main results}

In this work, we introduce a Lindbladian and develop algorithms that satisfy the following  features:
\begin{enumerate}[topsep=0pt,itemsep=-1ex,partopsep=1ex,parsep=1ex]
    \item[\textbullet] Correctness: The ground state is \emph{a} fixed point of a Lindblad evolution defined by \textit{a single jump operator}.
    \item[\textbullet] Efficient simulation: The continuous-time Lindblad evolution can be simulated using \textit{one ancilla qubit} and minimal control logic. With a discrete-time reformulation of the Lindblad evolution, the cost attains near-optimal complexity in both the evolution time and precision.
\end{enumerate}

In this work, the Lindblad dynamics takes the form
\begin{equation}\label{eqn:Lindblad_dynamics}
\frac{\mathrm{d}}{\mathrm{d} t}\rho = \mc{L}[\rho] = \underset{=:\mc{L}_H[\rho]}{\underbrace{-i [H, \rho]}}+ \underset{=:\mc{L}_K[\rho]}{\underbrace{K \rho K^{\dag}-\frac12 \{K^{\dag}K,\rho\}}}\,,
\end{equation}
with \textit{one} jump operator
%
\begin{align}
K&:=\sum_{i,j\in[N]} \hat{f}(\lambda_i-\lambda_j) \ket{\psi_i}\braket{\psi_i|A|\psi_j}\bra{\psi_j}\label{eqn:jump_freq}\\
&= \int_{-\infty}^{\infty} f(s) A(s) \ud s.
\label{eqn:jump_time}
\end{align}

Here $A$ is a Hermitian \textit{coupling operator} (or \textit{coupling matrix})\footnote{The Hermitian assumption is for simplicity of discussion, and the method presented in this paper can be generalized to non-Hermitian coupling operators.} that acts on the system with its Heisenberg evolution $A(s)=e^{i Hs} A e^{-iHs}$, and represents the interaction between the system and the environment. The time domain function $f(s):=\frac{1}{2\pi}\int_{\RR} \hat{f}(\omega)e^{-i\omega s}\ud \omega$ is the inverse Fourier transform of the filter function $\hat{f}$ in the frequency domain. The choice of $f$ is a key component in the algorithm and will be discussed in detail. We emphasize again that the choice of the jump operator $K$ is entirely algorithmic.

It is useful to compare the Lindblad-based method for ground state preparation with techniques that rely on QPE and subsequent methods. The core principle of QPE is ``post-selection'' within the energy domain, and the probability of success may diminish. For ground state preparation, the QPE algorithm assesses the energy state of the input state and only tries to retain components that overlap with the ground space. Consequently, if the initial overlap $p_0$ is $0$, a post-selection strategy cannot succeed. Conversely, Lindblad dynamics induces a completely positive trace-preserving (CPTP) map~\cite{GoriniKossakowskiSudarshan1976,Lindblad1976}. The jump operator $K$, taking the form of a \textit{linear combination of Heisenberg evolutions} of the coupling operator $A$, is designed to ``shovel'' high-energy components towards lower energies, culminating at the ground state (see \cref{fig:fc}). This ``shoveling'' process, being CPTP, ensures a success probability of $1$, circumventing the issues associated with post-selection. The critical factor now is the mixing time required for the Lindblad dynamics to converge to the ground state, which can vary with the system and the analysis of such strategies is only in its nascent stage. Ground state preparation is inherently linked to optimization problems. A recent work \cite{ChenHuangPreskillEtAl2023} demonstrates that even finding local minima in quantum systems under thermal perturbations can be computationally challenging for classical computers, while quantum computers using a Lindbladian formulation can solve the problem efficiently, offering a potential quantum advantage. 

In the following, we elaborate on the desirable features of the proposed algorithm.


\begin{figure}[t!]
\includegraphics[width=7cm]{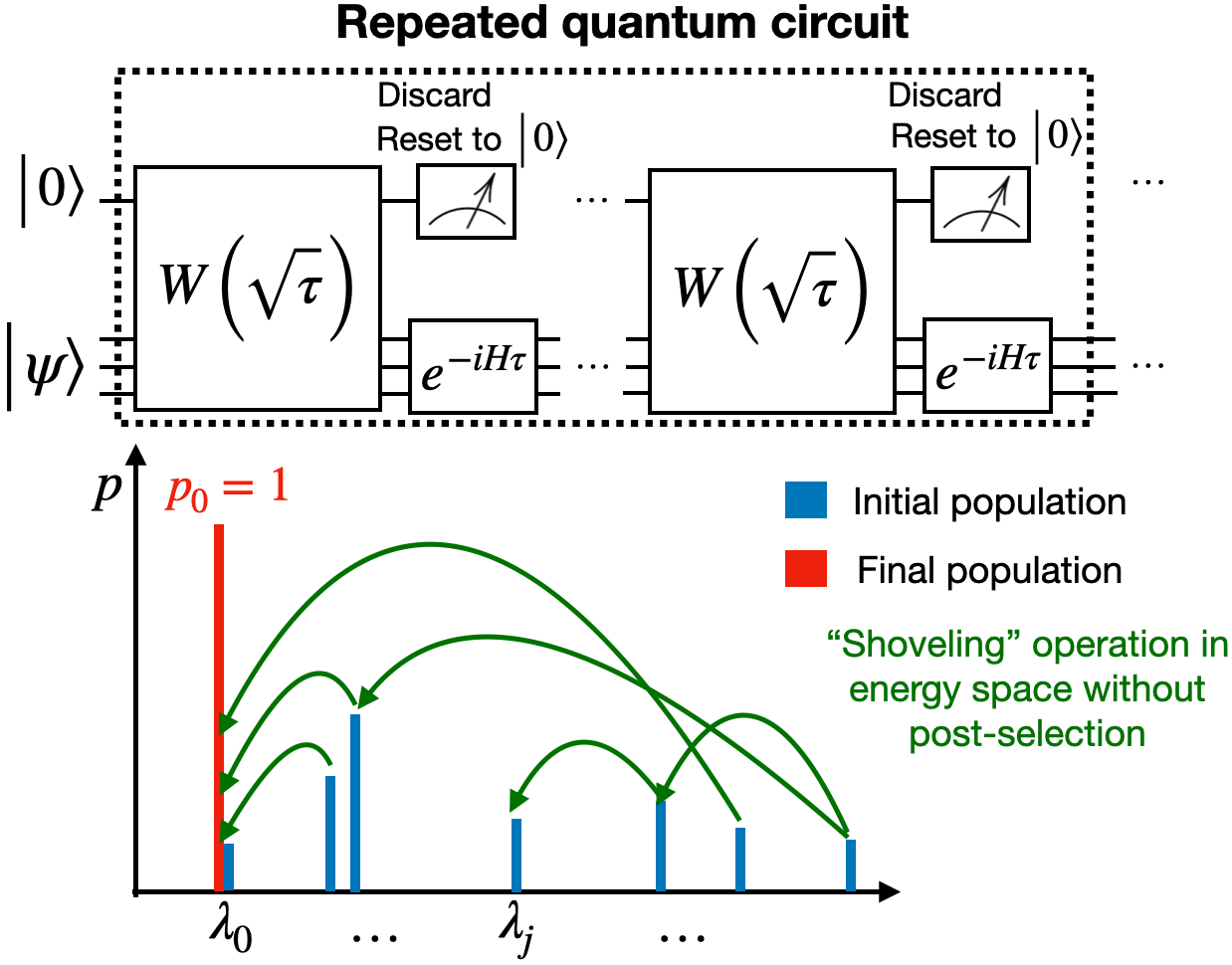}
\centering
\caption{(Top) Quantum circuit used for ground state preparation using Lindblad dynamics, using only one ancilla qubit.  This circuit structure of continuous and discrete time simulation of the Lindblad dynamics is very similar, and the main difference lies in the choice of the time step size. The measurement result of the ancilla qubit is discarded and the ancilla qubit is reset to $\ket{0}$ after each measurement. The implementation of $W(\sqrt{\tau})$ is shown later in \cref{fig:qc_1}.
(Bottom) a visual representation of how the Lindblad dynamics ``shovels'' high-energy components towards lower energies, culminating at the ground state. This process is implemented by a completely positive trace preserving (CPTP) map without post-selection.}
\label{fig:fc}
\end{figure}

\vspace{0.5em}
\subsection{Correctness}
The right-hand side of \eqref{eqn:Lindblad_dynamics} consists of two terms: The \textit{coherent part} features the Hamiltonian system $[H,\rho]$ that generates unitary dynamics; the \textit{dissipative part} $\mc{L}_K$ is parameterized by the jump operator $K$. Since the Hamiltonian trivially fixes its own ground state, it suffices to focus on the dissipative part; the fixed-point property remains valid with or without the coherent part\footnote{Nevertheless, our numerical findings suggest that incorporating the coherent part often leads to a significant reduction in the mixing time.}.  In the frequency domain, the jump operator $K$ is expressed as the coupling matrix $A$ in the energy basis weighted by the filter function $\hat{f}(\omega)$, which depends on the difference in energy $\omega = \lambda_i-\lambda_j$. This function should be real, nonnegative, and satisfy the following condition (detailed assumptions are given in \cref{assum:f_freq} in Appendix \ref{sec:filter}):
\begin{equation}\label{eqn:f_equal_to_zero}
\hat{f}(\omega)=0\quad \text{for any}\quad \omega\geq 0\,.
\end{equation}


Under \cref{eqn:f_equal_to_zero}, we obtain $\braket{\psi_i|K|\psi_j}=0$ when $\lambda_i\geq \lambda_j$. Intuitively, the jump operator $K$ forbids energy increments and favors a decrease in energy. 
Remarkably, such a simple condition guarantees that the ground state is an exact fixed point of the Lindbladian. To see this, we first notice that
\begin{equation}\label{eqn:K_fix_ground}
K\ket{\psi_0}=\sum_{i} \hat{f}(\lambda_i-\lambda_0) \ket{\psi_i}\braket{\psi_i|A|\psi_0}=0\,,
\end{equation}
which implies $\mc{L}_K[\ket{\psi_0}\bra{\psi_0}]=0$.
Together with $\mc{L}_H[\ket{\psi_0}\bra{\psi_0}]=0$, we  conclude that $\ket{\psi_0}\bra{\psi_0}$ is a fixed point of the Lindblad dynamics \eqref{eqn:Lindblad_dynamics}. That is,
\[
\quad e^{(\mc{L}_H+\mc{L}_K) t}[\ket{\psi_0}\bra{\psi_0}] =\ket{\psi_0}\bra{\psi_0}, \quad t \ge 0.
\]

\subsection{Simulating continuous-time Lindblad dynamics}

Quantum computers in the early fault-tolerant regime may have a limited number of logical qubits and relatively simple control logics. Unlike Hamiltonian simulation, where product formulas (or Trotterization) are well known to satisfy the above constraints, the case of Lindbladian simulation has received comparatively less attention. Existing Lindbladian simulation algorithms~\cite{CleveWang2017,CL17,ding2024simulating} often assume access to the block encoding of $K$. Since the jump operator $K$ is expressed as a linear combination of the Heisenberg evolution $A(s)$, encoding it into blocks requires the use of LCU and preparation oracles of $f$~\cite{ChenKastoryanoBrandaoEtAl2023,ChenKastoryanoGilyen2023}. This, in turn, entails a substantial number of ancillas. In this study, we present an efficient algorithm for simulating \eqref{eqn:Lindblad_dynamics} that uses a single ancilla qubit and simple controlled gates as follows:

\begin{thm}[Single ancilla simulation of  continuous-time Lindblad dynamics, informal]\label{thm:oneancilla_informal}
For any Hamiltonian $H$ and coupling operator $A$, there exists a quantum algorithm simulating the continuous Lindblad dynamics  \cref{eqn:Lindblad_dynamics} using one ancilla qubit. For simulation time $T$ and precision $\epsilon$, the total cost in terms of Hamiltonian simulation time is $ T_{H,\mathrm{total}}=\widetilde{\Theta}\left((1+\|H\|)\Delta^{-1}T^{2}\epsilon^{-1}\right)$. 
\end{thm}

The structure of our simulation circuit is sketched in \cref{fig:fc}, and the detailed description of the algorithm is in Section \ref{sec:overview_algorithm}. In the above results, the cost of simulating the coupling matrix $A$ is mild because $A$ is often much simpler to simulate than the (global) Hamiltonian $H$.
A more detailed version of the above theorem (\cref{thm:discretization_error_formal}), including the cost based on the number of controlled-$A$ gate queries, can be found in Appendix~\ref{sec:sc_cl}.

The most resource-intensive part of the algorithm is the Hamiltonian simulation. Thus, we quantify the cost through the total Hamiltonian simulation time, represented as $T_{H,\mathrm{total}}$. For an end-to-end cost analysis, one may further Trotterize the Hamiltonian simulation subroutine and analyze its discretization error, and there is no hidden block-encoding or extra ancilla involved (i.e., no controlled Hamiltonian simulation and only controlled-$A$, which is much easier to implement). 

The complexity of our algorithm approximately scales quadratically with the simulation time $T$ (a reasonable choice is the mixing time $T\sim \tmix$) and inversely with the error $\epsilon$. This scaling is similar to that of a first-order product formula, which has a second-order error $\exp((\mathcal{L}_H+\mathcal{L}_K)\tau)=\exp(\mathcal{L}_H\tau)\exp(\mathcal{L}_K\tau)+\mathcal{O}(\tau^2)$. We emphasize that the Lindbladian simulation qualitatively differs from the Hamiltonian simulation: Even arriving at the $\wt{\Or}(T^2/\epsilon)$ scaling with a single ancilla qubit requires a nontrivial argument to approximately implement $\exp(\mathcal{L}_K\tau)$. 

\vspace{0.5em}
\subsection{Simulating discrete-time Lindblad dynamics}\label{sec:discrete_main}

Is it possible to reduce the ``first-order'' scaling $T^2/\epsilon$ of the cost? Unlike Hamiltonian simulation, where higher-order formulas are available, implementing high-order Trotter splitting for Lindblad simulation presents challenges because dissipation is not generally reversible (i.e., simulating $\exp(-\mathcal{L}_K\tau)$) on a quantum computer. However, we can shift our perspective and reinterpret the scheme
\begin{equation}\label{eqn:discrete_informal}
\rho_{n+1}=\exp(\mathcal{L}_H\tau)\exp(\mathcal{L}_K\tau)\rho_n=\mathcal{N}_{\tau}(\rho_n)
\end{equation}
as a discrete-time dynamics with a time step $\tau=\mathcal{O}(1)$. Although \eqref{eqn:discrete_informal} does not necessarily approximate the continuous Lindbladian dynamics when $\tau=\mathcal{O}(1)$, the ground state is always a stationary state of the discretized dynamics for \textit{any} value of $\tau$. We can efficiently simulate the discrete dynamics using a quantum circuit similar to the one illustrated in Figure \ref{fig:fc} (refer to \cref{sec:acc_lindblad_dynamics} Figure \ref{fig:qc_2} for the discrete dynamics quantum circuit). In particular, setting $\tau=\mathcal{O}(1)$ ensures a significantly improved scaling with respect to both the simulation time $T$ and precision $\epsilon$. Remarkably, the total Hamiltonian simulation time of the discrete-time dynamics no longer linearly depends on the ratio $\norm{H}/\Delta$ as in \cref{thm:oneancilla_informal}, but only on $\Delta^{-1}$. This can be significantly faster when simulating large-scale systems or Hamiltonians involving unbounded operators (such as differential operators in first quantization).

\begin{thm}[Discrete-time algorithm for ground state, informal]\label{thm:ac_complex_informal}
For any Hamiltonian $H$ and coupling operator $A$, there exists a quantum algorithm simulating the discrete-time Lindblad dynamics  \cref{eqn:discrete_informal} using one ancilla qubit. For simulation time $T$ and precision $\epsilon$, the total Hamiltonian simulation time is  $T_{H,\mathrm{total}}=\widetilde{\Theta}(\Delta^{-1}T^{1+o(1)}\epsilon^{-o(1)})$.
\end{thm}

The detailed complexity result is shown in Appendix \ref{sec:sc_al} \cref{cor:acc_error_highorder}. Compared with the continuous-time Lindblad simulation, this algorithm exhibits a total Hamiltonian simulation time that scales nearly linearly with $T$ and sub-polynomially with $1/\epsilon$. In \cref{thm:ac_complex_informal}, the appearance of $o(1)$ stems from adopting a high-order Trotter formulation of order $p$ for the simulation of $\exp(\mathcal{L}_K\tau)$, which contributes an exponent of the form $1/p=o(1)$ as $p\to \infty$ (at the expense of a preconstant that scales exponentially in $p$).

The algorithm in \cref{thm:ac_complex_informal} approximates the discrete-time Lindblad dynamics \eqref{eqn:discrete_informal}. Theoretically, the ``effective'' mixing time $\tmix'=M_{\mathrm{mix}}\tau$ 
of the discrete dynamics may not be the same as the continuous-time Lindblad dynamics $\tmix$. On the other hand, numerical observations indicate that the mixing time of both dynamics can be comparable, but the discrete-time version requires significantly reduced total Hamiltonian simulation time, and hence the cost.

\section{Details of simulating  Lindblad dynamics}\label{sec:overview_algorithm}

The design of our specific Monte Carlo-style quantum algorithms draws inspiration from the recent work for the preparation of thermal states~\cite{ChenKastoryanoBrandaoEtAl2023,ChenKastoryanoGilyen2023}, which considers a set of jump operators labeled by $a,\omega$ as $\{\sqrt{\gamma(\omega)}\hat{A}^a(\omega)\}_{a,\omega}$, where $\hat{A}^a(\omega)=\int^\infty_{-\infty}f(s)A^a(s)\exp(-i\omega s)\mathrm{d}s$. Intuitively, analogous to the expansion in \eqref{eqn:jump_freq}, the set $\left\{\ket{\psi_i}\braket{\psi_i|A^a|\psi_j}\bra{\psi_j}\right\}_a$ corresponds to transitions of $A^a$ with an energy difference of approximately $\lambda_i-\lambda_j$. The transition rate is determined by  $f,\gamma$, which are temperature-dependent and are chosen to satisfy certain quantum detailed balance conditions. A key observation of this work is that for ground state preparation, such a condition can be considerably simplified to an \emph{energy reduction condition} as in \eqref{eqn:f_equal_to_zero}.

We describe the simulation of the Lindblad dynamics \eqref{eqn:Lindblad_dynamics} in three steps. Here, we mainly focus on the continuous-time dynamics simulation and leave the detailed discussion of discrete-time dynamics simulation in \cref{sec:acc_lindblad_dynamics}.

\subsection{Step 1: first-order Trotter splitting}\label{sec:first_order_trotter}

The Lindbladian consists of a coherent and a dissipative part. For simplicity, we use a first-order Trotter splitting 
\begin{equation}\label{eqn:first_trotter}
\begin{aligned}
   e^{\mc{L}t} &= e^{(\mc{L}_H +\mc{L}_K)t} \\
    & \approx (e^{\mc{L}_H\tau}e^{\mc{L}_K\tau} )^{t/\tau}\quad \text{for time step} \quad \tau.     
\end{aligned}
\end{equation}
Here, $e^{\mathcal{L}_H\tau}$ represents the Hamiltonian simulation $\exp(-iHt)$, which can be implemented directly. Consequently, following the Trotter splitting, our main focus is on simulating $e^{\mathcal{L}_K\tau}$ for a short time $\tau$. 

One might think that the first-order accuracy could be improved by employing an arbitrarily high-order Trotter splitting formula. However, this idea is not directly applicable due to two reasons: (1) using the standard high order Trotter splitting for \cref{eqn:first_trotter}, achieving an order higher than 2 requires the simulation of $\exp(-\mathcal{L}_K\tau)$ with $\tau>0$~\cite{BlanesCasas2005}. This operation is forbidden because $\exp(-\mathcal{L}_K\tau)$ ceases to be a physically realizable CPTP map.
(2) The first-order error encountered in \cref{eqn:first_trotter} is not the only reason for the first-order scaling in \cref{thm:oneancilla_informal}. Another source of first-order scaling emerges during the simulation of $e^{\mathcal{L}_K\tau}$ using a single ancilla qubit in \cref{lem:Lindblad_simulation_error}.  Due to these two reasons, extending the Trotter-type scheme to higher orders for \eqref{eqn:Lindblad_dynamics} presents a non-trivial challenge. Very recently, \cite{ding2024simulating} introduced a new approach for simulating the Lindblad dynamics using Hamiltonian simulation. This algorithm requires the block encoding of the jump operator $K$.
%

\subsection{Step 2: dilated Hamiltonian simulation problem}

To implement a \revise{non-unitary} dynamics $\exp(\mathcal{L}_K\tau)$ on a quantum computer, we define the \textit{dilated} Hermitian jump operator using one ancilla qubit as
\[
\wt{K}:=\begin{pmatrix}
0 & K^{\dag}\\
K & 0
\end{pmatrix}.
\]
Define the partial trace $\Tr_a\left(\sum^1_{i,j=0}\ket{i}\bra{j}\otimes \rho_{i,j}\right)=\sum^1_{i=0}\rho_{i,i}$, which traces out the ancilla qubit. Notice that
$e^{\mc{L}_{K} \tau}[\rho] \approx \Tr_a e^{-i \wt{K}\sqrt{\tau}}\left[\ket{0}\bra{0}\otimes\rho \right]e^{i \wt{K}\sqrt{\tau}}+\mathcal{O}(\tau^2)$, which reduces the Lindbladian simulation to the dilated Hamiltonian simulation. This is summarized in the following lemma:
\begin{lem}[Lindbladian simulation using one ancilla qubit]\label{lem:Lindblad_simulation_error}
Let 
\begin{equation}\label{eqn:trace_out}
    \sigma(t):=\Tr_a e^{-i \wt{K}\sqrt{\tau}} \left[\ket{0}\bra{0}\otimes\rho \right]e^{i \wt{K}\sqrt{\tau}}.
\end{equation}
Then for a short time $\tau\ge 0$, 
\[
    \|\sigma(\tau)-\exp(\mathcal{L}_K\tau)\rho\|_1=\mathcal{O}(\|K\|^4\tau^2)
\]
\end{lem}
The proof of \cref{lem:Lindblad_simulation_error} is in \cref{sec:lindblad_single}. According to \cref{lem:Lindblad_simulation_error}, we can approximate $e^{\mathcal{L}_K\tau}[\rho]$ by employing a dilated Hamiltonian simulation, tracing out the ancilla qubit, and resetting the ancilla qubit to the state $\ket{0}$. We note that Lemma \ref{lem:Lindblad_simulation_error}  introduces the first-order accuracy scaling for continuous-time dynamics simulation because the error at each simulation step is $\mathcal{O}(\tau^2)$.

In this work, we employ a measure-and-discard process to execute the trace-out step in \eqref{eqn:trace_out} using a single ancilla qubit, see \cref{fig:fc} for the circuit. There is another implementation method known as repeated interaction~\cite{pocrnic2023quantum,PhysRevX.7.021003,PhysRevLett.123.140601}, in which a new ancilla qubit is introduced in each iteration and measurements are made to trace out all ancilla qubits at the end of the simulation. We expect that the same total Hamiltonian simulation time of the two procedures are comparable to each other.

\begin{figure*}
\centering
\subfloat
{
\begin{quantikz}[column sep=0.3cm]
\lstick{$\ket{0}$} & \gate[2]{\wt{A}_{-M_s}(\sqrt{\tau})} &\qw& \gate[2]{\wt{A}_{-M_s+1}(\sqrt{\tau})}&\qw \raisebox{0em}{$\cdots$}&\gate[2]{\wt{A}_{M_s-1}(\sqrt{\tau})}&\qw &\gate[2]{\wt{A}_{M_s}(\sqrt{\tau})}& \qw\\
\lstick[2]{$\ket{\psi}$} &\qw & \gate[2]{e^{-iH\tau_s}} &\qw&\qw \raisebox{0em}{$\cdots$}&\qw&\gate[2]{e^{iH\tau_s}}&&\qw\\
&\qwbundle[alternate]{}&\qwbundle[alternate]{}&\qwbundle[alternate]{}&\qwbundle[alternate]{}\raisebox{0em}{$\cdots$}&\qwbundle[alternate]{}&\qwbundle[alternate]{}&\qwbundle[alternate]{}&\qwbundle[alternate]{}&\qwbundle[alternate]{}
\end{quantikz}
}
\caption{Detailed quantum circuits for $W(\sqrt{\tau})$. Here we assume $A$ is a local operator that only acts on one system qubit (or a small number of system qubits).}
\label{fig:qc_1}
\end{figure*}
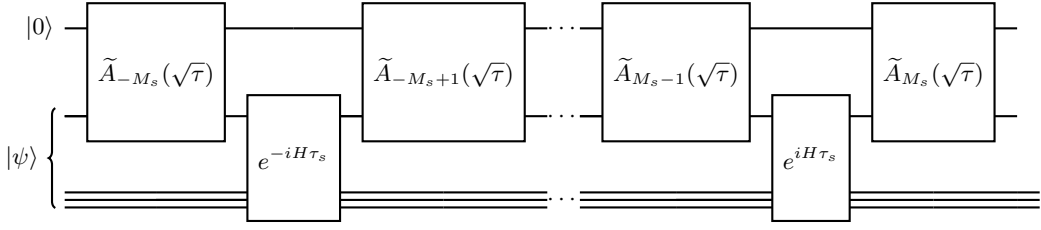

\subsection{Step 3: Simulation of \texorpdfstring{$\exp\left(-i\widetilde{K}\sqrt{\tau}\right)$}{Lg}}\label{sec:Step_3}

To implement the dilated Hamiltonian simulation $e^{-i \wt{K}\sqrt{\tau}}$, we observe that $K$  is expressed in an integral form. The integrand involves Heisenberg evolution, which requires the Hamiltonian simulation $e^{-iHs}$. The situation here is similar to Hamiltonian simulation using two very different types of algorithms: linear combinations of unitary (LCU)~\cite{ChildsWiebe2012} / quantum singular value transformation (QSVT)~\cite{GilyenSuLowEtAl2019}, and product formulas~\cite{Trotter1959,ChildsSuTranEtAl2021}. The former is theoretically transparent and can lead to asymptotically optimal complexity, but practitioners often turn to the latter because of its simple implementation, especially with limited quantum resources.



Our approach is motivated by the following observation. We consider the time-dependent problem
\[
\partial_s u(s)=\tau f(s) A(s) u(s)
\]
and let $\widetilde{u}(s)=\exp(-iHs)u(s)$. Then, we have
\[
\partial_s\widetilde{u}(s)=(\tau f(s)A-H)\widetilde{u}(s)
\]
and $u(S)=\exp(iHS)\widetilde{u}(S)$. Through this transformation we no longer need the Heisenberg evolution of the coupling operator $A$. This means when $\tau$ is sufficiently small,
\begin{equation}
\begin{split}
&e^{i \tau\int_{0}^S f(s) A(s) \ud s}
= \mc{T} e^{i \int_0^S \tau f(s) A(s) \ud s}+\mathcal{O}(\tau^2)\\
=& e^{i H S} \mc{T} e^{i  \int_0^S  (\tau f(s) A-H)\ud s}+\mathcal{O}(\tau^2).\label{eq:time-order}
\end{split}
\end{equation}
where $\mc{T}$ denotes time-ordering for the exponential. 
Here we recognize the time-ordered exponential as the interaction picture over Hamiltonian $H$, which leads to the second expression. Then, we can further discretize and Trotterize it into simple gates. 

For a technical reason that will become clear below, our actual implementation utilizes a second-order Trotter formula, rather than the first-order formula in~\eqref{eq:time-order}. This ensures that the local truncation error of simulating $\exp\left(-i\wt{K}\sqrt{\tau}\right)$ is of the order $\mathcal{O}(\tau^2)$ rather than $\mathcal{O}(\tau)$. 

Now, we outline the implementation roadmap of $\exp\left(-i\widetilde{K}\sqrt{\tau}\right)$, reserving all technical proofs for Appendix \ref{sec:analysis}.

\textbf{Discretizing the time integral.}

As discussed before, to implement $ \exp \left(-i \widetilde{K}\sqrt{\tau}\right)$, we need to truncate and discretize the integral $\int_{-\infty}^\infty f(s) A(s) ds$. 

For the discretization, choose a suitable $S_s>0$, restrict the integration range to $[-S_s,S_s]$, and discretize this interval using a uniform grid $s_l=l\tau_s$ and $l=-M_s,\ldots,M_s$, where $\tau_s=S_s/M_s$. We can then approximate the integral using a trapezoidal rule:
\begin{equation}\label{eqn:K_s}
K\approx K_{s}:=  \sum_{l=-M_s}^{M_s} f(s_l)  e^{iH s_l}Ae^{-iH s_l} w_l,
\end{equation}
where $w_l=\tau_s/2$ for $l=\pm M_s$, and $w_l=\tau_s$ for $-M_s<l<M_s$.  

Assuming sufficiently smooth and compactly supported $\hat{f}(\omega)$, we can show that $f(s)$ decays super polynomially as $\abs{s}\to \infty$. More specifically, in our analysis, we assume that $\hat{f}$ can be expressed as a product of Gevrey functions (\cref{assum:f_freq}). This guarantees rapid decay of $f(s)$ in the time domain (\cref{lem:as_f_simulation}), and allows us to effectively control the quadrature discretization error, which is summarized in the following lemma:
\begin{lem}[Convergence of the quadrature error, informal]\label{lem:trapezoidal_error}
If $S_s=\Omega\left(\frac{1}{\Delta}\mathrm{polylog}\left(\frac{1}{\epsilon'}\right)\right)$, $ \tau_s=\mathcal{O}\left(\frac{1}{\|H\|}\right)$, then
\[
    \|K-K_{s}\|=\mathcal{O}(\epsilon)\,.
\]
\end{lem}
We present the formal version of the above lemma in \cref{sec:pf_d_error} Lemma \ref{lem:trapezoidal_error_appendix}. In practical applications, it is sufficient to select a smooth $\hat{f}$ such that it is approximately supported for $\omega \leq 0$. A specific example is provided in~\cref{eqn:f_omega} and is utilized throughout our experiments in~\cref{sec:numerics}.

According to the above lemma, a small number of points $M_s$ already ensures a good approximation
\begin{equation}
\wt{K}\approx \wt{K}_s=
\begin{pmatrix}
0 & K_s^{\dag}\\
K_s & 0
\end{pmatrix}\\
=:\sum_{l=-M_s}^{M_s} \wt{H}_l.
\end{equation}
Since the Heisenberg evolution $A(s)=e^{iHs}Ae^{-iHs}$ is Hermitian, the term further factorizes as 
\begin{equation}
\wt{H}_l=\begin{pmatrix}
0 & f^*(s_l) A(s_l) w_l\\
 f(s_l) A(s_l) w_l& 0
\end{pmatrix}=\sigma_l\otimes A(s_l)
\end{equation}
where $\sigma_l:=w_l(\sigma_x \Re f(s_l) + \sigma_y\Im f(s_{l}))$ with Pauli matrices $\sigma_x$ and $\sigma_y$.

\textbf{Second-order Trotter splitting for $\exp\left(-i\widetilde{K}\sqrt{\tau}\right)$.}

After discretizing the time labels, the next step is to Trotterize the Hamiltonian evolution
\begin{align}
   e^{-i\sqrt{\tau} \wt{K}_s} = e^{-i\sqrt{\tau} \sum_{l} \wt{H}_{l}}.
\end{align}
Specifically, we employ the second-order Trotter formula to balance between efficiency and accuracy. 
The second-order Trotter formula for $\exp(-i\sqrt{\tau} \widetilde{K})$ can be expressed as:
\begin{equation}
\begin{aligned}
&\rightprod_{l} e^{-i \frac{\sqrt{\tau}}{2} \wt{H}_l}\leftprod_{l} e^{-i \frac{\sqrt{\tau}}{2} \wt{H}_l}-e^{-i\sqrt{\tau} \wt{K}_s}\\
=&\tau^{3/2}\sum_{l_1,l_2,l_3}a_{l_1,l_2,l_3}\wt{H}_{l_1}\wt{H}_{l_1}\wt{H}_{l_3}+\Or(\tau^2),    
\end{aligned}
\label{eq:eta3}
\end{equation}
where the coefficients $a_{l_1,l_2,l_3}$ can be calculated from Taylor expansion, and the left and right-ordered products are defined by 
\begin{equation}\label{eqn:secondordertaylor_collect}
\begin{aligned}
&\leftprod_{l} e^{-i \sqrt{\tau} \wt{H}_{l}}:=e^{-i \sqrt{\tau} \wt{H}_{M_s}} \cdots e^{-i \sqrt{\tau} \wt{H}_{-M_s}}, \quad \text{and} \\
&\rightprod_{l} e^{-i \sqrt{\tau} \wt{H}_l}:=e^{-i \sqrt{\tau} \wt{H}_{-M_s}} \cdots e^{-i \sqrt{\tau} \wt{H}_{M_s}}.    
\end{aligned}
\end{equation}

The Trotter error bounds~\eqref{eq:eta3} and Lemma \ref{lem:Lindblad_simulation_error} together give an approximation scheme for simulating the Lindblad dynamics in \eqref{eqn:Lindblad_dynamics}. For any initial state $\rho$, we have that
\[
\begin{aligned}
&\mathrm{Tr}_a\left(\rightprod_{l} e^{-i \frac{\sqrt{\tau}}{2} \wt{H}_l}\leftprod_{l} e^{-i \frac{\sqrt{\tau}}{2} \wt{H}_l} \ket{0}\bra{0}\otimes\rho \right.\\
&\left.\quad\quad\quad \times \rightprod_{l} e^{i \frac{\sqrt{\tau}}{2} \wt{H}_l} \leftprod_{l} e^{i \frac{\sqrt{\tau}}{2} \wt{H}_l}\right)\\
=&\mathrm{Tr}_a\left(e^{-i\sqrt{\tau} \wt{K}_s} [\ket{0}\bra{0}\otimes\rho ]e^{i\sqrt{\tau} \wt{K}_s}\right)+\Or(\tau^2)\\
=&e^{\mathcal{L}_{K_s}\tau}[\rho]+ \Or(\tau^2)\approx e^{\mathcal{L}_{K}\tau}[\rho]+ \Or(\tau^2)\,,
\end{aligned}
\]
Here, the second equality is obtained by using $\mathrm{Tr}_a\left(\wt{H}_{l_1}\wt{H}_{l_2}\wt{H}_{l_3}\ket{0}\bra{0}\otimes\rho\right)=0$ and $\mathrm{Tr}_a\left(\ket{0}\bra{0}\otimes\rho \wt{H}^\dagger_{l_3}\wt{H}^\dagger_{l_2}\wt{H}^\dagger_{l_1}\right)=0$. In the last equality, we use \cref{lem:Lindblad_simulation_error}.

From the above derivation, we find that if we were to replace the second-order formula with a first-order formula, the local truncation error becomes $\Or(\tau)$, and the global error becomes $\Or(1)$. While higher-order Trotter formulas can further suppress the error in simulating $\exp(-i\widetilde{K}\tau)$, the accuracy of the Lindbladian simulation is constrained by \cref{lem:Lindblad_simulation_error}.
Therefore, the second-order Trotter is adequate for the purposes here.
 

\textbf{Canceling out back-and-forth Hamiltonian evolution.}

Finally, to efficiently implement the products $e^{-i\frac{\sqrt{\tau}}{2} \widetilde{H}_l}$, we notice
\begin{align*}
&\exp\left(-i \frac{\sqrt{\tau}}{2} \sigma_l\otimes A(s_l)
\right)\\
= &(I\otimes e^{iH s_{l}})\underset{=:\wt{A}_l(\sqrt{\tau})}{\underbrace{e^{-i \frac{\sqrt{\tau}}{2} \sigma_l\otimes A }}}(I\otimes e^{-iH s_{l}}).
\end{align*}
Since $A$ is a simple operator (e.g., a local Pauli), the cost of $\wt{A}_l(\sqrt{\tau})$ is mostly negligible compared to that of the simulation of the system Hamiltonian. 

What about the $e^{iH s_{l}}$ terms? A moment of thought reveals that we may rewrite the consecutive product in a form that efficiently cancels out the back-and-forth Hamiltonian evolution. We present the most abstract form to emphasize the simplicity of this observation.
\begin{prop}[Cancellations in time-order products]
Consider a time-order product at discretized times $s_l= l \tau_s $ for $l = - M_s,\cdots, M_s$. Then, for any Hamiltonian $H$ and a set of matrices $A_l$ depending on $l$,
\begin{align*}
    \rightprod_{l} A_{l}(s_{l}) &= e^{-iH S_s} \left(\rightprod_{l} A_l e^{iH \tau_s}\right) e^{iH (S_s+\tau_s)}\\
    \leftprod_{l} A_{l}(s_{l}) &= e^{-iH (S_s+\tau_s)} \left(\leftprod_{l} e^{-iH \tau_s}A_l \right) e^{iH S_s}
\end{align*}     
where $A_l(s) = e^{i H s} A_l e^{-i Hs}$ denotes the Heisenberg evolution of $A_l$ with $H$.
\end{prop}
\begin{proof}
Because $e^{-iH s_l}e^{iH s_{l+1}}=e^{iH \tau_s}$, the Hamiltonian evolution from two consecutive steps nearly cancels (for the right-ordered product for example), meaning $A_{l}(s_{l})A_{l}(s_{l+1}) = e^{i H s_l} A_{l} e^{i H \tau_s} A_{l} e^{-i H s_{l+1}}$.
\end{proof}

For our usage, our second-order formula consists of both the left and right products, which is
\[
\begin{split}
&\rightprod_{l}   (I\otimes e^{iH s_l}) \wt{A}_l(\sqrt{\tau}) (I\otimes e^{-iH s_l})\\
&\quad\quad\times \leftprod_{l} (I\otimes e^{iH s_l}) \wt{A}_l(\sqrt{\tau})  (I\otimes e^{-iH s_l})\\
=&  (I\otimes e^{-iH S_s})\\
&\times \underset{=:W(\sqrt{\tau})}{\underbrace{\left(\rightprod_{l} \wt{A}_l(\sqrt{\tau}) (I\otimes e^{iH \tau_s})\right)
\left(\leftprod_{l}  (I\otimes e^{-iH \tau_s})\wt{A}_l(\sqrt{\tau})\right)}}\\
&\times (I\otimes e^{iH S_s}).
\end{split}
\]
where $W(\sqrt{\tau})$ is a product of \textit{short-time} Hamiltonian simulation.

Even nicer, the long-time simulation $(I\otimes e^{-iH S_s})$ from the previous time step exactly cancels with 
$(I\otimes e^{iH S_s})$ from the subsequent time step. Therefore, we may remove both long-time evolution steps and define 
the quantum channel 
\begin{equation}
\mc{W}(\tau)[\rho]:=\Tr_a \left(W(\sqrt{\tau}) \left[\ket{0}\bra{0}\otimes\rho\right] W^{\dag}(\sqrt{\tau})\right).
\end{equation}
The total simulation time of the system Hamiltonian for implementing the quantum channel $\mc{W}(\tau)$ now becomes 
\begin{equation}\label{eqn:total_H_1}
\sum_{l=-M_s}^{M_s} \tau_s= \Or(S_s)\,.
\end{equation}
In summary, the single ancilla simulation of the Lindblad dynamics takes the form
\begin{equation}
\rho_{m+1}=\mc{W}(\tau)[\rho_m]\quad\text{(purely dissipative)}\label{eqn:single_lindblad_shorttime}
\end{equation}
or
\begin{equation}
\rho_{m+1}=e^{\mathcal{L}_H\tau}\mc{W}(\tau)[\rho_m]\ \text{(Trotterizing the coherent part)} \label{eqn:single_lindblad_shorttime_2}
\end{equation}
The second line displays the option to include the coherence term via Trotter, which may introduce additional errors~\footnote{The dissipative part itself already fixes the ground state. However, numerically, without this coherent term, we observe that the algorithm can be stuck at other fixed points.}. The quantum circuit for $W(\sqrt{\tau})$ is drawn in Figure \ref{fig:qc_1}.

Finally, let $T= M_t \tau$ and allow $\tau$ to approach zero. We expect that $\rho_{M_t}$ from \eqref{eqn:single_lindblad_shorttime_2} converges to $\exp(-\mathcal{L}_H S_s)[\rho(T)]=e^{iHS_s}\rho(T)e^{-iHS_s}$, where $\rho(t)$ is the solution of the modified Lindblad dynamics
\begin{equation}\label{eqn:modified_lindblad}
\partial_t \rho(t)=\mc{L}_{H}[\rho(t)]+\mc{L}_{K}[\rho(t)]\,,
\end{equation}
where $\rho(0)=\exp(\mathcal{L}_H S_s)[\rho_I]$. 
Compared to the original Lindblad dynamics in \cref{eqn:Lindblad_dynamics}, the initial state is changed to $\exp(\mathcal{L}_H S_s)[\rho_I]=e^{-iHS_s}\rho_I e^{iHS_s}$. In particular, if $\rho_I$ commutes with $H$ (e.g., $\rho_I$ is the density operator corresponding to an eigenstate of $H$, or the maximally mixed state), then $\exp(\mathcal{L}_H S_s)[\rho_I]=\rho_I$. At the end of the simulation, if $\rho(T)$ is the ground state $\rho_g$, then   $\exp(-\mathcal{L}_H S_s)[\rho_g]=\rho_g$. We, therefore, expect the behavior of the modified Lindblad dynamics to be very similar to that of the original dynamics.

\section{Ergodicity and mixing time with random coupling matrix}\label{sec:convergence}

Although the ground state is a fixed point of the Lindblad dynamics~\eqref{eqn:Lindblad_dynamics}, it is possible for other fixed points to exist, and the map may not exhibit ergodic behavior. In general, the selection of the appropriate operator $A$ to guarantee ergodicity tends to be case-dependent. There is substantial literature~\cite{Evans_1977, Nigro_2019, Frig_1978, SPOHN1976189,Zhang:2023ayz} dedicated to investigating the irreducibility and ergodicity of the Lindblad dynamics. Our Lindblad dynamics involves a single jump operator, and the fixed point is a rank-one density matrix. As a result, the sufficient conditions derived in the literature are in general not applicable.

In our paper, we examine this problem in a simplified scenario, where we assume  $A$ is drawn from a distribution of random matrices with independent entries in the energy eigenbasis of $H$. More specifically, we define $A_{i,j}=\left\langle \psi_i\right|A\ket{\psi_j}$, $\hat{f}_{i,j}=\hat{f}(\lambda_i-\lambda_j)$, and make the following assumption:
\begin{assumption}\label{assump:A}
\begin{itemize}
    \item (Random matrix elements) Assume for any $t\geq0$, $A$ is independently drawn from random probability distribution $\Xi_A$ on the set of Hermitian matrices\footnote{Strictly speaking, in the above assumption, $A$ should be a function of time and denoted as $A_t$. However, for the sake of simplicity and consistency with other notations, we will omit the subscript $t$.} such that $A_{i,j}$ are independent and $\mathbb{E}(A_{i,j})=0$ when $i\neq j$. Denote $\sigma_{i,j}=\mathbb{E}(|A_{i,j}|^2)>0$. 

    \item (Support of filter) \revise{$[\lambda_0-\lambda_{N-1},\lambda_0-\lambda_1]\subset \mathrm{Supp}(f)$.}
    
    \item (Diagonal initial state) $\rho(0)$ is a diagonal matrix in the basis of $\{\ket{\psi_i}\}^{N-1}_{i=0}$.
\end{itemize}
\end{assumption}

Under \cref{assump:A}, we give a partial argument for ergodicity. 
\begin{thm}[Random coupling matrix and ergodicity, informal]\label{thm:fixed_point} Let $\rho(t)$ be the solution to the Lindblad dynamics \eqref{eqn:Lindblad_dynamics}. Under \cref{assump:A}, $\rho^\star=\ket{\psi_0}\bra{\psi_0}$ is the \emph{unique} fixed point of the Lindblad dynamics in the expectation sense. In particular, given any observable $O$, $\lim_{t\rightarrow\infty} \mathbb{E}[\mathrm{Tr}(O\rho(t))]=\bra{\psi_0}O\ket{\psi_0}$, where the expectation is taken on the randomness of $A$. 
\end{thm}

We put the proof of Theorem \ref{thm:fixed_point} in Appendix \ref{sec:pf_conv_lindblad}. To our knowledge, this is the first uniqueness argument for ground state preparation using a Lindblad dynamics. It should also be noted that, strictly speaking, the expected operator $\mathbb{E}(\rho(t))$ is not the density operator $\rho(t)$ that we store in the quantum memory but still gives some optimistic intuition about ergodicity (\cref{sec:random_coupling_conv}).
Technically, taking expectations over \textit{independent} entries of $A$ substantially simplifies the transition matrix. This independent assumption can also be seen as a version of the Eigenstate Thermalization Hypothesis (ETH) \cite{MS_1999,LY_2016}, which incorporates additional assumptions on the variance of $A_{i,j}$. In fact, ETH has been employed to explain finite-time thermalization in chaotic open quantum systems \cite{shtanko2023preparing,chen2023fast}. Under stronger ETH-type assumptions, one may be able to prove the convergence for $\rho(t)$ instead of $\mathbb{E}(\rho(t))$ as in~\cite{chen2023fast}, but we merely focus on the much simpler object $\mathbb{E}(\rho(t))$ without distracting from the presentation of the algorithm.

In addition to ergodicity, another crucial convergence criterion for a Markov chain process is the \textit{mixing time}, which is defined as follows:
\begin{defn}[Mixing time of Lindbladians]
\label{def:mixing_time}
We define the mixing time of any Lindbladian $\mc{L}$ as 
\begin{align*}
\norm{e^{\mc{L} \tmix }[\rho-\rho']}_{1} \le \frac{1}{2} \norm{\rho-\rho'}_{1} \quad \text{for any states}\quad \rho, \rho'.
\end{align*}
\end{defn}
Intuitively, the mixing time describes the time scale at which any two input states become close to each other. 
In particular, if the Lindblad $\mc{L}$ has mixing time $\tmix$, any initial state must be $\epsilon$ close to the ground state after time $\tmix \cdot \log_2(2/\epsilon)$ since the trace distance between any two quantum states is at most $2$. 

Analogous to classical Monte Carlo sampling, we do not know \textit{a priori} the mixing time
associated with the jump operator $K$ and the Hamiltonian $H$; we expect this to be system-dependent.~\revise{In fact, proving the  mixing time for Lindbladians
can be a highly challenging problem~\cite{KastoryanoBrandao2016,capel2020modified}, and is not the primary objective of our work. As an example towards understanding the convergence behavior, we show that under additional assumptions on the coupling matrix $A$ and the eigenvalue distributions of $H$, the Lindblad dynamics discussed in our study can achieve polynomial mixing time, see \cref{sec:conv} \cref{thm:converge_Lindblad_K} for detail. To ensure fast mixing in practical applications, we might also select multiple coupling operators based on the structure of the Hamiltonian $H$. Specifically, if prior knowledge of the eigenbasis of the Hamiltonian is available, we should choose a set of $\{A^a\}$ such that each $A^a$ satisfies $A^a_{i,j}= \left\langle \psi_i \right| A^a \left| \psi_j \right\rangle \neq 0$ for a large set of $(i,j)$ pairs. This ensures sufficient transitions between different energy states, which is likely to achieve rapid mixing. Generically, a reasonable choice might be the set of local Pauli matrices, which are known to be good choices for thermalizing commuting Hamiltonians in~\cite{KastoryanoBrandao2016,BardetCapelGaoEtAl2023,gilyrn2024quantum,rouz2024}.}

\section{Numerical results}\label{sec:numerics}

In our numerical simulations, we test the efficiency of our algorithm on three different Hamiltonians: the transverse field Ising model (TFIM) with $4$ sites (TFIM-4), TFIM-6, and the Hubbard model.

For simplicity, we choose a filter function $\hat{f}(\omega)$ with the following analytic form in the frequency space:
\begin{equation}\label{eqn:f_omega}
\hat{f}(\omega) := \frac{1}{2}\left(\text{erf}\left(\frac{\omega+a}{\delta_a}\right)-\text{erf}\left(\frac{\omega+b}{\delta_b}\right)\right),
\end{equation}
where $\text{erf}(\omega) = \frac{2}{\sqrt{\pi}}\int_0^\omega e^{-x^2} \mathrm{d}x$ denotes the error function. The parameters $a$ and $\delta_a$ are chosen to be of the order $S_\omega$, while $b$ and $\delta_b$ are of the order $\Delta$. The inverse Fourier transform $f(s) = \int_{\mathbb{R}} \hat{f}(\omega)e^{-i\omega s}\mathrm{d}\omega$ is given by
\begin{equation}\label{eqn:F}
f(s) = \frac{e^{-\frac{\delta_a^2 s^2}{4}}e^{ias}-e^{-\frac{\delta_b^2 s^2}{4}}e^{ibs}}{2\pi is}.
\end{equation}
It is worth noting that $\lim_{t\to 0} f(s) = \frac{a-b}{2\pi}$ is well-defined, and $f(s)$ is a smooth complex function that is approximately supported on the interval $[-S_s, S_s]$, where $S_s = \Theta(\delta_b^{-1})$. The shape of $\abs{f(s)}$ and $\hat{f}(\omega)$ (for TFIM-6) is shown in Figure \ref{fig:f}, and $\hat{f}(\omega)$ is approximately supported in $[-2a,0]$. 
Although \eqref{eqn:f_omega} is not strictly compactly supported, we can multiply it by a compactly supported ``bump'' function to satisfy \cref{assum:f_freq}.  

For simulating the continuous/discrete-time Lindblad dynamics, we always choose $a=2.5\|H\|$, $\delta_a=0.5\|H\|$, $b=\delta_b=\Delta$, $S_s=5/\delta_b$, and $\tau_s=\pi/(2a)$. The code used for generating the numerical results is available on Github (\url{https://github.com/lin-lin/oneancillaground}).

\begin{figure*}[!htbp]
\centering
  \subfloat[Frequency domain]{\includegraphics[width=7cm]{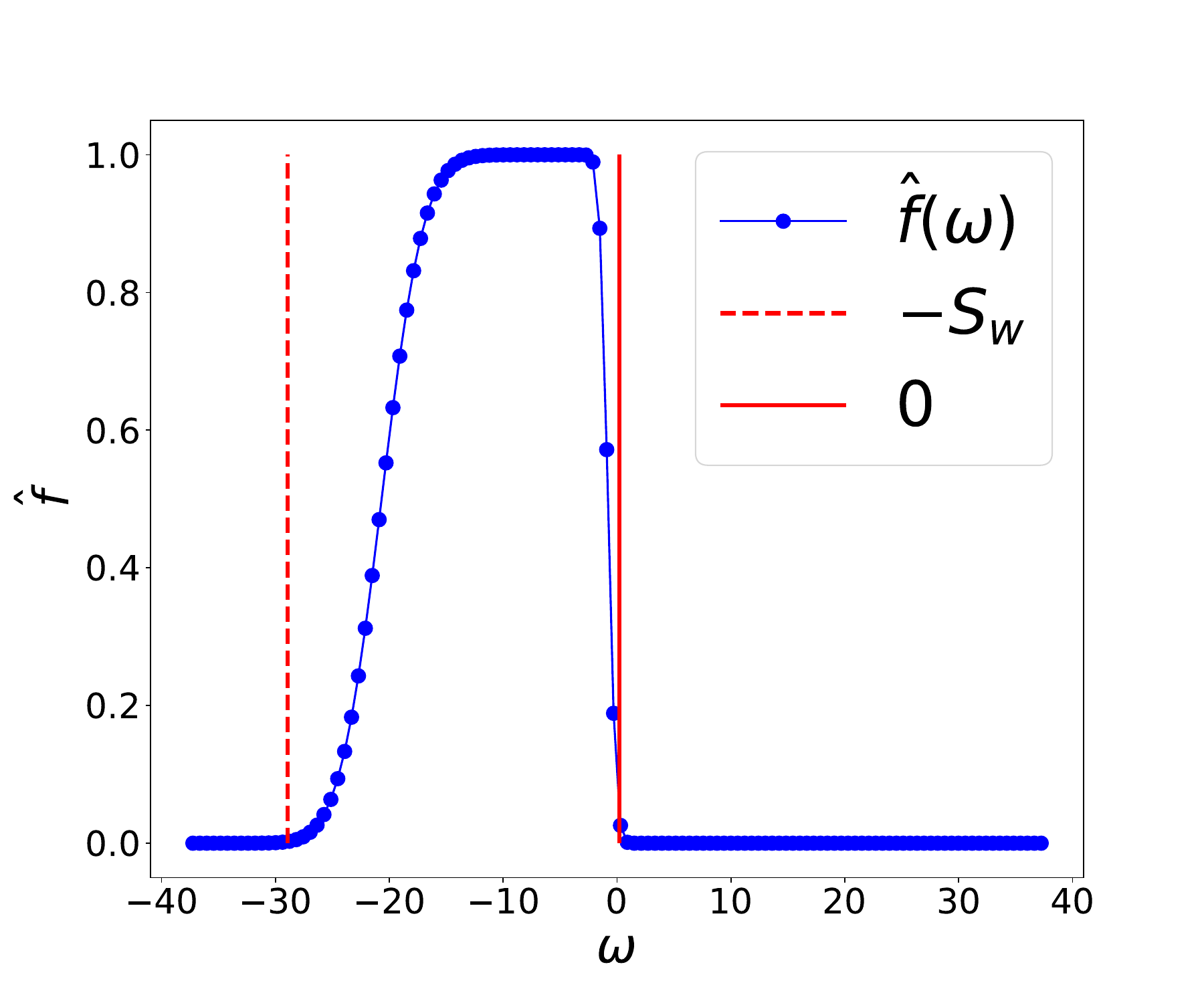}\label{fig2}}
  \subfloat[Time domain]{\includegraphics[width=7cm]{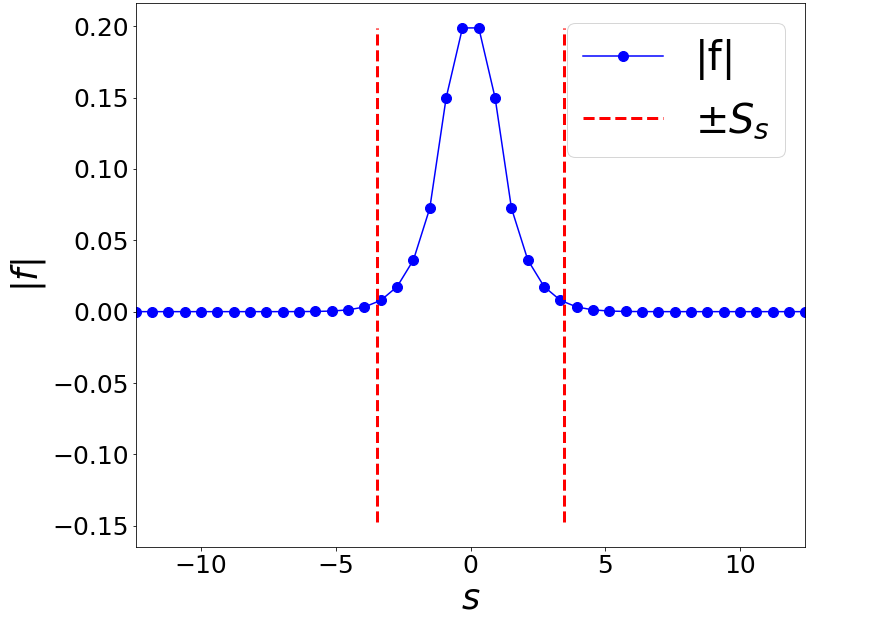}\label{fig1}}
   \caption{Illustration of $\hat{f}(\omega)$ and the absolute value of its Fourier transform $\abs{f(s)}$ following \cref{eqn:f_omega} and \cref{eqn:F} used for the TFIM-6 model.
   }\label{fig:f}
\end{figure*}

\subsection{TFIM-4 model}\label{sec:Ising}
The TFIM Hamiltonian with $L$ sites reads:
\begin{equation}\label{eqn:H_Ising}
H=-\left(\sum^{L-1}_{i=1} Z_{i}Z_{i+1}\right) -g\sum^L_{i=1} X_i\,,
\end{equation}
where $g$ is the coupling coefficient, $Z_i,X_i$ are Pauli operators for the $i$-th site and the dimension of $H$ is $2^L$. We set $L=4$ and the coupling constant $g=1.2$. We choose the local Hermitian operator $A=Z_1=Z\otimes I^{\otimes L-1}$. Here $I^{\otimes L-1}$ is the identity operator acting on qubits $2$ to $L$ and its dimension is $2^{L-1}$.

In our numerical simulations, we set $\tau=1$ and $r=1$ for discrete-time Lindblad simulation ($r$ is the number of segments that we use to approximately simulate $\exp(\mathcal{L}_K\tau)$ in the discrete-time Lindblad simulation algorithm, see \cref{sec:acc_lindblad_dynamics}),
while for continuous-time Lindblad simulation, we use $\tau=0.1$. The stopping times are set to $T=80$ for TFIM-4. We start with an initial state with \textit{zero overlap} ($\bra{\psi_0}\rho_0\ket{\psi_0}\approx10^{-17}$), repeat each Lindblad dynamics simulation 100 times, and compute the average energy and overlap with the ground state. The results are depicted in Figure \ref{fig:TFIM_4}. Our observations indicate that both Lindblad dynamics exhibit efficient convergence to the ground state starting from the initial state with zero overlaps. Moreover, the discrete-time Lindblad dynamics (with $\tau=1$ and $r=1$) has a comparable rate of mixing as the continuous dynamics (with $\tau=0.1$) and can reduce the total Hamiltonian simulation time by one order of magnitude.

\begin{figure*}[!htbp]
\centering
  \subfloat[Lindblad simulation time vs energy]{\includegraphics[width=6cm,height=4.5cm]{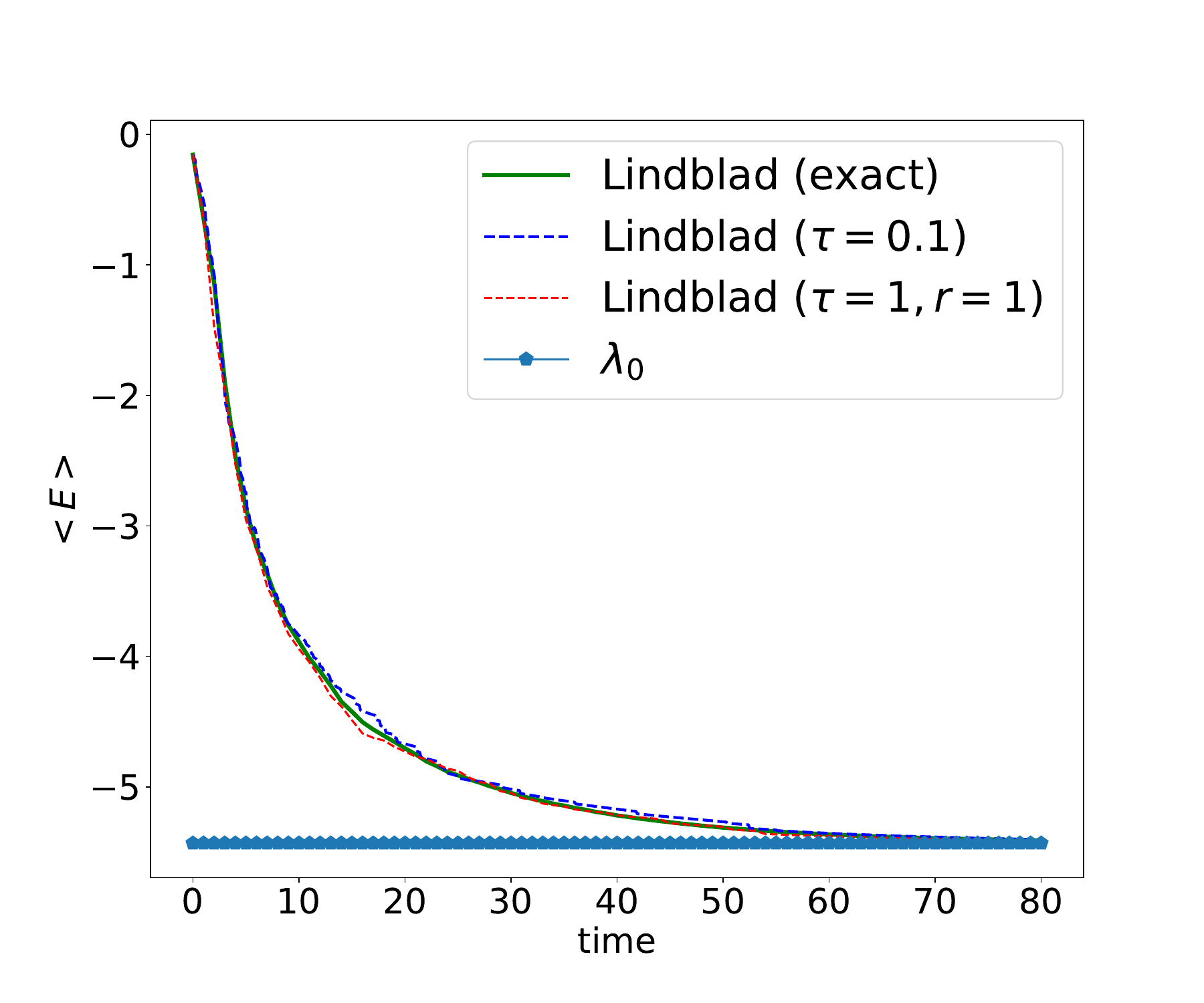}}
  \subfloat[Lindblad simulation time vs overlap]{\includegraphics[width=6cm,height=4.5cm]{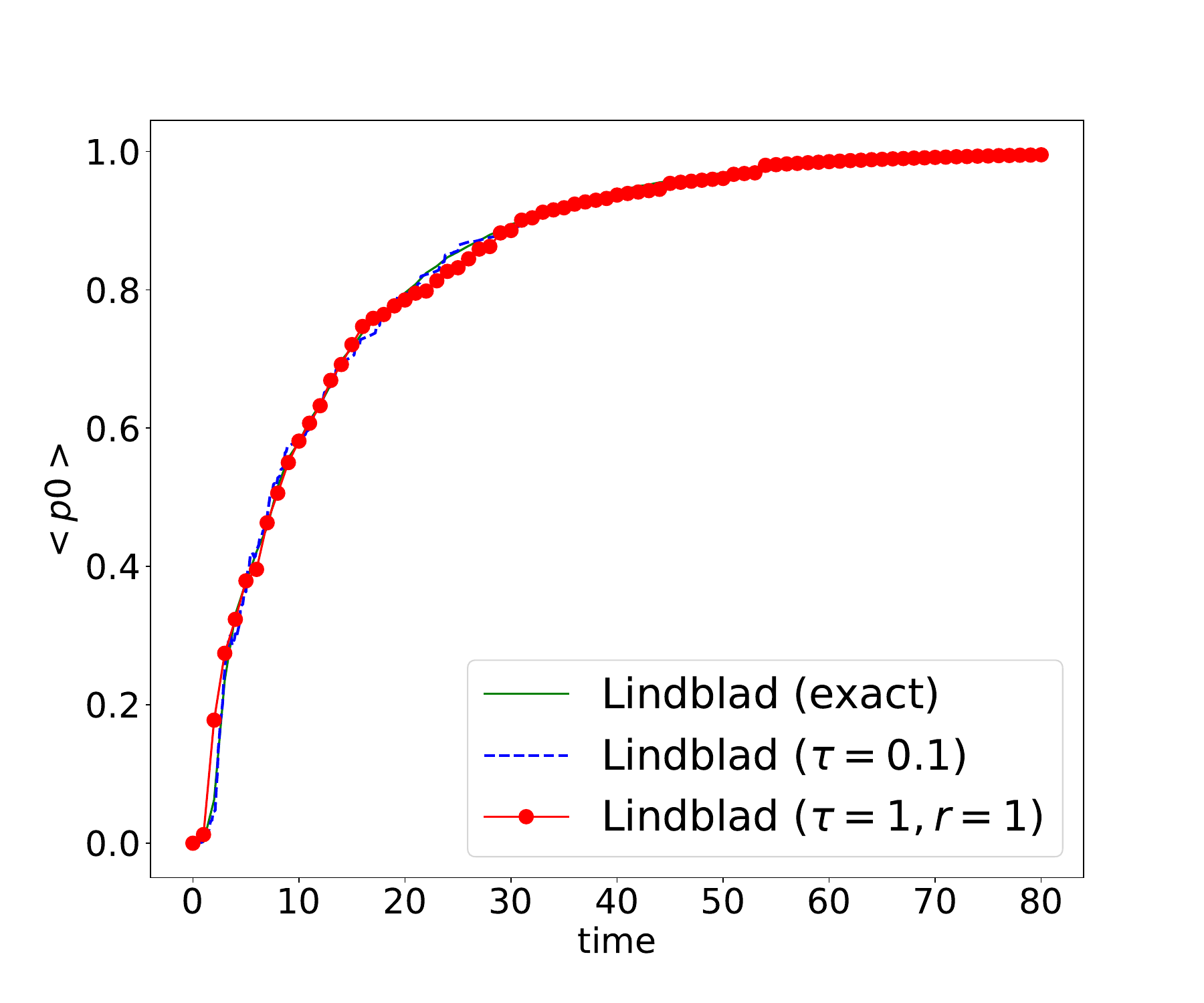}}
  \\
     \subfloat[Hamiltonian simulation time vs energy]{\includegraphics[width=6cm,height=4.5cm]{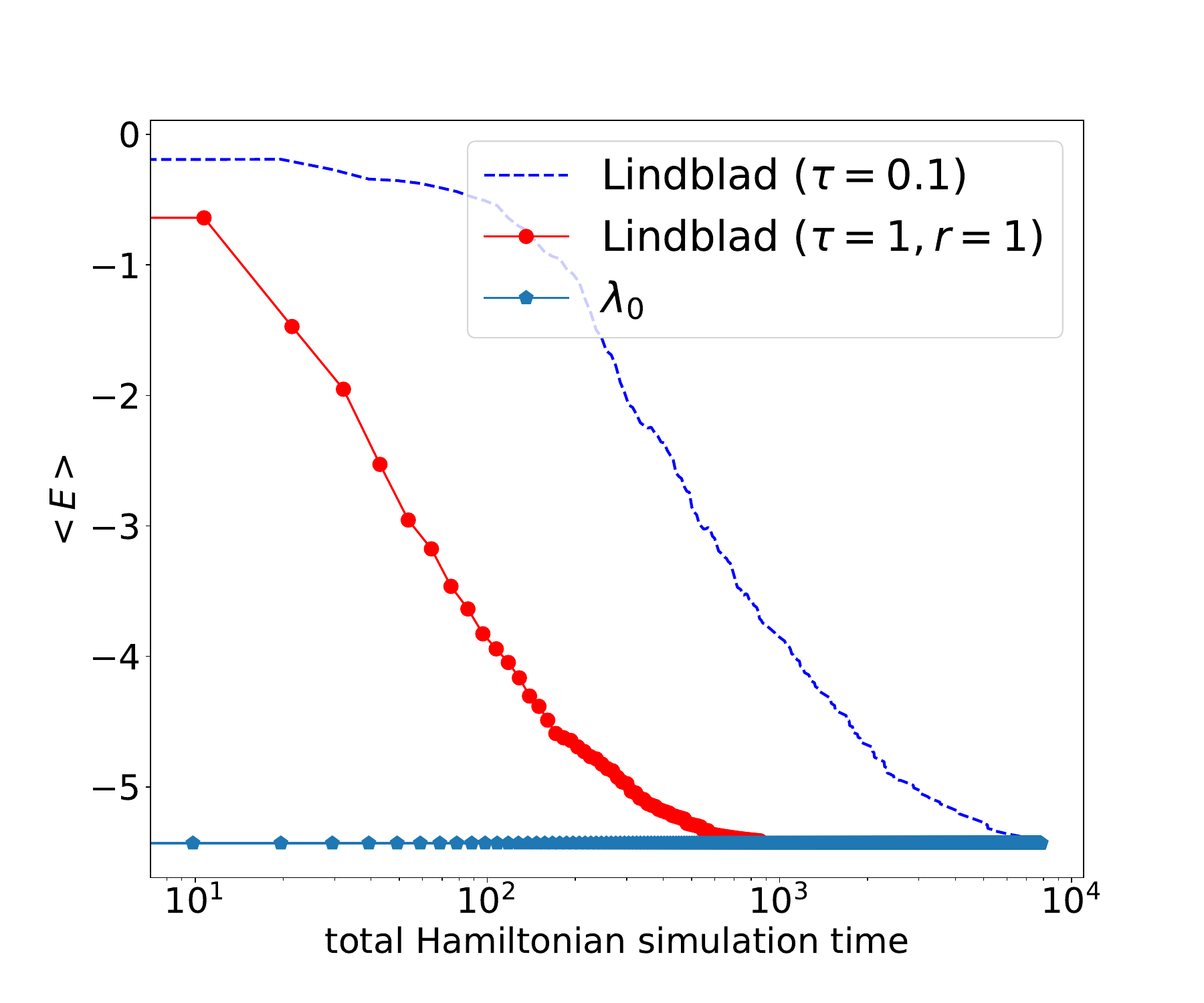}}
  \subfloat[Hamiltonian simulation time vs overlap]{\includegraphics[width=6cm,height=4.5cm]{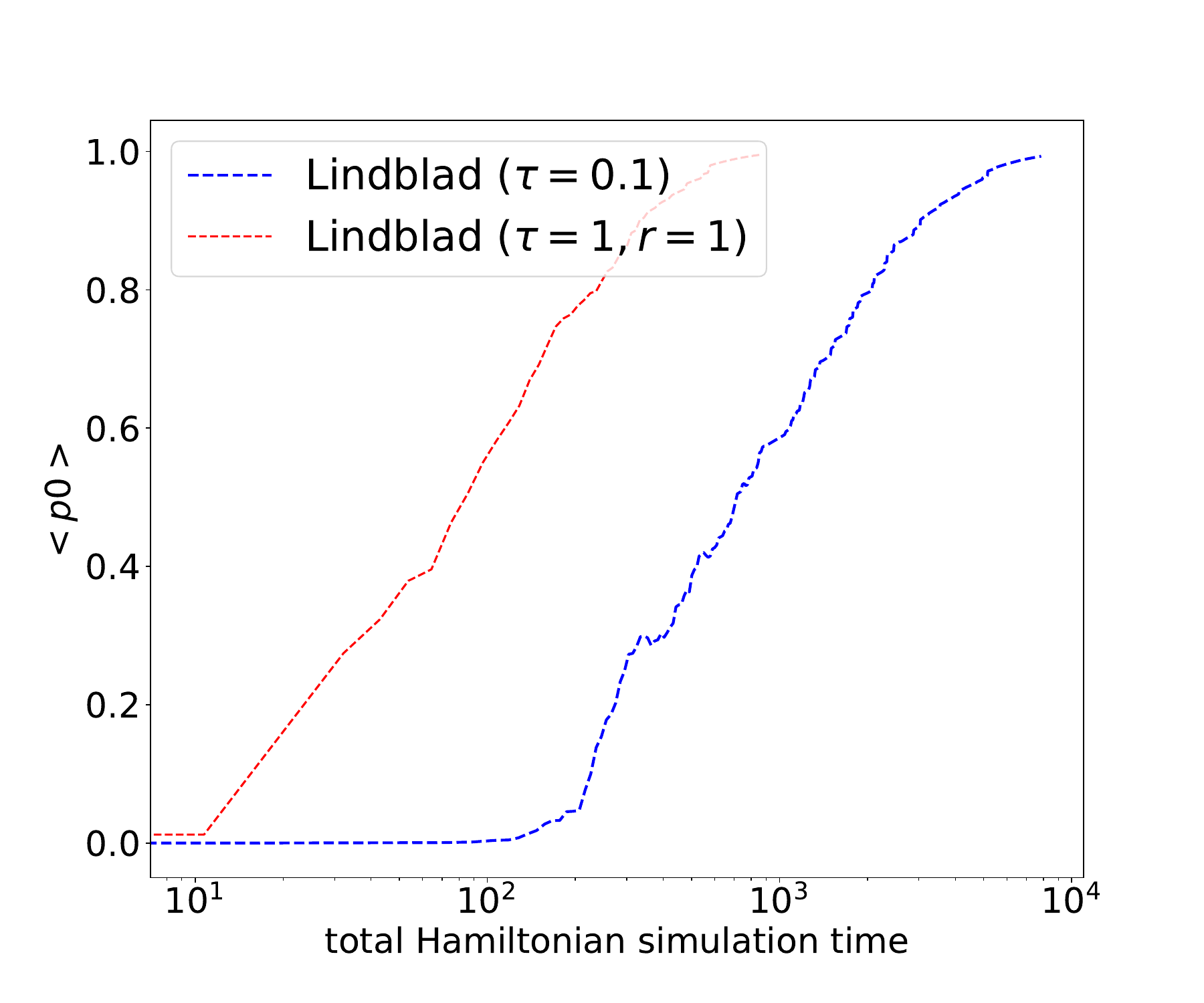}}
   \caption{Continuous vs discrete-time Lindblad dynamics for TFIM-4. The continuous-time simulation uses a small time step $\tau=0.1$. Here, the Hamiltonian simulation time refers to the sum of the Hamiltonian simulation $t$ in all $e^{\pm iHt}$ subroutines used in the circuit.}
   \label{fig:TFIM_4}
\end{figure*}

\subsection{TFIM-6 model}
To simulate Lindblad dynamics with TFIM-6 model, we set $L=6$, the coupling constant $g=1.2$ in \eqref{eqn:H_Ising}, and the local Hermitian operator $A=Z_1=Z\otimes I^{\otimes L-1}$.

In our numerical simulations, we set $\tau=1$ and $r=2$ for discrete-time Lindblad simulation  and $\tau=0.1$ for the continuous-time Lindblad simulation. Again, we start with an initial state with \textit{zero overlap} and repeat each Lindblad dynamics simulation 100 times. The results are depicted in Figure \ref{fig:TFIM_6}. Our result demonstrates the effectiveness of the Lindblad dynamics in generating the ground state of TFIM-6 model. The trajectories of the continuous and discrete-time dynamics noticeably differ from each other, yet they can both prepare the ground state, and the mixing times are comparable. Additionally, the discrete-time Lindblad dynamics ($\tau=1,r=2$) demands almost one-tenth of the total Hamiltonian simulation time required by the approximated continuous-time Lindblad dynamics ($\tau=0.1$). 

\begin{figure*}[!htbp]
\centering
  \subfloat[Hamiltonian simulation time vs energy]{\includegraphics[width=6cm,height=4.5cm]{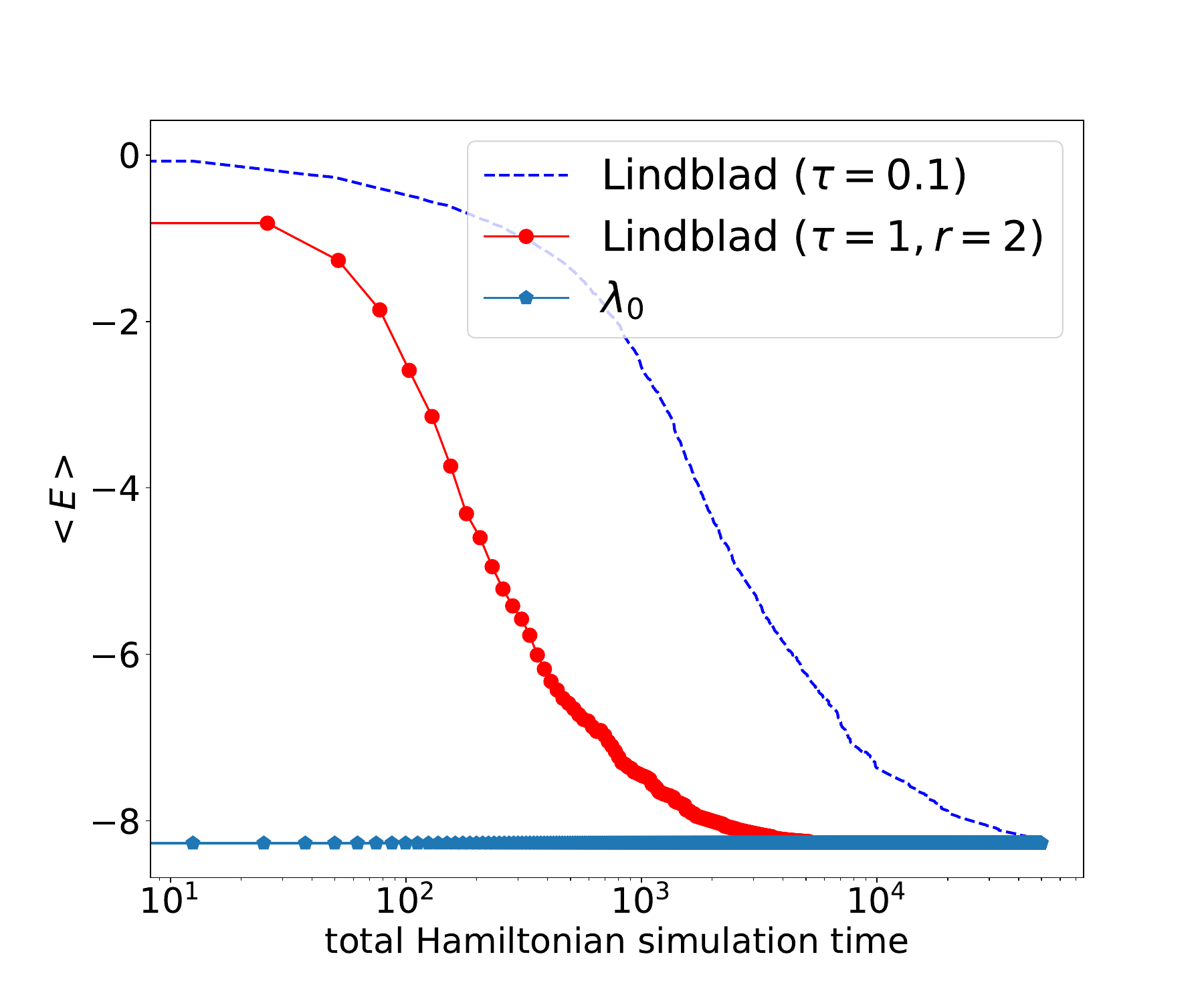}}
     \subfloat[Hamiltonian simulation time vs overlap]{\includegraphics[width=6cm,height=4.5cm]{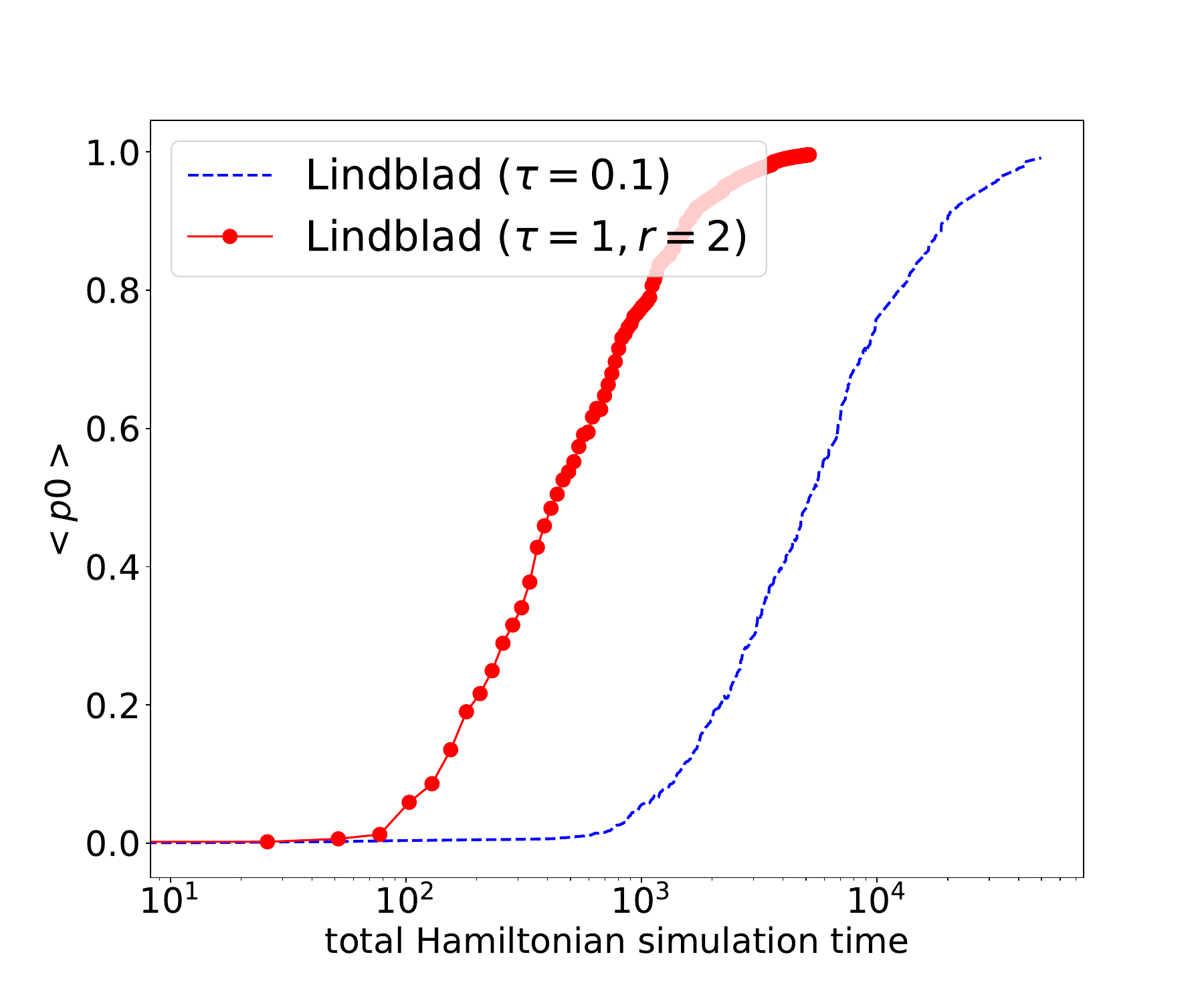}}
  \subfloat[Lindblad simulation time vs overlap]{\includegraphics[width=6cm,height=4.5cm]{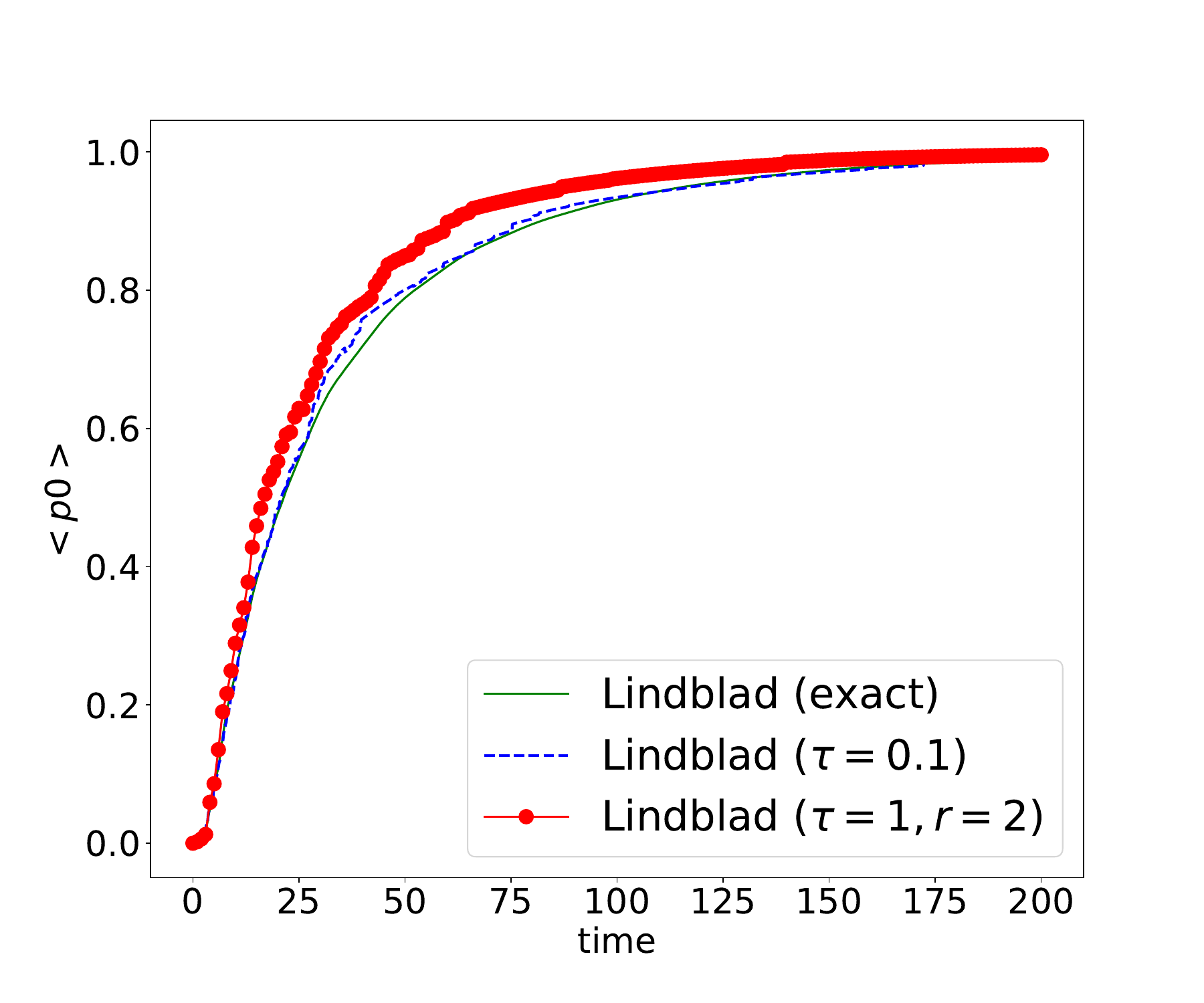}}
   \caption{Performance of continuous versus discrete-time Lindblad dynamics for preparing the ground state for the TFIM with $6$ sites. The trajectory of the discrete-time dynamics (using a large time step $\tau=1$) deviates from that of the continuous-time Lindblad dynamics, but it successfully prepares the ground state and is more efficient than the continuous-time dynamics.}
\label{fig:TFIM_6}
\end{figure*}

\subsection{Hubbard model}\label{sec:Hubbard}
Consider the one-dimensional Hubbard model defined on $L$ spinful sites with open boundary conditions
\[
\begin{aligned}
H=&-t\sum^{L-1}_{j=1}\sum_{\sigma\in\{\uparrow,\downarrow\}}c^\dagger_{j,\sigma}c_{j+1,\sigma}\\
&+U\sum^L_{j=1}\left(n_{j,\uparrow}-\frac{1}{2}\right)\left(n_{j,\downarrow}-\frac{1}{2}\right).    
\end{aligned}
\]
Here $c_{j,\sigma}(c^\dagger_{j,\sigma})$ denotes the fermionic annihilation (creation) operator on the site $j$ with spin $\sigma$. $\left\langle\cdot,\cdot\right\rangle$ denotes sites that are adjacent to each other. $n_{j,\sigma}=c^\dagger_{j,\sigma}c_{j,\sigma}$ is the number operator. 

We choose $L=4$, $t=1$, $U=4$, and the coupling operator is chosen to be a local hopping operator $A=c^\dagger_{1,\uparrow}c_{2,\uparrow}-c_{1,\uparrow}c^\dagger_{2,\uparrow}+c^\dagger_{1,\downarrow}c_{2,\downarrow}-c_{1,\downarrow}c^\dagger_{2,\downarrow}$. We set $\tau=0.5$ and $r=2$ for discrete-time Lindblad dynamics simulation and $\tau=0.025$ for continuous-time Lindblad dynamics simulation. The stopping times are set to $T=100$, and each Lindblad dynamics simulation is repeated 100 times starting from an initial state with zero overlap between the ground state. The results are presented in Figures \ref{fig:Hubbard_4}. In our observations, we find that in both dynamics, the energy decreases to $\lambda_0$ and the overlap with the ground state increases to $1$. We observe that both the continuous-time and discrete-time Lindblad dynamics exhibit a fast convergence rate over time and achieve an accurate ground state construction. 

\begin{figure*}[!htbp]
\centering
  \subfloat[Lindblad simulation time vs energy]{\includegraphics[width=6cm,height=4.5cm]{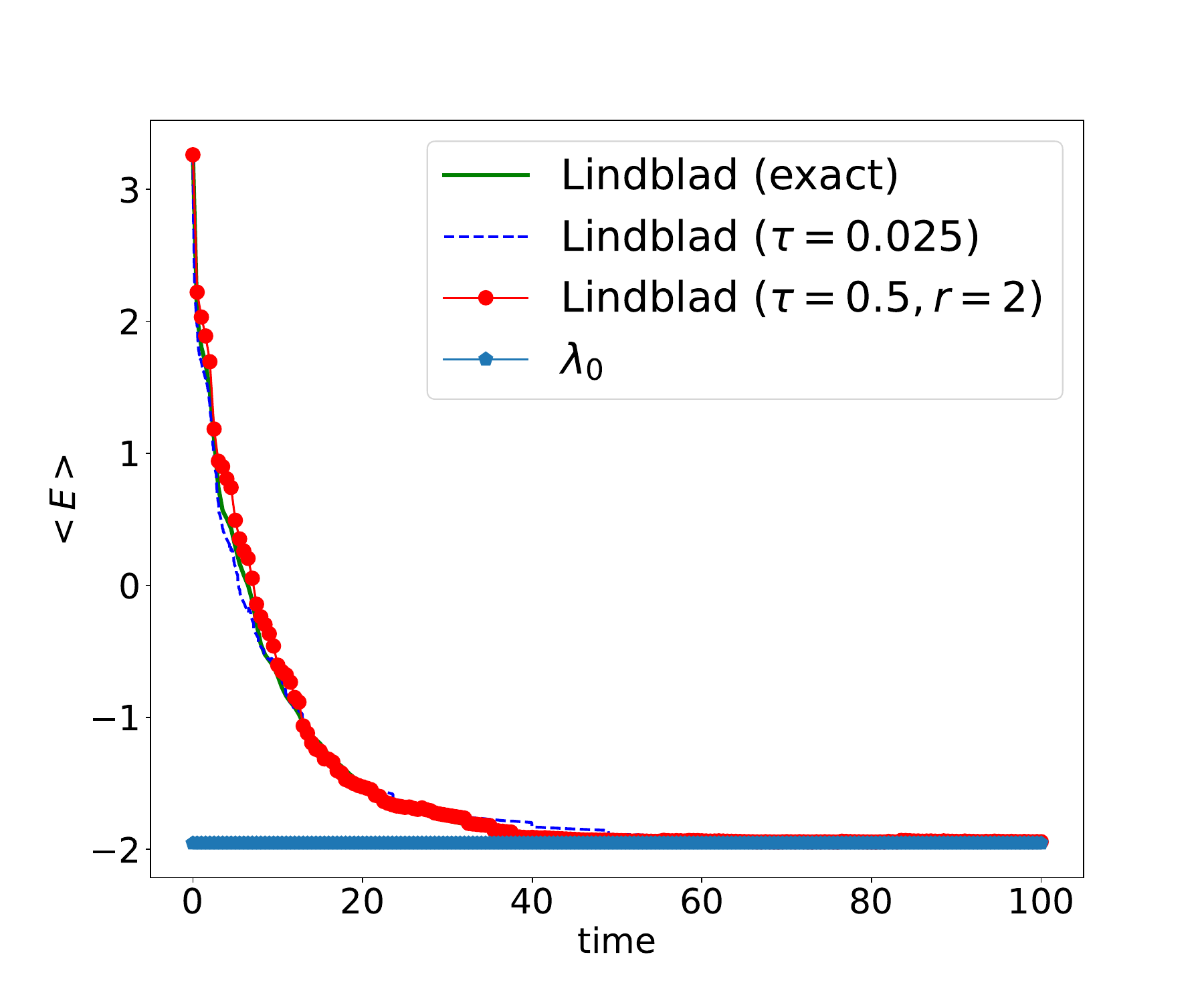}}
  \subfloat[Lindblad simulation time vs overlap]{\includegraphics[width=6cm,height=4.5cm]{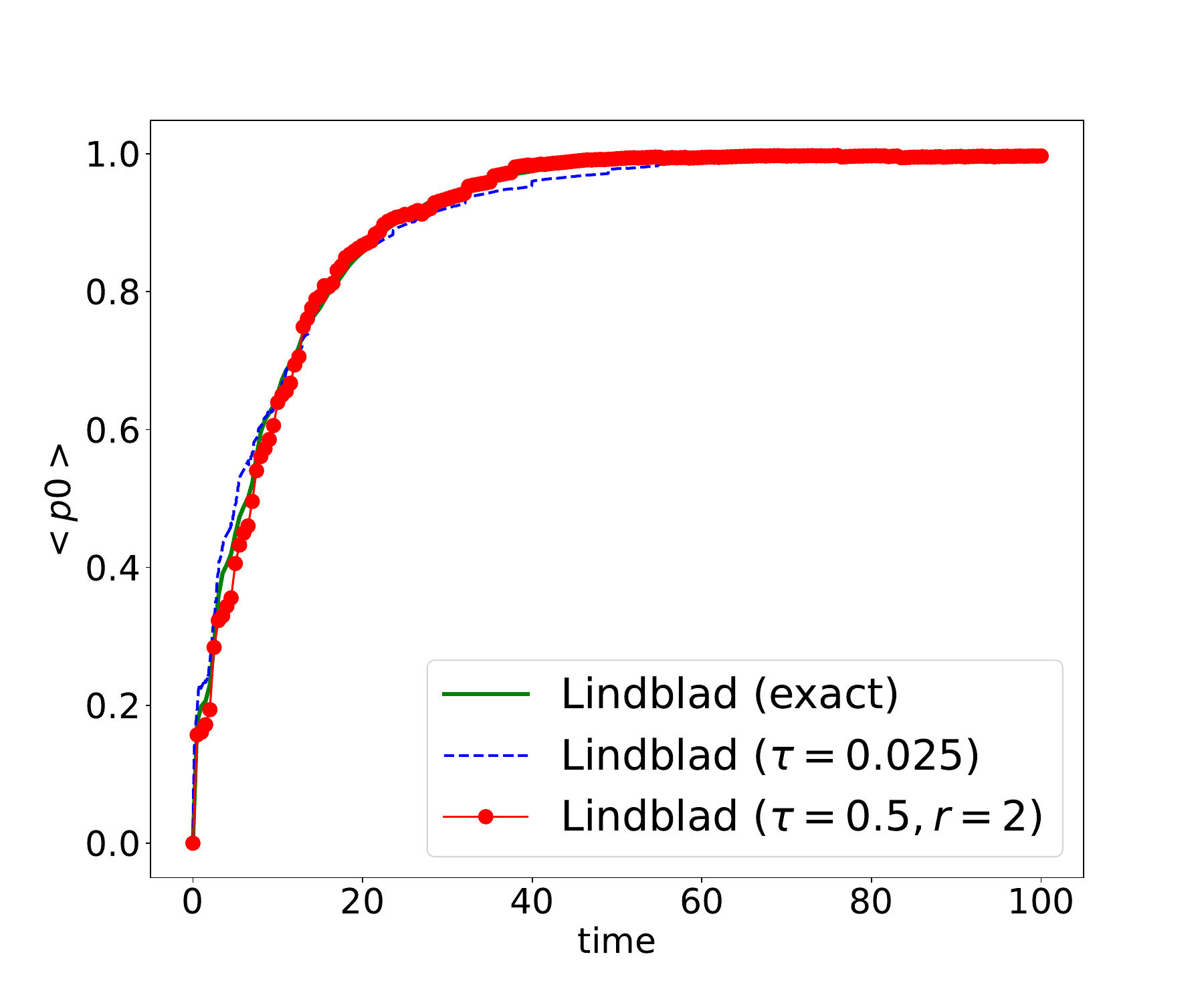}}\\
  \subfloat[Hamiltonian simulation time vs energy]{\includegraphics[width=6cm,height=4.5cm]{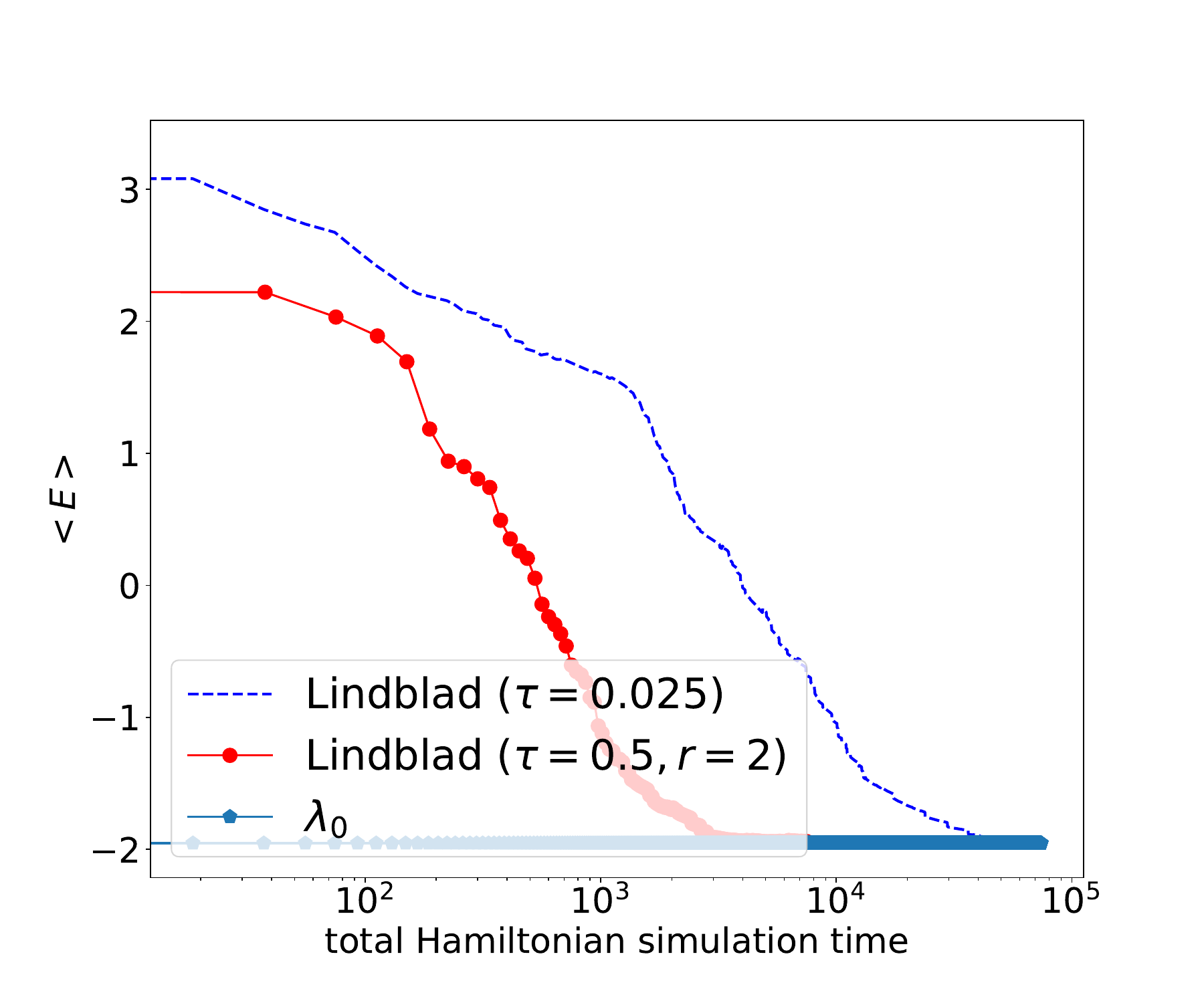}}
  \subfloat[Hamiltonian simulation time vs overlap]{\includegraphics[width=6cm,height=4.5cm]{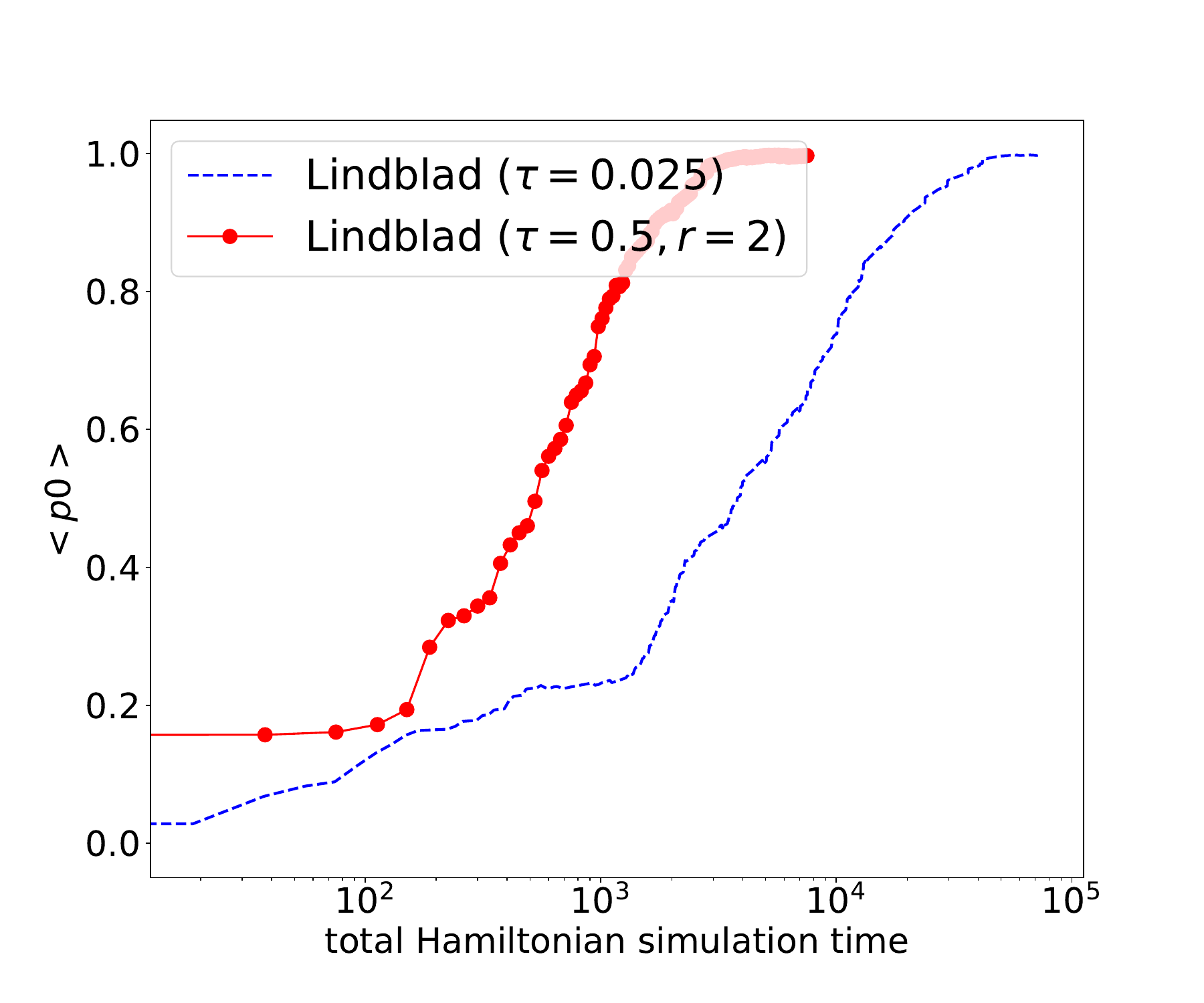}}
   \caption{Continuous vs discrete-time Lindblad dynamics for Hubbard-4.}
   \label{fig:Hubbard_4}
\end{figure*}

\section{Discussion}

This paper presents a Monte Carlo approach for preparing the ground state of a quantum Hamiltonian using a single ancilla qubit. Merely a single jump operator ensures that the system's ground state is a fixed point of the Lindblad dynamics. For existing quantum devices, the maximal Hamiltonian simulation time required (and consequently the circuit depth) may still exceed the capabilities according to our numerical results. Still, in the rapidly advancing field of quantum computing, particularly in quantum error correction and mid-circuit measurement capabilities, our algorithm is well-aligned with these emerging trends and shows potential for future hardware implementation in the early fault-tolerant regime. 

Although it is feasible to  bound the mixing time for specific quantum Gibbs samplers for commuting Hamiltonians~\cite{KastoryanoBrandao2016,BardetCapelGaoEtAl2023}, it is still very challenging to estimate the mixing time for noncommuting Hamiltonians. Very recently, Rouz\'e et al~\cite{rouz2024} succeeded in bounding the mixing time of a certain quantum detailed balanced Lindbladian~\cite{ChenKastoryanoGilyen2023} for preparing the thermal states of certain $k$-local Hamiltonians at high temperatures. However, it is not feasible to directly extend their analysis since the ground state has a very low-temperature. Our paper offers a partial justification that the mixing time for ground state preparation \textit{can} scale polynomially with respect to the system size, which relies on ETH type assumptions and is similar to that in \cite{chen2023fast}. However, the applicability of ETH-like assumptions for ground state preparation requires further investigation. Another possible route that we are currently investigating is to examine the mixing time for quasi-free systems (both bosonic and fermionic), where Lindblad dynamics are exactly solvable~\cite{Prosen2008,ProsenZunkovic2010,BarthelZhang2022}. 

Unlike many ground state preparation algorithms, such as QPE, the performance of this algorithm is mainly determined by the mixing time and may work even with zero initial overlap. Similar to classical Monte Carlo methods, the relationship between mixing time and system size, the practical selection of the coupling matrix $A$, and the optimal choice of the filtering function $f$ may be system-dependent.
For instance, optimizing these choices to prepare the ground state of strongly correlated chemical systems like FeMoCo~\cite{LeeLeeZhaiEtAl2023}, which are currently inaccessible via adiabatic state preparation, might offer a compelling avenue of research.

\subsection{Acknowledgments}
This material is based upon work supported by the U.S. Department of Energy, Office of Science, National Quantum Information Science Research Centers, Quantum Systems Accelerator (Z.D.). 
Additional support is acknowledged from the Challenge Institute for Quantum Computation (CIQC) funded by National Science Foundation (NSF) through grant number OMA-2016245, the  Applied Mathematics Program of the US Department of Energy (DOE) Office of Advanced Scientific Computing Research under contract number DE-AC02-05CH1123, and a Google
Quantum Research Award (L.L.). L.L. is a Simons investigator in Mathematics (Grant Number 825053).   We thank Joao Basso, Garnet Chan, Toby Cubitt, Yulong Dong, Xiantao Li, Subhayan Roy Moulik, Yu Tong for helpful discussions.

\bibliographystyle{abbrv}
\bibliography{ref}

\appendix 
\newpage 
\clearpage
\thispagestyle{empty}
\onecolumngrid

\renewcommand{\thesubsection}{\Alph{section}.\arabic{subsection}}


\section{Notation, facts, and organization}\label{sec:facts}

In the appendixes, we use capital letters for matrices and curly font for superoperators. Besides the usual $\Or$ notation, we use the following asymptotic notations: we write $f=\Omega(g)$ if $g=\Or(f)$; $f=\Theta(g)$ if $f=\Or(g)$ and $g=\Or(f)$; $f=\wt{\Or}(g)$ if $f=\Or(g\operatorname{polylog}(g))$. We use $\|\cdot\|$ to denote vector or matrix 2-norm: when $v$ is a vector we denote by $\|v\|$ its 2-norm, and when $A$ is matrix we denote by $\|A\|$ its operator norm (or the Schatten $\infty$-norm). Given any Hermitian matrix $H$ and a positive integer $N$, there exists a constant $C_N$ such that
\begin{equation}\label{eqn:matrix_fact}
\left\|\exp(-i Ht)-\sum^N_{n=0}\frac{(-iH)^n}{n!}t^n\right\|\leq C_N\|H\|^{N+1}t^{N+1}\,.
\end{equation}

The trace norm (or the Schatten $1$-norm) of a matrix $A$ is $\norm{A}_1=\Tr\left[\sqrt{A^{\dag}A}\right]$. A useful  inequality that we often use in our proof is that
\begin{equation}\label{eqn:trace_norm_inequ}
\|AB\|_1\leq \|A\|\|B\|_1\,.
\end{equation}
for any two matrices $A,B$ such that $AB$ is well defined. 

Given a superoperator $\mathcal{L}$ that acts on operators, the induced $1$-norm is
\begin{equation}
\|\mathcal{L}\|_1:=\sup_{\|\rho\|_1\leq 1}\|\mathcal{L}(\rho)\|_1\,.
\end{equation}

According to GKLS theorem~\cite{GoriniKossakowskiSudarshan1976,Lindblad1976}, if $\mathcal{L}$ is a Lindbladian, which has the form
\[
\mathcal{L}(\rho)=-i [H, \rho]+\sum^M_{m=1}K_m \rho K^{\dag}_m-\frac12 \{K^{\dag}_mK_m,\rho\}\,,
\]
then $\exp(\mathcal{L}t)$ is a quantum channel, i.e., it is a completely positive trace-preserving (CPTP) map. It is also contractive under trace distance~\cite{Ruskai_1994}.
For any two density operators $\rho_1,\rho_2$, and any $t>0$,
\begin{equation}\label{eqn:lindblad_contractive}
\|\exp(\mathcal{L}t)\rho_1-\exp(\mathcal{L}t)\rho_2\|_1 \leq \|\rho_1-\rho_2\|_1.
\end{equation}

The following lemma provides an upper bound on the first-order Trotter error between two superoperators:
\begin{lem}\label{lem:operator_fact} Given two superoperators $\mathcal{L}_1,\mathcal{L}_2$ such that $\exp(\mathcal{L
}_1 t)$, $\exp(\mathcal{L
}_2 t)$, and $\exp((\mathcal{L
}_1+\mathcal{L}_2) t)$ are quantum channels for all $t>0$. Then for all $t>0$,
\begin{equation}\label{eqn:operator_fact}
\left\|\exp((\mathcal{L
}_1+\mathcal{L}_2) t)-\exp(\mathcal{L
}_1 t)\exp(\mathcal{L
}_2 t)\right\|_1=\mathcal{O}\left(\|[\mathcal{L}_1,\mathcal{L}_2]\|_1t^2\right).
\end{equation}
\end{lem}
\begin{proof} This proof is similar to the proof of \cite[Theorem 6]{ChildsSuTranEtAl2021}. Define
\[
\mathcal{E}(t)=\exp(\mathcal{L
}_1 t)\exp(\mathcal{L
}_2 t)\,.
\]
Then, we obtain
\[
\frac{\ud \mathcal{E}(t)}{\ud t}=(\mathcal{L
}_1+\mathcal{L}_2)\mathcal{E}(t)+[\exp(\mathcal{L
}_1 t),\mathcal{L
}_2]\exp(\mathcal{L
}_2 t)\,,
\]
which implies
\[
\mathcal{E}(t)=\exp((\mathcal{L
}_1+\mathcal{L}_2) t)+\int^t_0\exp((\mathcal{L
}_1+\mathcal{L}_2)(t-\tau))[\exp(\mathcal{L
}_1 \tau),\mathcal{L
}_2]\exp(\mathcal{L
}_2 \tau)\ud \tau\,.
\]
Because $\|[\exp(\mathcal{L
}_1 \tau),\mathcal{L
}_2]\|_1=\mathcal{O}\left([\mathcal{L}_1,\mathcal{L}_2]\tau\right)$ and $\|\exp((\mathcal{L
}_1+\mathcal{L}_2)(t-\tau))\|_1=\|\exp(\mathcal{L
}_2 \tau)\|_1=1$, we obtain that
\[
\left\|\mathcal{E}(t)-\exp((\mathcal{L
}_1+\mathcal{L}_2) t)\right\|_1=\mathcal{O}\left([\mathcal{L}_1,\mathcal{L}_2]t^2\right)\,,
\]
which proves \eqref{eqn:operator_fact}.
\end{proof}

Next, we bound the $1$-norm of $\mc{L}_K$ corresponding to the jump operator $K$ in \cref{eqn:jump_time} of the Lindblad dynamics in \cref{eqn:Lindblad_dynamics}:
\begin{lem}\label{lem:facts}
Let $\mathcal{L}_K$ be defined in \eqref{eqn:Lindblad_dynamics} with $f\in L^1(\mathbb{R})$, then
\begin{equation}\label{eqn:K_bound}
\|\mathcal{L}_K\|_1=\mathcal{O}(\|f\|_{L^1}^2\|K\|^2)=\mathcal{O}(\|f\|_{L^1}^2\|A\|^2)\,.
\end{equation}
\end{lem}
\begin{proof}  To prove \eqref{eqn:K_bound}, we first use \eqref{eqn:trace_norm_inequ} to obtain $\|\mathcal{L}_K\|_1=\mathcal{O}(\|K\|^2)$. Then
\[
\|K\|\leq \|A\|\int^\infty_{-\infty}|f(s)|\ud s.
\]
\end{proof}
Our analysis employs the Gevrey function, a subclass of smooth functions characterized by well-controlled decay of the Fourier coefficients. This characteristic plays a crucial role in the quadrature analysis.
\begin{defn}[Gevrey function] \label{def:gevrey}
Let $\Omega\subseteq \RR^d$ be a domain. A complex-valued $C^\infty$ function $h: \Omega\to \CC$ is a \emph{Gevrey function} of order $s\ge 0$, if there exist constants $C_1,C_2>0$ such that for every $d$-tuple of nonnegative integers $\alpha = (\alpha_1,\alpha_2,\ldots,\alpha_d)$, 
\begin{equation}
\left\|\partial^\alpha h\right\|_{L^\infty(\Omega)}\leq C_1C^{|\alpha|}_2|\alpha|^{|\alpha|s}\,,
\end{equation}
where $|\alpha|=\sum^d_{i=1} |\alpha_i|$. For fixed constants $C_1,C_2,s$, the set of Gevrey functions is denoted by $\mathcal{G}^{s}_{C_1,C_2}(\Omega)$. Furthermore, $\mathcal{G}^s=\bigcup_{C_1,C_2>0}\mathcal{G}^s_{C_1,C_2}$.
\end{defn}

The rest of the supplementary material is organized as follows: 
\begin{itemize}
    \item Section \ref{sec:filter} introduces the assumption and properties of the filter function $f$.

\item Section \ref{sec:acc_lindblad_dynamics} introduces the discrete-time Lindblad dynamic and the simulation algorithm.

\item Section \ref{sec:analysis} analyzes the costs associated with our algorithm and provides proofs for the lemmas in \cref{sec:overview_algorithm}.

\item Section \ref{sec:conv} discusses the convergence properties of the continuous-time Lindblad dynamics under the Eigenstate Thermalization Hypothesis (ETH) type ansatz.


\end{itemize}

\section{Choice of filter function}\label{sec:filter}

As mentioned earlier, ensuring that the ground state remains a fixed point of the Lindblad dynamics requires $\hat{f}$ to satisfy the condition stated in \cref{eqn:f_equal_to_zero}. 
Furthermore, as discussed in \cref{sec:Step_3} in order to numerically simulate $\exp(\mathcal{L}_K(t))$, we need to discretize and truncate the integral $\int^\infty_{-\infty} f(s)A(s)\ud s$. To ensure that this truncation does not introduce significant errors, additional assumptions must be imposed on $\hat{f}$ such that $f(t)$ decays rapidly as $|t|$ approaches infinity. In summary, we adopt the following assumption regarding $\hat{f}$:


\begin{assumption}[Filter function in the frequency domain]\label{assum:f_freq}
Given  $S_\omega>2\Delta=2(\lambda_1-\lambda_0)>0$ and $M=\lceil S_\omega/\Delta\rceil+1$, the filter function in the frequency domain is
\begin{equation}\label{eqn:hat_f}
\hat f(\omega)=\hat{u}(\omega/S_\omega)\hat{v}(\omega/\Delta)\,.
\end{equation}
Here, we assume $\hat{u}$ is a positive function and belongs to Gevrey class $\mathcal{G}^{\alpha}_{A_{1,u},A_{2,u}}(\mathbb{R})$ for some $A_{1,u},A_{2,u}>0$ and $\alpha>1$, meaning that
    \[
    \sup_{\omega\in\mathbb{R}}\left|\frac{\mathrm{d}^N}{\mathrm{d}\omega^N}\left(\hat{u}(\omega)\right)\right|\leq A_{1,u} A^N_{2,u}N^{N\alpha}
    \]
    for any $N\in\mathbb{N}$.
Also, $\mathrm{supp}(\hat{u})\subset[-1,1]$ and $\hat{u}(\omega)=\Omega(1)$ when $\omega\in [-1/2,1/2]$. In addition, we assume $\hat{v}\in \mathcal{G}^{\alpha}_{A_{1,v},A_{2,v}}(\mathbb{R})$, $\left\|\frac{\mathrm{d}}{\mathrm{d}\omega}\hat{v}\right\|_{L^1}=\mathcal{O}(1)$, $\mathrm{supp}(\hat{v})\subset (-\infty,0]$.
\end{assumption}
\begin{rem} It is not difficult to find Gevrey functions that have a compact support. Let $\mathcal{G}^s=\cup_{c,a\in\mathbb{R}_+}\mathcal{G}^s_{c,a}$ be the set of Gevrey functions with power $s$. According to \cite[Corollary 2.8]{Adwan_2017}, for any $s>1$, there exists $\phi\in \mathcal{G}^s$ such that $\mathrm{supp}(\phi)\subset [-1,1]$ and $\phi(\omega)=1$ when $|\omega|\leq 1/2$.
\end{rem}
Under the above assumptions, the filter function $\hat{f}$ in the Fourier domain satisfies the following conditions:
\begin{align}\label{eqn:f_time_bound}
   \hat f(\omega) = \begin{cases}
        \Omega(1) &\quad \omega = [-S_{\omega}+\Delta,-\Delta]\\
        0 &\quad \omega \notin [-S_{\omega},0] 
    \end{cases}\quad \text{for all}\quad \omega\in\RR.
\end{align}


The Gevrey assumption simplifies our quadrature error analysis. 
Moreover, the expression is employed to ensure that $\hat{f}$ possesses a sufficiently large support. In practical applications, it is sufficient to find a smooth function $f$ that approximately satisfies \eqref{eqn:f_time_bound}. 
We present two results on the Fourier transform of $\hat{f}(\omega)$
\begin{equation}\label{eqn:fomega_fourier}
f(s):=\frac{1}{2\pi}\int_{\RR} \hat{f}(\omega)e^{-i\omega s}\ud \omega
\end{equation}
This is slightly different from the standard convention of the Fourier transform, but it agrees with the convention used in Ref.~\cite{ChenKastoryanoBrandaoEtAl2023}.

To study the property of $f$, we first introduce two results for Gevrey functions from \cite[Appendix C]{ding2024efficient}:
\begin{lem}[Product of Gevrey functions]\label{lem:product_gevrey} 
Given $h\in\mathcal{G}^{s}_{C_1,C_2}(\mathbb{R}^d)$ and $h'\in\mathcal{G}^{s'}_{C'_1,C'_2}(\mathbb{R}^d)$, then
\[
h\cdot h'\in \mathcal{G}^{\max\{s,s'\}}_{C_1C'_1,C_2+C'_2}(\mathbb{R}^d)\,.
\]
\end{lem}
\begin{lem}\label{lem:hat_Gev_decay}
Given $h\in\mathcal{G}^s_{C_1,C_2}(\mathbb{R}^d)$ with compact support $\Omega=\mathrm{supp}(h)$ and $s\geq 1$, define 
\begin{equation*}
    H(y)=\frac{1}{(2\pi)^d}\int_{\mathbb{R}^d}h(x)e^{-ix\cdot y} \ud x\,.
\end{equation*}
Then, for any $y\in\mathbb{R}^d$, there holds
\[
\left|H(y)\right|\leq \frac{C_1|\Omega|}{(2\pi)^d} \, e^{\frac{esd}{2}-\frac{s}{C^{1/s}_2e} \norm{y}^{1/s}}\,,
\]
where $|\Omega|=\int_{\Omega}1\,\mathrm{d}x$ is the volume of $\Omega$ and $\norm{y}$ is the $2$-norm of the vector $y$. 
\end{lem}

We present the properties of $f(s)$. First, we demonstrate that in the time domain, $f(s)$ decays superpolynomially as $\abs{s}\to \infty$, with the decay rate being linearly dependent on $\Delta$. Furthermore, utilizing the structure of \eqref{eqn:hat_f}, we show that $f(s)=\mathcal{O}(1/|s|)$, where the constant remains independent of $S_\omega$ and $\Delta$. By combining these two observations, we show that the dependence of $\|f\|_{L^1}$ on the gap $\Delta$ is only logarithmic. These findings are summarized in the following lemma.
\begin{lem}\label{lem:as_f_simulation}
Let $\hat{f}(\omega)$ satisfy \cref{assum:f_freq}. There exist constants $C_{1,f},C_{2,f}>0$ depending on $A_{1,u},A_{1,v},A_{2,u},A_{2,v},\alpha$ such that
\begin{equation}\label{eqn:f_super_decay} \abs{f(s)}=\mathcal{O}\left(C_{1,f}S_\omega\exp\left(-C_{2,f}\left|s\Delta\right|^{1/\alpha}\right)\right)\,.
\end{equation}
and
\begin{equation}\label{eqn:f_L_1} \norm{f(s)}_{L^1}=\widetilde{\mathcal{O}}\left(\log(S_\omega/\Delta)\right)\,.
\end{equation}
where $\widetilde{\mathcal{O}}$ contains constants depending on $A_{1,u},A_{1,v},A_{2,u},A_{2,v},\alpha$.
\end{lem}
\begin{proof}

\cref{lem:product_gevrey} gives
\[
\hat{f}(\omega)\in \mathcal{G}^{\alpha}_{A_{1,u}A_{2,v},A_{1,u}/S_\omega+A_{2,v}/\Delta},\quad \left|\mathrm{supp}\left(\hat{f}(\omega)\right)\right|=S_\omega\,.
\]
According to \cref{lem:hat_Gev_decay}, we have
\[
\left|f(s)\right|= \mathcal{O}\left(C_{1,f}S_\omega\exp\left(-C_{2,f}\left|s\Delta\right|^{1/\alpha}\right)\right)\,,
\]
where $C_{1,f},C_{2,f}>0$ are two constants only depending on $A_{1,u},A_{1,v},A_{2,u},A_{2,v},\alpha$. 

Next, we consider $\|f(s)\|_{L^1}$. Notice
\[
\begin{aligned}
|sf(s)|&\leq \frac{1}{2\pi}\int_{\mathbb{R}}\left|\frac{\mathrm{d}}{\mathrm{d}\omega}\hat{f}(\omega)\right|\mathrm{d}\omega\leq \frac{1}{2\pi}\int_{\mathbb{R}}\left|\frac{1}{S_\omega}\frac{\mathrm{d}}{\mathrm{d}\omega}\hat{u}\left(\frac{\omega}{S_\omega}\right)\hat{v}\left(\frac{\omega}{\Delta}\right)\right|+\left|\frac{1}{\Delta}\hat{u}\left(\frac{\omega}{S_\omega}\right)\frac{\mathrm{d}}{\mathrm{d}\omega}\hat{v}\left(\frac{\omega}{\Delta}\right)\right|\mathrm{d}\omega\\
&\leq \|\hat{v}\|_{L^\infty}\left\|\frac{\mathrm{d}}{\mathrm{d}\omega}\hat{u}\right\|_{L^1}+\|\hat{u}\|_{L^\infty}\left\|\frac{\mathrm{d}}{\mathrm{d}\omega}\hat{v}\right\|_{L^1}=\mathcal{O}(1)\,.
\end{aligned}
\]
This inequality implies
\begin{equation}\label{eqn:f_1_inverse_s}
|f(s)|= \mathcal{O}\left(\abs{s}^{-1}\right)\,.
\end{equation}
Using the above inequality and \eqref{eqn:f_super_decay}, for any given $T>0$, we have
\begin{equation}\label{eqn:f_L_1_bound_derive}
\begin{aligned}
\|f\|_{L^1}&=\int^{T}_{-T}|f(s)|\mathrm{d}s+\int_{|s|>T}|f(s)|\mathrm{d}s= \mathcal{O}\left(\int^{T}_{-T}\min\left\{\frac{1}{|s|},S_\omega\right\}\mathrm{d}s+C_{1,f}S_\omega\int_{|s|>T}\exp\left(-C_{2,f}\left|s\Delta\right|^{1/\alpha}\right)\mathrm{d}s\right)\\
&=\mathcal{O}\left(\int^{1/S_\omega}_{-1/S_\omega}S_\omega\mathrm{d}s+\int_{1/S_\omega\leq |s|\leq T}\frac{1}{|s|}\mathrm{d}s+C_{1,f}S_\omega\int_{|s|>T}\exp\left(-C_{2,f}\left|s\Delta\right|^{1/\alpha}\right)\mathrm{d}s\right)\\
&=\mathcal{O}\left(\log(T)+\log(S_\omega)+\frac{C_{1,f}S_\omega}{\Delta}\int_{|s|>\Delta T}\exp\left(-C_{2,f}\left|s\right|^{1/\alpha}\right)\mathrm{d}s\right)\\
&=\mathcal{O}\left(\log(T)+\log(S_\omega)+\frac{C_{1,f}S_\omega}{\Delta}\int_{|s|>(\Delta T)^{1/\alpha}}\alpha s^{\alpha-1}\exp\left(-C_{2,f}s\right)\mathrm{d}s\right)\,,
\end{aligned}
\end{equation}
where we use \eqref{eqn:f_super_decay},\eqref{eqn:f_1_inverse_s}, and $\|f\|_{L^\infty}=\mathcal{O}(S_\omega)$ in the second equality. In the third and fourth equalities, we apply the change of variables. 

Let $C_\alpha$ be a constant depending on $\alpha,C_{2,f}$ such that when $s>C_\alpha$, $s^{\alpha-1}\exp\left(-C_{2,f}s\right)\leq \exp\left(-C_{2,f}s/2\right)$. Because $S_\omega/\Delta>2$, we can set
\[
T=\left(\frac{C_{\alpha}}{\log(2)}+\frac{2}{C_{2,f}}\right)^{\alpha}\log^\alpha(S_\omega/\Delta)/\Delta\,.
\]
Then
\[
\begin{aligned}
&\frac{C_{1,f}S_\omega}{\Delta}\int_{|s|>(\Delta T)^{1/\alpha}}\alpha s^{\alpha-1}\exp\left(-C_{2,f}s\right)\mathrm{d}s\\
\leq &\frac{\alpha C_{1,f}S_\omega}{\Delta} \int_{|s|>(\Delta T)^{1/\alpha}}\exp\left(-\frac{C_{2,f}s}{2}\right)\mathrm{d}s\leq \frac{2\alpha C_{1,f}S_\omega}{C_{2,f}\Delta}\exp\left(-\frac{C_{2,f}(\Delta T)^{1/\alpha}}{2}\right)=\mathcal{O}(1)
\end{aligned}
\]
where the constant only depends on $C_{1,f},C_{2,f},\alpha$. Plugging this in \eqref{eqn:f_L_1_bound_derive}, We conclude the proof of \cref{eqn:f_L_1}.
\end{proof}
Finally, we present the Nyquist--Shannon sampling theorem (or Poisson summation formula) to express the continuous Fourier transform exactly as an infinite sum:
\begin{thm}[Nyquist--Shannon sampling theorem]
\label{thm:TW_bound} Suppose $\mathrm{Supp}\left(\hat{h}\right)\subset[-S,S]$. Then for any $\tau_s<\frac{\pi}{S}$ and $\omega\in [-S,S]$, we have
\[
\int^\infty_{-\infty}h(s)\exp(-i\omega s)\ud s=\tau_s\sum_{n\in\ZZ}h(n\tau_s)\exp(-in\omega\tau_s)\,.
\]
\end{thm}

\section{Discrete-time Lindblad dynamics with one ancilla qubit}\label{sec:acc_lindblad_dynamics}

The short time propagator for simulating the continuous-time Lindblad dynamics in \cref{eqn:single_lindblad_shorttime_2}  is limited to first-order accuracy, which is mainly limited by the inexactness of the propagator according to \cref{lem:Lindblad_simulation_error} (also discussed at the end of \cref{sec:first_order_trotter} and after \cref{thm:discretization_error_formal}). However, our goal is not to simulate the continuous-time Lindblad dynamics but to prepare the ground state. If we apply the dilated jump operator $\wt{K}$ to $\ket{0}\bra{0}\otimes \rho_g$, then according to \cref{eqn:K_fix_ground}, 
\begin{equation}
\wt{K}\ket{0}\bra{0}\otimes \rho_g=\begin{pmatrix}
0\\ K\rho_g
\end{pmatrix}=\begin{pmatrix}
0\\0
\end{pmatrix}.
\end{equation}
This implies 
\[
\exp\left(-i\wt{K}\tau\right)[\ket{0}\bra{0}\otimes \rho_g]\exp\left(i\wt{K}\tau\right)=\ket{0}\bra{0}\otimes \rho_g\,,
\] 
for any $\tau>0$. Thus, when the quadrature error and the Trotter error are properly controlled, we have
\begin{equation}
\exp(\mathcal{L}_H\tau)\mc{W}_a(\tau)[\rho_g]\approx \rho_g, \quad \text{for all}\  \tau>0.
\end{equation}
where $\mc{W}_a(\tau)[\rho_g]=\mathrm{Tr}_a\left(\exp\left(-i\wt{K}\sqrt{\tau}\right)[\ket{0}\bra{0}\otimes \rho_g]\exp\left(i\wt{K}\sqrt{\tau}\right)\right)$.

Therefore, the simulation scheme outlined in \cref{eqn:single_lindblad_shorttime_2} could conceivably be used for ground state preparation with an \textit{arbitrarily large} time step $\tau$. By choosing a large time step, the quantum state $\rho_m$ may not necessarily approximate the exact dynamics $\rho(m\tau)$. However, we expect that after the mixing time $\tmix'=M_{\mathrm{mix}}\tau$, $\rho_{M_{\mathrm{mix}}}$ becomes a good approximation to $\rho_g$ because of the fixed point argument. Here, the ``discrete" mixing time $M_{\mathrm{mix}}$ is given in \cref{def:mixing_time}.
If $\tmix'$ is not much larger than $\tmix$, we can significantly reduce the simulation cost thanks to the increase of $\tau$. This gives rise to the \textit{discrete-time Lindblad dynamics}, which is defined as the continuous-time Lindblad simulation with a large time step.

When using a large $\tau$, additional attention must be paid to controlling the Trotter error. Specifically, $e^{-i \wt{K}\sqrt{\tau}}$ should be simulated as
\begin{equation}
e^{-i \wt{K}\sqrt{\tau}}=\left(e^{-i \wt{K}\sqrt{\tau}/r}\right)^r\approx \left(W(\sqrt{\tau}/r)\right)^r,
\end{equation}
where $r$ denotes the number of segments and should be properly chosen.

In summary, one single step of the discrete-time Lindblad dynamics is
\begin{equation}\label{eqn:accelerated_lindblad_shorttime}
\rho_{m+1}=\exp(\mathcal{L}_H\tau)\mc{W}(\tau,r)[\rho_m].
\end{equation}
where
\begin{equation}
\mc{W}(\tau,r)[\rho]=\Tr_a [W(\sqrt{\tau}/r)]^r \left(\ket{0}\bra{0}\otimes\rho\right)  [(W(\sqrt{\tau}/r))^{\dag}]^r\,.
\end{equation}
The quantum circuit for \eqref{eqn:accelerated_lindblad_shorttime} is drawn in Figure \ref{fig:qc_2}. 

Since we also cancel the long-time simulation when formulating $W(\tau,r)$ in this case, if we increase $r$ and fix $T,\tau$, we expect that $\rho_{M_t}$ from
\eqref{eqn:accelerated_lindblad_shorttime} converges to the solution $\rho^a_{M_t}$ of the following discrete dynamics:
\begin{equation}\label{eqn:exact_acc_Lind}
\rho^a_{m+1}=\exp(\mathcal{L}_H\tau)\mathrm{Tr}_a\left(e^{-i \wt{K}\sqrt{\tau}}[\ket{0}\bra{0}\otimes\rho^a_m]e^{i \wt{K}\sqrt{\tau}}\right),\quad\text{where}\ \rho^a_0=\exp(\mathcal{L}_H S_s)[\rho_0]\,.
\end{equation}

\begin{figure}
\centering
{
\includegraphics[width=8cm]{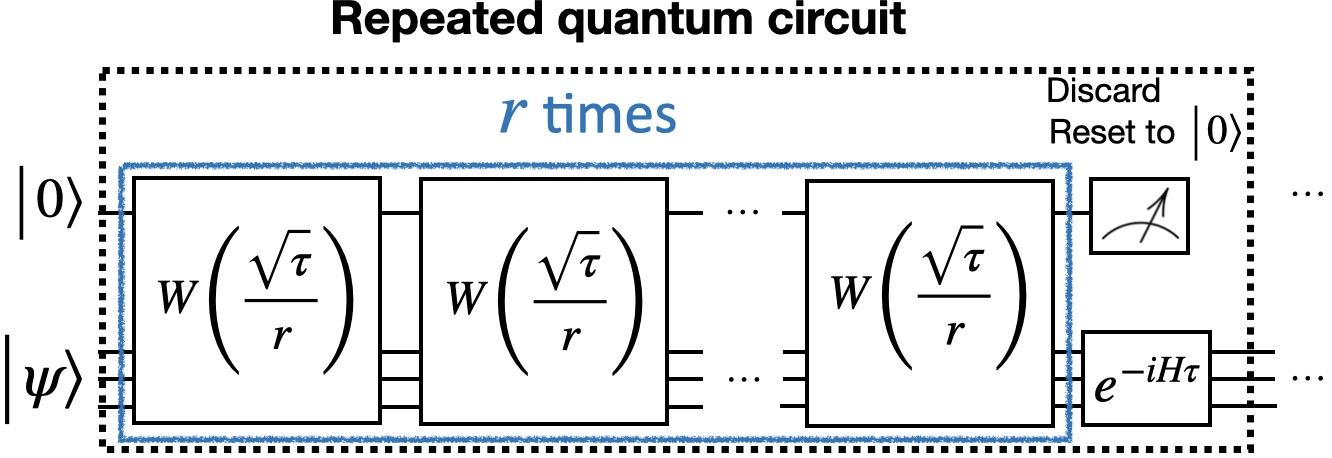}
}
\caption{Quantum circuit for discrete-time Lindblad dynamics simulation. The measurement result of the ancilla qubit is discarded and the ancilla qubit is reset to $\ket{0}$ after each measurement. The circuit for $W$ is given in \cref{fig:qc_1}. The time step $\tau$ can be chosen to be independent of the precision $\epsilon$, and the number of segments $r$ should be properly chosen to control the Trotter error.}
\label{fig:qc_2}
\end{figure}

\section{Analysis of the algorithm}\label{sec:analysis}
In this section, we investigate the complexity of simulating the Lindblad dynamics as described by equations \eqref{eqn:modified_lindblad} and \eqref{eqn:exact_acc_Lind} using the schemes \eqref{eqn:single_lindblad_shorttime_2} and \eqref{eqn:accelerated_lindblad_shorttime}. The rigorous version of Theorems \ref{thm:oneancilla_informal} and \ref{thm:ac_complex_informal} are presented in Theorems \ref{thm:discretization_error_formal} and \ref{thm:acc_error}, respectively.

\subsection{Simulation cost of continuous-time Lindblad dynamics}\label{sec:sc_cl}

Given a stopping time $T>0$ and the time step $\tau$, we assume the number of iterations $M_t=T/\tau$ is an integer. We now prove that $\rho_{M_t}$ from \eqref{eqn:single_lindblad_shorttime_2} converges to $\rho(T)$ from \eqref{eqn:modified_lindblad} with first order accuracy in $\tau$. 

\begin{thm}[Simulation cost of the continuous-time Lindblad dynamics with one ancilla qubit]\label{thm:discretization_error_formal} Given a stopping time $T>0$ and the accuracy $\epsilon>0$, we assume that $\hat{f}$ satisfies \cref{assum:f_freq}.
Then, to obtain $\left\|e^{iHS_s}\rho(T)e^{-iHS_s}-\rho_{M_t}\right\|_{1}\leq \epsilon$, we can choose
\[
\begin{aligned}
&S_s=\Theta\left(\frac{1}{\Delta}\left(\log\left(\frac{S_\omega\|A\|T}{\Delta\epsilon}\right)\right)^\alpha\right),\quad \tau_s=\Theta\left(\frac{1}{1+\max\{\|H\|,S_\omega\}}\right),\\
\  &\tau=\widetilde{\Theta}\left(\frac{\epsilon}{(\|A\|^4+\|A\|^2\|H\|)T}\right)\,,    
\end{aligned}
\]
where the constant in $S_s$ depends on the parameters appearing in \cref{assum:f_freq}.
In particular, the total Hamiltonian simulation time for $H$ 
 \[T_{H,\mathrm{total}}=\Theta(T\tau^{-1}S_s
+T)=\widetilde{\Theta}((1+\|H\|)\Delta^{-1}T^{2}\epsilon^{-1})\,.\]
In addition, the total number of controlled-$A$ gates is
\[
N_{A,\mathrm{gate}}=\Theta(T\tau^{-1}
S_s\tau^{-1}_s)=\widetilde{\Theta}((1+\|H\|)(1+\max\{S_\omega,\|H\|\})\Delta^{-1}T^{2}\epsilon^{-1})\,.
\]
Note that the Hamiltonian simulation time is independent of the step size $\tau_s$ for discretizing the integral in forming the jump operator $K$.
\end{thm}
The proof of \cref{thm:discretization_error_formal} is in Appendix \ref{sec:pf_d_error}.
We emphasize that the bottleneck leading to the first-order accuracy is not due to the Trotter splitting scheme used in \cref{eqn:first_trotter}, but the first-order accuracy of the Lindblad simulation method in \cref{lem:Lindblad_simulation_error}.

\subsection{Simulation cost of discrete-time Lindblad dynamics}\label{sec:sc_al}
The simulation cost of \eqref{eqn:accelerated_lindblad_shorttime} is shown in the following theorem:
\begin{thm}[Simulation cost of the discrete-time Lindblad dynamics]\label{thm:acc_error} Assume that $\hat{f}$ satisfies \cref{assum:f_freq}. Given a stopping time $T>0$, a time step $\tau=\mathcal{O}(\|A\|^{-2})$ such that $M_t=T/\tau\in \mathbb{N}$, we generate $\rho^a_{M_t}$ and $\rho_{M_t}$ using \eqref{eqn:exact_acc_Lind} and \eqref{eqn:accelerated_lindblad_shorttime} correspondingly.

Given the accuracy $\epsilon>0$, to obtain $\|e^{iHS_s}\rho^a_{M_t}e^{-iHS_s}-\rho_{M_t}\|_1\leq \epsilon$, we can choose
\begin{equation}\label{eqn:accelerate_parameter_choice}
S_s=\Theta\left(\frac{1}{\Delta}\left(\log\left(\frac{S_\omega\|A\|T}{\Delta\epsilon}\right)\right)^\alpha\right),\quad \tau_s=\Theta\left(\frac{1}{1+\max\{\|H\|,S_\omega\}}\right),\quad r=\widetilde{\Theta}\left(\frac{\|A\|T^{1/2}}{\epsilon^{1/2}}\right)\,,
\end{equation}
where the constant in $S_s$ depends on the parameters in \cref{assum:f_freq}.
In particular, the total Hamiltonian simulation time for $H$ is
 \[
 T_{H,\mathrm{total}}=\widetilde{\Theta}(T\tau^{-1}S_s
r+T)=\widetilde{\Theta}(\Delta^{-1}T^{3/2}\epsilon^{-1/2})\,,
\] 
and the total number of controlled-$A$ evolution gates is
\[
N_{A,\mathrm{gate}}=\Theta(T\tau^{-1}S_s\tau^{-1}_s
r)=\widetilde{\Theta}\left((1+\max\{S_\omega,\|H\|\})\Delta^{-1}T^{3/2}\epsilon^{-1/2}\right)\,.
\]
\end{thm}
We put the proof in Appendix \ref{sec:pf_acc_error}.

Unlike the simulation of the continuous-time Lindblad dynamics in \cref{thm:discretization_error_formal}, the bottleneck of the second order accuracy in \cref{thm:acc_error} is due to the second order Trotter formula for the short time propagator $W(\sqrt{\tau})$. Replacing the second-order Trotter formula with a $p$-th order Trotter formula in defining $W(\sqrt{\tau})$, we can further improve the asymptotic scaling to be nearly linear in $T$. This is shown in \cref{cor:acc_error_highorder}.

\begin{cor}[Simulation cost of the discrete-time Lindblad dynamics with high order splitting for $W$]\label{cor:acc_error_highorder} Under the same assumptions of \cref{thm:acc_error}, but assume that the propagator $W(\sqrt{\tau}/r)$ is constructed using a $p$-th order Trotter method, we may choose
\[
r=\widetilde{\mathcal{O}}((\|A\|^2T/\epsilon)^{1/p}).
\]
Then with the same choice of $S_s,\tau_s$ as in \cref{thm:acc_error} and sending $p$ to $\infty$, the total Hamiltonian simulation time for $H$ is 
\[
 T_{H,\mathrm{total}}=\Theta(T\tau^{-1}S_s
r+T)=\widetilde{\Theta}(\Delta^{-1}T^{1+o(1)}\epsilon^{-o(1)})\,,
\]
In addition, the total number of controlled-$A$ gates is 
\[
N_{A,\mathrm{gate}}=\Theta(T\tau^{-1}S_s\tau^{-1}_s
r)=\widetilde{\Theta}\left((1+\max\{S_\omega,\|H\|\})\Delta^{-1}T^{1+o(1)}\epsilon^{-o(1)}\right)\,.
\]
\end{cor}

\begin{rem}
While \cref{eqn:single_lindblad_shorttime_2} partitions the coherent and dissipative components of the continuous-time Lindblad dynamics in a manner similar to a first-order Trotter method, this scheme still maintains the ground state. Therefore this first-order-like splitting may influence the mixing time, but may not contribute to the error of the ground state.

On the other hand, the accuracy of $W(\sqrt{\tau})$, which approximates $e^{-i\wt{K}\sqrt{\tau}}$ up to a change of frame, plays a vital role in the fixed point argument and needs to be implemented accurately. In practical scenarios (Sec. Numerics results), we observe that employing a second-order Trotter method with $\tau=\Or(1)$ and $r=\Or(1)$ is often sufficient.
\end{rem}

\subsection{Proof of \texorpdfstring{\cref{lem:Lindblad_simulation_error}}{Lg}}\label{sec:lindblad_single}

To prepare the proof of \cref{thm:discretization_error_formal} and \cref{thm:acc_error}, we first prove \cref{lem:Lindblad_simulation_error}. We recall it in the following:
\begin{lem}[Lindbladian simulation using one ancilla qubit]\label{lem:Lindblad_simulation_error_appendix}
Let 
\[
    \sigma(\tau):=\Tr_a e^{-i \wt{K}\sqrt{\tau}} \left[\ket{0}\bra{0}\otimes\rho \right]e^{i \wt{K}\sqrt{\tau}}\label{eqn:simplied_lindblad_simulation}.
\]
Then for a short time $\tau\ge 0$, 
\[
    \|\sigma(\tau)-\rho(\tau)\|_1=\mathcal{O}(\|K\|^4\tau^2)
\]
where $\rho(t)$ is the solution to the Lindblad dynamics
\[
\partial_t \rho(t)=\mc{L}_K[\rho(t)]\quad \text{with initial condition}\quad \rho(0) = \rho.
\]
\end{lem}
\begin{proof}[Proof of \cref{lem:Lindblad_simulation_error_appendix}]
According to \cref{eqn:matrix_fact}, we first obtain
\[
\left\|\exp\left(-i\wt{K}\sqrt{\tau}\right)-\left(1-i\wt{K}\sqrt{\tau}+\wt{K}^2\tau/2+\wt{K}^3\tau^{3/2}/6\right)\right\|=\mathcal{O}\left(\left|\wt{K}\right|^4\tau^2\right)=\mathcal{O}\left(\|K\|^4\tau^2\right)\,.
\]
Let $T_{3,K}=\left(1-i\wt{K}\sqrt{\tau}+\wt{K}^2\tau/2+\wt{K}^3\tau^{3/2}/6\right)$ denote the truncated Taylor expansion, and $1$ denotes the identity operator. We have
\begin{equation}\label{eqn:first_bound}
\left\|\sigma(\tau)-\Tr_a T_{3,K} \left[\ket{0}\bra{0}\otimes\rho \right]T_{3,K}^\dagger\right\|_1=\mathcal{O}\left(\|K\|^4\tau^2\right)\,.
\end{equation}
Because $\Tr_a (\ket{i}\bra{j}\otimes \rho)=\delta_{i,j}\rho$, We  note that terms with an odd power of $\Or(t^{1/2})$ vanish after applying the partial trace operation. This implies
\begin{equation}\label{eqn:second_bound}
\begin{aligned}
    &\left\|\Tr_a T_{3,K} \left[\ket{0}\bra{0}\otimes\rho \right]T_{3,K}^\dagger-\left(1+K\rho K^\dagger \tau-\frac{1}{2}(K^\dagger K\rho \tau+\rho K^\dagger K \tau)\right)\right\|_1\\
    =&\left\|\Tr_a T_{3,K} \left[\ket{0}\bra{0}\otimes\rho \right]T_{3,K}^\dagger-(1+\mathcal{L}_K\tau)\rho\right\|_1\\=
    &\mathcal{O}\left(\|K\|^4\tau^2\right)
\end{aligned}
\end{equation}
By combining equations \eqref{eqn:first_bound} and \eqref{eqn:second_bound}, we arrive at the following expression:
\[
\|\sigma(\tau)-(1+\mathcal{L}_K\tau)\rho\|_1=\mathcal{O}\left(\|K\|^4\tau^2\right)\,.
\]
Finally, since
\[
\|\rho(\tau)-(1+\mathcal{L}_K\tau)\rho\|_1=\|\exp(\mathcal{L}_K \tau)-(1+\mathcal{L}_K\tau)\rho\|_1=\mathcal{O}\left(\|\mathcal{L}_K\|^2_1\tau^2\right)=\mathcal{O}\left(\|K\|^4\tau^2\right)\,,
\]
we conclude the proof.
\end{proof}

\subsection{Proof of \texorpdfstring{\cref{thm:discretization_error_formal}}{Lg}}\label{sec:pf_d_error}
To prove \cref{thm:discretization_error_formal}, we first show the following proposition:
\begin{prop}\label{prop:discretization_error} Under the assumptions of \cref{thm:discretization_error_formal}, we assume $\tau=\mathcal{O}(\|A\|^{-2})$, then
\begin{equation}\label{eqn:discretization_error}
\begin{aligned}
\left\|e^{iHS_s}\rho(T)e^{-iHS_s}-\rho_{M_t}\right\|_{1}=\widetilde{\mathcal{O}}\left(\|K-K_{s}\|\|A\|T\right)+\widetilde{\mathcal{O}}\left((\|A\|^4+\|A\|^2\|H\|)T\tau\right).
\end{aligned}
\end{equation}
\end{prop}
\begin{proof}[Proof of \cref{prop:discretization_error}] 
We first notice that
\[
\left\|e^{iHS_s}\rho(T)e^{-iHS_s}-\rho_{M_t}\right\|_{1}= \left\|\rho(T)-e^{-iHS_s}\rho_{M_t}e^{iHS_s}\right\|_{1}\,.
\]
Define $\bar{\rho}_m=e^{-iHS_s}\rho_{m}e^{iHS_s}$. Then
\[
\bar{\rho}_{m+1}=e^{-iH\tau} \Tr_a \left(e^{-i H S_s} W(\sqrt{\tau}) e^{i H S_s}\right)\left(\ket{0}\bra{0}\otimes\bar{\rho}_m\right) \left(e^{-i H S_s} W(\sqrt{\tau}) e^{i H S_s}\right)^{\dag} e^{iH\tau} 
\]
with $\bar{\rho}_0=e^{-iHS_s}\rho_{I}e^{iHS_s}$.

Because the Lindblad dynamics is contractive in trace distance according to \cref{eqn:lindblad_contractive}, to prove \eqref{eqn:discretization_error}, it suffices to prove
\begin{equation}\label{eqn:one_step_error}
\begin{aligned}
    \left\|\rho(\tau)-\bar{\rho}_1\right\|_{1}=\widetilde{\mathcal{O}}\left(\|K-K_{s}\|\|A\|\tau\right)+\widetilde{\mathcal{O}}\left(\|A\|^2(\|A\|^2+\|H\|)\tau^2\right)
\end{aligned}
\end{equation}

Define 
\[
\widetilde{\rho}_1=e^{-iH\tau}\left(\Tr_a e^{-i\sqrt{\tau}\wt{K}} \ket{0}\bra{0}\otimes\rho(0) e^{i\sqrt{\tau}\wt{K}}\right)e^{iH\tau}\,,\] 
and \[\wt{\rho}_{1,s}=e^{-iH\tau}\left(\Tr_a e^{-i\sqrt{\tau}\wt{K}_{s}} \ket{0}\bra{0}\otimes\rho(0) e^{i\sqrt{\tau}\wt{K}_{s}}\right)e^{iH\tau}\,.\] 
Then
\begin{equation}\label{eqn:error_separation_2}
\left\|\rho(\tau)-\rho_1\right\|_{1}\leq \left\|\rho(\tau)-\widetilde{\rho}_1\right\|_{1}+\left\|\widetilde{\rho}_1-\wt{\rho}_{1,s}\right\|_{1}+\left\|\wt{\rho}_{1,s}-\bar{\rho}_1\right\|_{1}.
\end{equation}
The first term contains the Lindblad simulation error, the second term contains the error from the numerical integration, and the third term contains the error from Trotter splitting. According to Lemma \ref{lem:Lindblad_simulation_error} and Lemma \ref{lem:operator_fact} \eqref{eqn:operator_fact},
we first bound the Lindblad simulation error:
\begin{equation}\label{eqn:Lindbald_simulation_error}
\begin{aligned}
    &\left\|\rho(\tau)-\widetilde{\rho}_1\right\|_{1}\\
\leq &\left\|\widetilde{\rho}_1-\exp(\mathcal{L}_H\tau)\exp(\mathcal{L}_K\tau)\rho(0)\right\|_{1}+\left\|\exp((\mathcal{L}_H+\mathcal{L}_K)\tau)\rho(0)-\exp(\mathcal{L}_H\tau)\exp(\mathcal{L}_K\tau)\rho(0)\right\|_{1}\\
=&\mathcal{O}\left(\|K\|^4\tau^2\right)+\mathcal{O}\left(\|\left[\mathcal{L}_H,\mathcal{L}_K\right]\|_1\tau^2\right)\\
=&\widetilde{\mathcal{O}}\left(\|A\|^2(\|A\|^2+\|H\|)\tau^2\right)
\end{aligned}
\end{equation}
where we use \cref{lem:Lindblad_simulation_error} in the first equality and \cref{lem:facts} in the last equality. 

Next, in order to bound the error of the numerical integration, we observe that 
\[e^{iH\tau}\widetilde{\rho}_1 e^{-iH\tau}=\Tr_a \left(e^{-i\sqrt{\tau}\wt{K}} \ket{0}\bra{0}\otimes\rho(0) e^{i\sqrt{\tau}\wt{K}}\right)\,,\]
and
\[
e^{iH\tau}\wt{\rho}_{1,s}e^{-iH\tau}=\Tr_a \left(e^{-i\sqrt{\tau}\wt{K}_s} \ket{0}\bra{0}\otimes\rho(0) e^{i\sqrt{\tau}\wt{K}_s}\right) \,.
\]
Define
\[
\widetilde{\rho}_{1,a}(t)=e^{-it\wt{K}} \ket{0}\bra{0}\otimes\rho(0) e^{it\wt{K}},\quad \widetilde{\rho}_{1,s,a}(t)=e^{-it\wt{K}_s} \ket{0}\bra{0}\otimes\rho(0) e^{it\wt{K}_s}\,.
\]
Noticing
\[
\left\|\widetilde{\rho}_{1,s,a}(t)-\widetilde{\rho}_{1,s,a}(0)\right\|_1=\mathcal{O}(\|\wt{K}_s\|t)\,,
\]
we have 
\begin{equation}\label{eqn:numerical_integral_error}
\begin{aligned}
&\|\widetilde{\rho}_1-\wt{\rho}_{1,s}\|_1=\left\|\mathrm{Tr}_a\left(\widetilde{\rho}_{1,a}(\sqrt{\tau})-\widetilde{\rho}_{1,s,a}(\sqrt{\tau})\right)\right\|\\
=&\left\|\mathrm{Tr}_a\left(\int^{\sqrt{\tau}}_0\exp\left(-i\wt{K}(\sqrt{\tau}-t)\right)[\wt{K}-\wt{K}_s,\widetilde{\rho}_{1,s,a}(t)]\exp\left(i\wt{K}(\sqrt{\tau}-t)\right)\ud t\right)\right\|_1\\
\leq& \left\|\mathrm{Tr}_a\left(\int^{\sqrt{\tau}}_0\exp\left(-i\wt{K}(\sqrt{\tau}-t)\right)[\wt{K}-\wt{K}_s,\widetilde{\rho}_{1,s,a}(t)-\widetilde{\rho}_{1,s,a}(0)]\exp\left(i\wt{K}(\sqrt{\tau}-t)\right)\ud t\right)\right\|_1\\
&+\left\|\mathrm{Tr}_a\left(\int^{\sqrt{\tau}}_0\exp\left(-i\wt{K}(\sqrt{\tau}-t)\right)[\wt{K}-\wt{K}_s,\widetilde{\rho}_{1,s,a}(0)]\exp\left(i\wt{K}(\sqrt{\tau}-t)\right)\ud t\right)\right\|_1\\
\leq& \left\|e^{-i\sqrt{\tau}\wt{K}} \ket{0}\bra{0}\otimes\rho(0) e^{i\sqrt{\tau}\wt{K}}-e^{-i\sqrt{\tau}\wt{K}_s} \ket{0}\bra{0}\otimes\rho(0) e^{i\sqrt{\tau}\wt{K}_s}\right\|_1\\
=&\mathcal{O}\left(\left\|\wt{K}-\wt{K}_{s}\right\|\left(\left\|\wt{K}_{s}\right\|+\left\|\wt{K}\right\|\right)\tau\right)\\
=&\widetilde{\mathcal{O}}\left(\left\|K-K_{s}\right\|\|A\|\tau\right)\\
\end{aligned}\,,
\end{equation}
where we use the Duhamel principle in the second equality, $\mathrm{Tr}_a([\wt{K}-\wt{K}_s,\widetilde{\rho}_{1,s,a}(0)])=0$ in the second inequality, and $\|K_{s}\|=\widetilde{\mathcal{O
}}(\|A\|)$ in the last equality. 

Finally, we bound the error of Trotter splitting. Using \eqref{eqn:matrix_fact} to expand each $\exp\left(-i\tilde{H}_l\sqrt{\tau}/2\right)$ up to $N=3$, we obtain that:
\begin{equation}\label{eqn:trotter_error}
\begin{aligned}
\left\|\wt{\rho}_{1,s}-\bar{\rho}_1\right\|_{1}=&\mathcal{O}\left(\left(\sqrt{\tau}\sum_{l}\|H_l\|\right)^4\right)=\mathcal{O}\left(\tau^2\|A\|^4\left(\tau_s\sum_{l}|f(s_l)|\right)^4\right)=\widetilde{\mathcal{O}}\left(\tau^2\|A\|^4\right)
\end{aligned}
\end{equation}
where we use \cref{lem:as_f_simulation} \eqref{eqn:f_L_1} in the third equality. Here all $\tau^{3/2}$ terms disappear because of the partial trace. Plugging \eqref{eqn:Lindbald_simulation_error}, \eqref{eqn:numerical_integral_error}, and \eqref{eqn:trotter_error} into \eqref{eqn:error_separation_2}, we prove \eqref{eqn:one_step_error}.
\end{proof}

To bound $\|K-K_{s}\|$ in \eqref{eqn:discretization_error}, we need to use \cref{lem:trapezoidal_error}. We state its formal version here:
\begin{lem}[Convergence of the quadrature error]
\label{lem:trapezoidal_error_appendix} 
Let $\hat{f}$ satisfy \cref{assum:f_freq}. Then, for any $\epsilon>0$, if 
\[
S_s=\Omega\left(\frac{1}{\Delta}\left(\frac{\alpha}{C_{2,f}}\log\left(\frac{C_{1,f}S_\omega\|A\|\alpha}{C_{2,f}\Delta \epsilon'}\right)\right)^{\alpha}\right),\quad \tau_s=\mathcal{O}\left(\frac{1}{1+\max\{\|H\|,S_\omega\}}\right)\,,\]
where $C_{1,f},C_{2,f}$ come from \cref{lem:as_f_simulation}. 
we have
\begin{equation}\label{eqn:error_quadrature}
    \|K-K_{s}\|=\mathcal{O}(\epsilon')\,.
\end{equation}
\end{lem}

\begin{proof}[Proof of Lemma \ref{lem:trapezoidal_error_appendix}]  We separate the error into two parts:
\begin{equation}\label{eqn:error_separation}
\|K-K_{s}\|\leq \|K-K_\infty\|+\left\|K_\infty-K_{s}\right\|\,,
\end{equation}
where $K_\infty=\sum^\infty_{l=-\infty}f(s_l)e^{iH s_l}Ae^{-iHs_l} \tau_s$. For a given $N$, using \cref{lem:as_f_simulation} \eqref{eqn:f_super_decay}, the second part of error can be bounded as
\begin{equation}\label{eqn:dis_infty_M_s}
\begin{aligned}
&\left\|K_\infty-K_{s}\right\|\\
=&\mathcal{O}\left( \|A\|\sum_{|l|\geq M_s}C_{1,f}S_\omega\exp\left(-C_{2,f}|s_l\Delta|^{1/\alpha}\right)\tau_s\right)\\
=&\mathcal{O}\left(C_{1,f}S_\omega\|A\|\int^\infty_{S_s}\exp\left(-C_{2,f}|s\Delta|^{1/\alpha}\right)\ud s\right)\\
=&\mathcal{O}\left(\frac{C_{1,f}S_\omega\|A\|}{\Delta}\int^\infty_{S_s\Delta}\exp\left(-C_{2,f}|s|^{1/\alpha}\right)\ud s\right)\\
=&\mathcal{O}\left(\frac{C_{1,f}S_\omega\|A\|\alpha}{\Delta}\int^\infty_{\left(S_s\Delta\right)^{1/\alpha}}s^{\alpha-1}\exp\left(-C_{2,f}s\right)\ud s\right)\\
=&\mathcal{O}\left(\frac{C_{1,f}S_\omega\|A\|\alpha}{\Delta}\int^\infty_{\left(S_s\Delta\right)^{1/\alpha}}\exp\left(-C_{2,f}s/2\right)\ud s\right)\\
=&\mathcal{O}\left(\frac{2C_{1,f}S_\omega\|A\|\alpha}{C_{2,f}\Delta}\exp\left(-\frac{C_{2,f}\left(S_s\Delta\right)^{1/\alpha}}{2}\right)\right)\\
=&\mathcal{O}\left(\epsilon'\right),  
\end{aligned}
\end{equation}
where we use $s^{\alpha-1}<\exp(C_{2,f}s/2)$in the fifth equality and  $S_s=\Omega\left(\frac{1}{\Delta}\left(\frac{\alpha}{C_{2,f}}\log\left(\frac{C_{1,f}S_\omega\|A\|\alpha}{C_{2,f}\Delta \epsilon'}\right)\right)^{\alpha}\right)$ in the last equality.
\\
Now, we control $\|K-K_\infty\|$. Recall $\mathrm{supp}\left(\hat{f}\right)\subset [-S_\omega,0]\subset [-\max\{2\|H\|,S_\omega\},\max\{2\|H\|,S_\omega\}]$. Applying \cref{thm:TW_bound} with $S=\max\{2\|H\|,S_\omega\}$ and $\tau_s=\mathcal{O}(1/S)$, we obtain 
\[
K=\left(\sum_{i,j}\int^\infty_{-\infty}f(t)\exp(i(\lambda_i-\lambda_j)t)\ud t\bra{i}A\ket{j}\right)\ket{i}\bra{j}=\left(\sum_{i,j}\left(\sum^\infty_{l=-\infty}f(s_l)\exp(i(\lambda_i-\lambda_j)s_l) \tau_s\right)\bra{i}A\ket{j}\right)\ket{i}\bra{j}=K_\infty\,.
\]
This implies $K=K_\infty$ and concludes the proof.
\end{proof}

Now, we are ready to prove \cref{thm:discretization_error_formal}:

\begin{proof}[Proof of \cref{thm:discretization_error_formal}]
    According to \cref{prop:discretization_error} \eqref{eqn:discretization_error} and \cref{lem:trapezoidal_error}, we set $\epsilon'=\frac{\epsilon}{\|A\|T}$ in $\eqref{eqn:error_quadrature}$ to obtain that
    \[
    \left\|\rho(T)-\rho_{M_t}\right\|_{1}=\frac{\epsilon}{2}+\widetilde{\mathcal{O}}\left((\|A\|^4+\|A\|^2\|H\|)T\tau\right)\leq \epsilon\,,
    \]
    where we use $\tau=\widetilde{\mathcal{O}}\left(\frac{\epsilon}{(\|A\|^4+\|A\|^2\|H\|)T}\right)$ in the inequality. 

    Next, we calculate the total Hamiltonian simulation time for $H$. According to \eqref{eqn:total_H_1}, each implementation of $W(\tau)$ needs to simulate the system Hamiltonian for time $\Theta(S_s)$. Thus, the total Hamiltonian simulation time:
    \[
     T_{H,\mathrm{total}}=\text{number of steps}\times (\tau+\Theta(S_s))=M_t(\tau+\Theta(S_s))=\Theta(T+TS_s
    \tau^{-1})=\widetilde{\Theta}(T^{2}\epsilon^{-1})\,.
    \]
    
    Finally, to calculate the number of controlled-$A$ evolution gates, we notice that each trotter splitting step in $W(\tau)$ needs to implement $\Theta(1)$ controlled-$A$ evolution gates. This implies 
    \[
     N_{A,\mathrm{gate}}=\text{number of steps}\times \Theta(S_s\tau^{-1}_s)=\widetilde{\Theta}(T^{2}\epsilon^{-1})\,.
    \]
    
\end{proof}

\subsection{Proof of \texorpdfstring{\cref{thm:acc_error}}{Lg} and \texorpdfstring{\cref{cor:acc_error_highorder}}{Lg}}\label{sec:pf_acc_error}
To prove \cref{thm:acc_error}, we first show the following proposition:
\begin{prop}\label{prop:acc_discretization_error} Under the conditions of \cref{thm:acc_error}, we further assume $r=\Omega(1)$, then
\begin{equation}\label{eqn:acc_discretization_error}
\begin{aligned}
\left\|e^{iHS_s}\rho^a_{M_t}e^{-iHS_s}-\rho_{M_t}\right\|_{1}=\widetilde{\mathcal{O}}\left(\|K-K_{s}\|\|A\|T\right)+\widetilde{\mathcal{O}}\left(\|A\|^4T\tau/r^2\right).
\end{aligned}
\end{equation}
\end{prop}
\begin{proof}[Proof of \cref{prop:acc_discretization_error}] 
 
For simplicity, we ignore the effect of $e^{iHS_s}$. Because the Lindblad dynamics is contractive in trace distance  according to \cref{eqn:lindblad_contractive}, to prove \eqref{eqn:acc_discretization_error}, it suffices to prove
\begin{equation}\label{eqn:one_step_error_acc}
\begin{aligned}
    \left\|\rho^a_1-\rho_1\right\|_{1}=\widetilde{\mathcal{O}}\left(\|K-K_{s}\|\|A\|\tau\right)+\widetilde{\mathcal{O}}\left(\|A\|^4\tau^2/r^2\right).
\end{aligned}
\end{equation}

Define 
\[
\wt{\rho}_{1,s}=e^{-iH\tau}\left(\Tr_a e^{-i\sqrt{\tau}\wt{K}_s} \ket{0}\bra{0}\otimes\rho(0) e^{i\sqrt{\tau}\wt{K}_s}\right)e^{iH\tau}\,.\] 
Then
\begin{equation}\label{eqn:error_separation_2_acc}
\left\|\rho^a_1-\rho_1\right\|_{1}\leq \left\|\rho^a_1-\wt{\rho}_{1,s}\right\|_{1}+\left\|\wt{\rho}_{1,s}-\rho_1\right\|_{1}\,,
\end{equation}
where the first term contains the error from the numerical integration, and the third term contains the error from Trotter splitting. Similar to \cref{eqn:numerical_integral_error} in the proof of Theorem \ref{thm:discretization_error_formal}, for the first term, we can bound is by: 
\begin{equation}\label{eqn:numerical_integral_error_acc}
\left\|\rho^a_1-\wt{\rho}_{1,s}\right\|_{1}\leq \widetilde{\mathcal{O}}(\|K-K_s\|\|A\|\tau)\,.
\end{equation}

To bound the second term, we notice
\begin{equation}\label{eqn:error_first_acc}
\begin{aligned}
    &\left\|\wt{\rho}_{1,s}-\rho_1\right\|_{1}\\
    \leq &\left\|\Tr_a \left(e^{-i\wt{K}_s\sqrt{\tau}} \ket{0}\bra{0}\otimes\rho(0) e^{i\wt{K}_s\sqrt{\tau}}-[W(\sqrt{\tau}/r)]^r \left(\ket{0}\bra{0}\otimes\rho(0)\right)  [(W(\sqrt{\tau}/r))^{\dag}]^r\right)\right\|_1.
\end{aligned}
\end{equation}
Using the second-order Trotter splitting formula and \cref{eqn:matrix_fact} with $N=3$, we can rewrite $W(\sqrt{\tau}/r)$ as 
\begin{equation}\label{eqn:W_eta_r}
W(\sqrt{\tau}/r)=e^{-i\sqrt{\tau}\wt{K}_s/r}+\frac{\tau^{3/2}}{r^3}E_1+E_2(\tau,r)\,.
\end{equation}
Here the derivation of the expression 
\[
E_1=\sum_{l_1,l_2,l_3}a_{l_1,l_2,l_3}\wt{H}_{l_1}\wt{H}_{l_1}\wt{H}_{l_3}
\]
is similar to that in \cref{eqn:secondordertaylor_collect}. Moreover,
\[
\|E_2(\tau,r)\|=\mathcal{O}\left(\frac{\tau^2}{r^4}\left(\sum_{l}\|H_l\|\right)^4\right)=\mathcal{O}\left(\frac{\tau^2}{r^4}\|A\|^4\left(\tau_s\sum_{l}|f(s_l)|\right)^4\right)=\widetilde{\mathcal{O}}\left(\frac{\|A\|^4\tau^2}{r^4}\right)\,.
\]
Here, in order to obtain the correct leading order error, we need to expand $W(\sqrt{\tau}/r)$ up to $\mathcal{O}(\tau^2)$.

Plugging \eqref{eqn:W_eta_r} into \eqref{eqn:error_first_acc}, we have
\[
\begin{aligned}
    &\left\|\wt{\rho}_{1,s}-\rho_1\right\|_{1}\\
    \leq &\sum^r_{k=1}\left\|  \Tr_a \left(\frac{\tau^{3/2}}{r^3}e^{-i(k-1)\sqrt{\tau}\wt{K}_s/r}E_1e^{-i(r-k)\sqrt{\tau}\wt{K}_s/r}\ket{0}\bra{0}\otimes\rho(0) e^{i\sqrt{\tau}\wt{K}_s}\right)\right\|_1\\
+&\sum^r_{k=1}\left\|  \Tr_a \left(\frac{\tau^{3/2}}{r^3}e^{-i\sqrt{\tau}\wt{K}_s}\ket{0}\bra{0}\otimes\rho(0) e^{i(k-1)\sqrt{\tau}\wt{K}_s/r}E^\dagger_1e^{i(r-k)\sqrt{\tau}\wt{K}_s/r}\right)\right\|_1\\
+&\widetilde{\mathcal{O}}\left(\frac{\tau^2\|A\|^4}{r^4}\right)\\
\end{aligned}
\]
Here, we use the relationship $r=\Omega(1)$ and $\tau=\mathcal{O}(S^{-2}_\Omega\|A\|^{-2})$ to incorporate all terms of the form $\left(\frac{|\tau|^{p/2}\|A\|^p}{r^p}\right)$ (for $p\ge 4$) into the asymptotic notation  $\widetilde{\mathcal{O}}\left(\frac{\tau^2\|A\|^4}{r^4}\right)$.

Next, since $\Tr_a \left(E_1\ket{0}\bra{0}\otimes\rho(0) \right)=0$, we can expand $e^{i\sqrt{\tau}\wt{K}_s}$ to first order and obtain the refined estimate
\[
    \left\|\Tr_a \left(\frac{\tau^{3/2}}{r^3}e^{-i(k-1)\sqrt{\tau}\wt{K}_s/r}E_1e^{-i(r-k)\sqrt{\tau}\wt{K}_s/r}\ket{0}\bra{0}\otimes\rho(0) e^{i\sqrt{\tau}\wt{K}_s}\right)\right\|_1=\widetilde{\mathcal{O}}\left(\frac{\tau^2\|A\|^4}{r^3}\right)
\]
for all $1\leq k\leq r$. Here we use $\norm{E_1}=\Or(\norm{A}^3)$, and $\norm{K}=\Or(\norm{A})$. 
Therefore
\begin{equation}
\sum^r_{k=1}\left\|  \Tr_a \left(\frac{\tau^{3/2}}{r^3}e^{-i(k-1)\sqrt{\tau}\wt{K}_s/r}E_1e^{-i(r-k)\sqrt{\tau}\wt{K}_s/r}\ket{0}\bra{0}\otimes\rho(0) e^{i\sqrt{\tau}\wt{K}_s}\right)\right\|_1=\widetilde{\mathcal{O}}\left(\frac{\tau^2\|A\|^4}{r^2}\right).
\end{equation}

Similarly using $\Tr_a \left(\ket{0}\bra{0}\otimes\rho(0)E^\dagger_1 \right)=0$, we have
\begin{equation}
\sum^r_{k=1}\left\|  \Tr_a \left(\frac{\tau^{3/2}}{r^3}e^{-i\sqrt{\tau}\wt{K}_s}\ket{0}\bra{0}\otimes\rho(0) e^{i(k-1)\sqrt{\tau}\wt{K}_s/r}E^\dagger_1e^{i(r-k)\sqrt{\tau}\wt{K}_s/r}\right)\right\|_1\\
=\widetilde{\mathcal{O}}\left(\frac{\tau^2\|A\|^4}{r^2}\right)
\end{equation}
This gives
\[
\left\|\wt{\rho}_{1,s}-\rho_1\right\|_{1}=\widetilde{\mathcal{O}}\left(\frac{\tau^2\|A\|^4}{r^2}\right)\,.
\]
Plugging this equality and \eqref{eqn:numerical_integral_error_acc} into \eqref{eqn:error_separation_2_acc}, we prove \eqref{eqn:one_step_error_acc}.
\end{proof}

Now, we are ready to prove \cref{thm:acc_error}:

\begin{proof}[Proof of \cref{thm:acc_error}]
    We note that \cref{lem:trapezoidal_error} also holds for this case. Thus, we apply \cref{prop:acc_discretization_error} \eqref{eqn:acc_discretization_error} and \cref{lem:trapezoidal_error} by setting $\epsilon'=\frac{\epsilon}{\|A\|T}$ in $\eqref{eqn:error_quadrature}$ to obtain that
    \[
    \left\|\rho(T)-\rho_{M_t}\right\|_{1}=\frac{\epsilon}{2}+\widetilde{\mathcal{O}}\left(\|A\|^4T\tau/r^2\right)\leq \epsilon\,,
    \]
    where we use $r=\widetilde{\Theta}(\|A\|T^{1/2}\epsilon^{-1/2})$ and $\tau=\widetilde{\mathcal{O}}(\|A\|^{-2})$ in the inequality.

    Next, we calculate the total Hamiltonian simulation time for $H$. According to \eqref{eqn:total_H_1}, the implementation of each $W(\tau/r)$ needs to simulate the system Hamiltonian for time $\Theta(S_s)$. Thus, the total Hamiltonian simulation time:
    \[
     T_{H,\mathrm{total}}=\text{number of steps}\times (\tau+r\Theta(S_s))=\Theta(T+rTS_s
    \tau^{-1})=\widetilde{\Theta}(T^{3/2}\epsilon^{-1/2})\,.
    \]
        Finally, each step within the Trotter-splitting process of $W(\tau/r)$ requires the implementation of $\Theta(1)$ controlled-$A$ evolution gates.
This implies 
    \[
     N_{A,\mathrm{gate}}=\text{number of steps}\times \Theta(S_s\tau^{-1}_sr)=\widetilde{\Theta}(T^{3/2}\epsilon^{-1/2})\,.
    \]
\end{proof}

\begin{proof}[Proof of \cref{cor:acc_error_highorder}] We note that \cref{lem:trapezoidal_error} also holds for this case. In addition, when applying $p$-th order Trotter scheme, we obtain
\begin{equation}\label{eqn:acc_discretization_error_2}
\begin{aligned}
\left\|e^{iHS_s}\rho^a_{M_t}e^{-iHS_s}-\rho_{M_t}\right\|_{1}=\widetilde{\mathcal{O}}\left(\|K-K_{s}\|\|A\|T\right)+\widetilde{\mathcal{O}}\left(\|A\|^{p+2}T\tau^{p/2}/r^p\right)\,.
\end{aligned}
\end{equation}
Then, we apply \eqref{eqn:acc_discretization_error_2} and \cref{lem:trapezoidal_error} by setting $\epsilon'=\frac{\epsilon}{\|A\|T}$ in $\eqref{eqn:error_quadrature}$ to obtain that
\[
\left\|\rho(T)-\rho_{M_t}\right\|_{1}=\frac{\epsilon}{2}+\widetilde{\mathcal{O}}\left(\|A\|^{p+2}T\tau^{p/2}/r^p\right)\leq \epsilon\,,
\]
where we use $r=\widetilde{\Theta}(\|A\|^{2/p}T^{1/p}\epsilon^{-1/p})$ and $\tau=\widetilde{\mathcal{O}}(\|A\|^{-2})$ in the inequality. 

    Similar to the previous proof, the total Hamiltonian simulation time is
    \[
     T_{H,\mathrm{total}}=\text{number of steps}\times (\tau+r\Theta(S_s))=\Theta(T+rTS_s
    \tau^{-1})=\widetilde{\Theta}(T^{1+1/p}\epsilon^{-1/p})
    \]
    and the total number of controlled-$A$ gates is 
    \[
     N_{A,\mathrm{gate}}=\text{number of steps}\times \Theta(S_s\tau^{-1}_sr)=\widetilde{\Theta}(T^{1+1/p}\epsilon^{-1/p})\,.
    \]
    This concludes the proof.
\end{proof}

\section{Convergence under ETH type ansatz}\label{sec:conv}

In this section, we focus on the convergence property of the continuous-time Lindblad dynamics \eqref{eqn:Lindblad_dynamics}. We will first introduce the mixing time result of the dynamics in \cref{thm:converge_Lindblad_K}. We then give the proof of \cref{thm:fixed_point} and \cref{thm:converge_Lindblad_K} in \cref{sec:pf_conv_lindblad}. Finally, we show that our continuous Lindblad
dynamics simulation can approximate the expectation in \cref{sec:random_coupling_conv} when the time step $\tau$ is small enough.

To study the mixing time of the continuous-time Lindblad dynamics~\eqref{eqn:Lindblad_dynamics}, we define $\textbf{p}(t)=(p_0(t),p_1(t),\dots,p_{N-1}(t))^\top$. In \cref{sec:pf_conv_lindblad}, we can show that the solution of \eqref{eqn:Lindblad_dynamics} satisfies 
\begin{equation}
\frac{\ud\textbf{p}(t)}{\ud t}=\mathbf{T}\mathbf{p}(t), \quad \mathbb{E}(\rho(t))=\sum^{N-1}_{i=0}p_i(t)\ket{\psi_i}\left\langle \psi_i\right|.
\end{equation}
Here the transition matrix elements are
\begin{equation}\label{eqn:Transition_matrix}
\mathbf{T}_{j,i}=\hat{f}^2_{j,i}\sigma_{j,i},\quad i\neq j,\quad \mbox{and}\quad  \mathbf{T}_{i,i}=-\sum^{N-1}_{j\neq i}\hat{f}^2_{j,i}\sigma_{j,i}.
\end{equation}

We are now prepared to demonstrate that the solution of the continuous-time Lindblad dynamics \eqref{eqn:Lindblad_dynamics} rapidly converges to low-energy states, as proven by the following theorem:
\begin{thm}[Polynomial mixing time]\label{thm:converge_Lindblad_K} Together with Assumption \ref{assump:A}, we also assume that there exists a decreasing sequence $\{R_l\}^{L}_{l=1}$ with $L=\mathcal{O}(\mathrm{poly}(n))$ such that
\begin{itemize}
\item $R_1=N-1$, and $R_L=\mathcal{O}(\mathrm{poly}(n))$.
\item For each $1\leq l\leq L-1$
\begin{equation}\label{eqn:weight_out}
\sum^{R_{l+1}}_{i=0}\hat{f}^2_{i,j}\sigma_{i,j}=\Omega(1/\mathrm{poly}(n)),\quad \text{for all}\  j\in (R_{l+1},R_l].
\end{equation}
\end{itemize}
Then, there exists $T^\star=\mathcal{O}(\mathrm{poly}(n))$ such that
\begin{equation}\label{eqn:convergence}
\mathbb{E}\left(\sum^{R_L}_{i=0}\bra{\psi_i}\rho(t)\ket{\psi_i}\right)=\sum^{R_L}_{i=0}p_i(t)=\Omega(1),\quad \text{for all}\  t>T^\star\,.
\end{equation}
for any initial condition $\textbf{p}(0)$ such that $p_i(0)\geq 0$ and $\sum^N_{i=1}p_i(0)=1$.
\end{thm}
We would like to emphasize that condition \eqref{eqn:weight_out} is crucial in ensuring that the weight in $p_{R_{l+1}+1\leq i\leq R_{l}}$ can be efficiently transferred to $p_{i\leq R_{l}+1}$ within $\mathcal{O}(\mathrm{poly}(n))$ time. Since this transfer only needs to occur $L=\mathcal{O}(\mathrm{poly}(n))$ times, the total time required for the weight translation from $p_{R_{L}+1\leq i}$ to $p_{i\leq R_{L}}$ remains $\mathcal{O}(\mathrm{poly}(n))$. To provide a clear visual representation of the transition matrix's structure, we have included its graph in \eqref{fig:T}.

The theorem above implies that, given appropriate assumptions, the continuous-time Lindblad dynamics \eqref{eqn:Lindblad_dynamics} can rapidly drive the quantum state towards low-energy eigenstates, resulting in an increased overlap with them. Now, if we assume that the mixing time of the sub-transition matrix $\textbf{T}_{(1:R_L,1:R_L)}\in \mathbb{R}^{R_L\times R_L}$ is $t^{\mathrm{sub}}_{\mathrm{mix}}=\mathcal{O}(1/\mathrm{poly}(n))$, then the mixing time for preparing the ground state is $t_{\mathrm{mix}}=\mathcal{O}(\mathrm{poly}(n))$. In other words, when $t>t_{\mathrm{mix}}$, the overlap with the ground state $\mathbb{E}\left(\bra{\psi_0}\rho(t)\ket{\psi_0}\right)=a_0(s)=\Omega(1)$.

\begin{figure}[t]
\includegraphics[width=8cm]{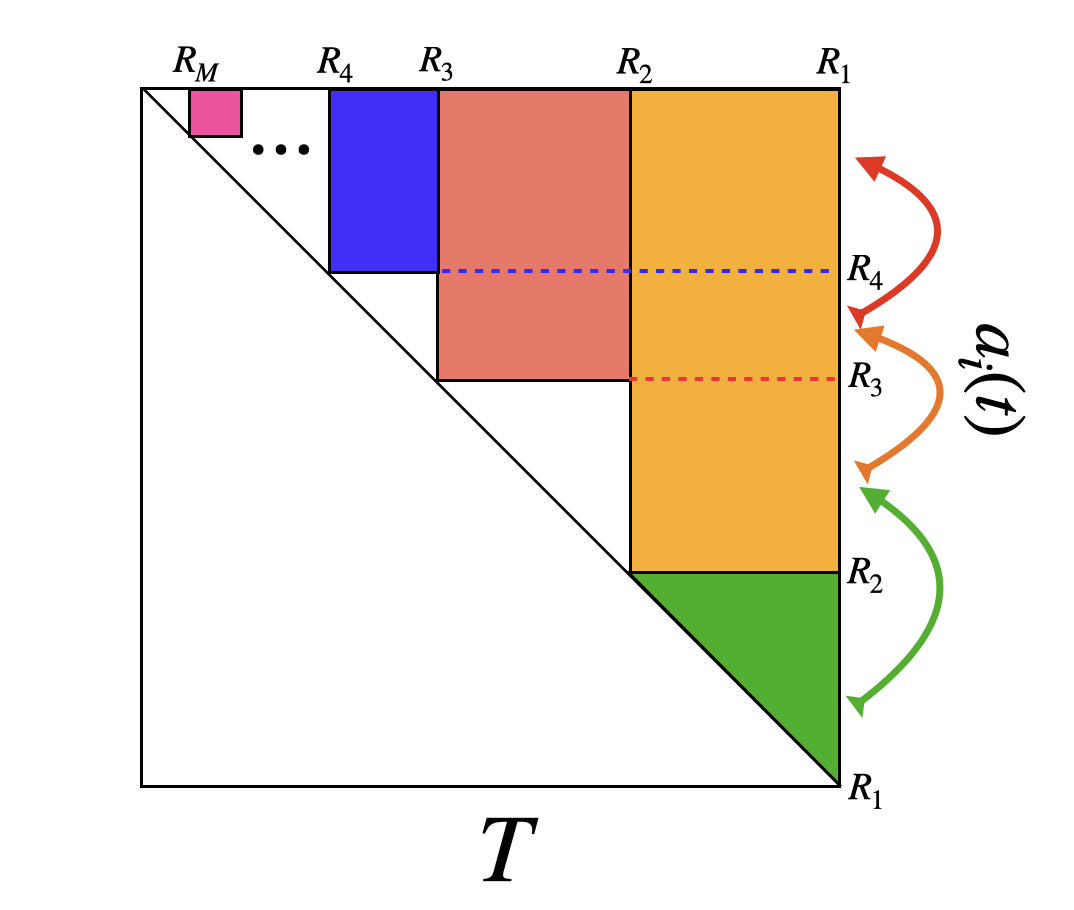}
\centering
\caption{Structure of the transition matrix $T$ in the eigenbasis of the Hamiltonian.}
\label{fig:T}
\end{figure}

\subsection{Proof of \texorpdfstring{\cref{thm:fixed_point} and \cref{thm:converge_Lindblad_K}}{Lg}}\label{sec:pf_conv_lindblad}
\begin{proof}[Proof of Theorem \ref{thm:fixed_point}] 
Since $\hat{f}(\omega)=0$ for $\omega\geq 0$, we have $\mc{L}_K[\ket{\psi_0}\bra{\psi_0}]=0$. Plugging this into \eqref{eqn:Lindblad_dynamics}, we prove that 
$\rho^\star=\ket{\psi_0}\bra{\psi_0}$ is a fixed point of \eqref{eqn:Lindblad_dynamics}.

Now, we assume $A,\rho(0)$ satisfy Assumption \ref{assump:A}. To track the evolution of $\mathbb{E}(\rho(t))$, we first observe,
\[
K\ket{\psi_i}=\sum^{N-1}_{j=0} \hat{f}_{j,i}A_{j,i}\ket{\psi_j},\quad K^\dagger\ket{\psi_i}=\sum^{N-1}_{j=0} \hat{f}_{i,j}A_{j,i}\ket{\psi_j}\,,
\]
where we use $A$ as a Hermitian matrix in the second equality. Taking the expectation on the randomness of $A$, we obtain
\[
\mathbb{E}\left(K\ket{\psi_i}\left\langle \psi_i\right|K^\dagger\right)=\sum^{N-1}_{j=0} \hat{f}^2_{j,i}\sigma_{j,i}\ket{\psi_j}\left\langle \psi_j\right|\,,
\]
and
\[
\begin{aligned}
\mathbb{E}\left(K^\dagger K\ket{\psi_i}\left\langle \psi_i\right|\right)&=\mathbb{E}\left(\sum^{N-1}_{j=0}\hat{f}_{j,i}A_{j,i}K^\dagger \ket{\psi_j}\left\langle \psi_i\right|\right)\\
&=\mathbb{E}\left(\sum^{N-1}_{j,k=0}\hat{f}_{j,i}A_{j,i}\hat{f}_{j,k}A_{k,j}\ket{\psi_k}\left\langle \psi_i\right|\right)\\
&=\sum^{N-1}_{j=0} \hat{f}^2_{j,i}\sigma_{j,i}\ket{\psi_i}\left\langle \psi_i\right|\,.
\end{aligned}
\]
These two equalities imply that
\begin{equation}\label{eqn:lindbload_psi}
\mathbb{E}\left(\mathcal{L}_K\left(\ket{\psi_i}\left\langle \psi_i\right|\right)\right)=\sum^{N-1}_{j\neq i} \hat{f}^2_{j,i}\sigma_{j,i}\left(\ket{\psi_j}\left\langle \psi_j\right|-\ket{\psi_i}\left\langle \psi_i\right|\right)\,.
\end{equation}

Since $\rho(0)$ is a diagonal matrix in the basis of $\{\ket{\psi_i}\}^{N-1}_{i=0}$, plugging \eqref{eqn:lindbload_psi} into \eqref{eqn:Lindblad_dynamics} and taking expectation on both sides, we find that $\mathbb{E}(\rho(t))$ is always a diagonal matrix in the basis of $\{\ket{\psi_i}\}^{N-1}_{i=0}$ and 
\[
\frac{\ud\mathbb{E}(\rho(t))}{\ud t}=\mathbb{E}\left(\mathcal{L}_K\left(\rho(t)\right)\right)\,,
\]
where we use $[H,\mathbb{E}(\rho(t))]=0$ for any $t$. According to \eqref{eqn:lindbload_psi}, we can rewrite the above equation as
\[
\mathbb{E}(\rho(t))=\sum^{N-1}_{i=0}p_i(t)\ket{\psi_i}\left\langle \psi_i\right|\,
\]
where $p_i(t)$ solves
\begin{equation}\label{eqn:evolution_a_n}
\frac{\ud p_i(t)}{\ud t}=\sum^{N-1}_{j\neq i} \hat{f}^2_{i,j}\sigma_{i,j}p_{j}(t)-\sum^{N-1}_{j\neq i} \hat{f}^2_{j,i}\sigma_{j,i}p_i(t)\,.
\end{equation}
In particular,
\[
\frac{\ud p_0(t)}{\ud t}=\sum^{N-1}_{j=1} \hat{f}^2_{0,j}\sigma_{0,j}p_{j}(t)\,,
\]
which implies $p_0(t)=1,\ p_{i\geq 1}(t)=0$ is the unique fixed point. Since $\sigma_{i,j}>0$ for all $i,j$, we conclude that $\ket{\psi_0}\bra{\psi_0}$ is the unique fixed point of the Lindblad dynamics~\eqref{eqn:Lindblad_dynamics}, and the solution of~\eqref{eqn:Lindblad_dynamics} converges to $\ket{\psi_0}\bra{\psi_0}$ as $t\rightarrow\infty$.

\end{proof}

Next, we prove \cref{thm:converge_Lindblad_K}.

\begin{proof}[Proof of Theorem \ref{thm:converge_Lindblad_K}] 
To show \eqref{eqn:convergence}, it suffices to prove that for any $1\leq l\leq L-1$, there exists $T_l=\mathcal{O}(\mathrm{poly}(n))$ such that 
\begin{equation}\label{eqn:iteration}
\sum^{N-1}_{i=R_{l+1}+1} p_i(t)\leq \frac{1}{2}-\frac{1}{l+3},\quad \text{for all}\  t>T_l\,.
\end{equation}

Since $\hat{f}(x)=0$ when $x\geq 0$, we first notice that $\textbf{T}$ is an upper triangular matrix, which implies the weight can only move from high energy states to low energy states. Then, for $l=1$, using \eqref{eqn:weight_out}, we obtain
\[
\frac{\ud \sum^{N-1}_{i=R_{2}+1} p_i(t)}{\ud t}\leq -\left(\min_{j\in (R_{2},R_1]}\sum^{R_{2}}_{i=0}\hat{f}^2_{i,j}\sigma_{i,j}\right)\left(\sum^{N-1}_{i=R_{2}+1} p_i(t)\right)=\mathcal{O}\left(-\frac{1}{\mathrm{poly}(n)}\left(\sum^{N-1}_{i=R_{2}+1} p_i(t)\right)\right)\,,
\]
Since $\sum^{N-1}_{i=R_{2}+1} p_i(0)\leq 1$, there exists $T_1=\mathcal{O}(\mathrm{poly}(n))$ such that
\begin{equation}\label{eqn:first_iteration}
\sum^{N-1}_{i=R_{2}+1} p_i(t)<\frac{1}{4},\quad \text{for all}\  t>T_1\,.
\end{equation}
Next, for $l=2$, according \eqref{eqn:weight_out}, we obtain 
\[
\frac{\ud \sum^{N-1}_{i=R_{3}+1} p_i(t)}{\ud t}\leq -\left(\min_{j\in (R_{3},R_2]}\sum^{R_{3}}_{i=0}\hat{f}^2_{i,j}\sigma_{i,j}\right)\left(\sum^{R_2}_{i=R_{3}+1} p_i(t)\right)=\mathcal{O}\left(-\frac{1}{\mathrm{poly}(n)}\left(\sum^{R_2}_{i=R_{3}+1} p_i(t)\right)\right)\,,
\]
where we use $\min_{j\in (R_{3},R_2]}\sum^{R_{3}}_{i=0}\hat{f}^2_{i,j}\sigma_{i,j}=\Omega(\frac{1}{\mathrm{poly}(n)})$ in the inequality. This implies that when $T>T_1$, we have
\[
\frac{\ud \sum^{N-1}_{i=R_{3}+1} p_i(t)}{\ud t}\leq \left\{
\begin{aligned}
    &\mathcal{O}\left(-\frac{1}{20}\frac{1}{\mathrm{poly}(n)}\right),\quad \text{if}\quad \sum^{R_2}_{i=R_{3}+1} p_i(t)>\frac{1}{20}\\
    &0,\quad \text{otherwise}\\
\end{aligned}
\right.\,.
\]
Because $\sum^{R_2}_{i=R_{3}+1} p_i(t)\leq \sum^{N-1}_{i=R_{3}+1} p_i(t)$, there exists $T_2=\max\left\{T_1,\mathcal{O}(\mathrm{poly}(n))\right\}$ such that
\[
\sum^{R_2}_{i=R_{3}+1} p_i(t)<\frac{1}{1+3}-\frac{1}{2+3}=\frac{1}{20},\quad \text{for all}\  t>T_2\,.
\]
Combining this with \eqref{eqn:first_iteration}, we can obtain that
\[
\sum^{N-1}_{i=R_{3}+1} p_i(t)<\frac{3}{10}=\frac{1}{2}-\frac{1}{2+3},\quad \text{for all}\  t>T_2\,.
\]
We can continue this process to any $l\leq L-1$. Specifically, for any $l$, we obtain that
\[
\frac{\ud \sum^{N-1}_{i=R_{l+1}+1} p_i(t)}{\ud t}\leq \left\{
\begin{aligned}
    &\mathcal{O}\left(-\frac{1}{(l-1+3)(l+3)}\frac{1}{\mathrm{poly}(n)}\right),\quad \text{if}\quad \sum^{R_l}_{i=R_{l+1}+1} p_i(t)>\frac{1}{(l-1+3)(l+3)}\\
    &0,\quad \text{otherwise}\\
\end{aligned}
\right.\,.
\]
Thus, there exists $T_l=\max\left\{T_{l-1},\mathcal{O}(\mathrm{poly}(n))\right\}$ such that
\[
\sum^{R_l}_{i=R_{l+1}+1} p_i(t)<\frac{1}{l-1+3}-\frac{1}{l+3}=\frac{1}{20}
\]
and
\[
\sum^{N-1}_{i=R_{l+1}+1} p_i(t)<\frac{1}{2}-\frac{1}{l+3}
\]
for any $t>T_l$. This proves \eqref{eqn:iteration}.
\end{proof}

\subsection{Concentration around deterministic dynamics }\label{sec:random_coupling_conv}

In this section, we demonstrate that, when the coupling operator $A$ satisfies \cref{assump:A}, our continuous Lindblad dynamics simulation in \cref{sec:overview_algorithm} approximates $\mathbb{E}\left(\exp(-\mathcal{L}_H S_s)[\rho(T)]\right)$ when $\tau\rightarrow0$, where $\rho(T)$ the solution of the modified Lindblad dynamics \eqref{eqn:modified_lindblad}. 

For simplicity, we omit the phase shift $\exp(-\mathcal{L}_H S_s)$ and assume our simulation scheme can fit into the following classical setting: Given a stopping time $T>0$ and a small time step $\tau>0$ such that $M=T/\tau\in\mathbb{N}$, we approximate $\rho(n\tau)$ using $\rho_n$, which is defined by
\begin{equation}\label{eqn:simulation_update}
\rho_n=\mathcal{F}(\tau,A_n)\rho_{n-1},\quad \text{where}\ \rho_0=\rho(0)\,.
\end{equation}
Here, $\mathcal{F}(\tau,A_n)$ is a linear operator that depends on the random operator $A_n$ and $\{A_n\}^M_{n=1}$ are independently drawn from a probability distribution $\Xi_{A}$. 

Next, we introduce the Frobenius norm of a matrix, which is defined as 
\[
\|A\|_F:=\sqrt{\mathrm{Tr}\left(A^\dagger A\right)}\,.
\]
In addition, given a superoperator $\mathcal{L}$ that acts on matrices, the induced $F$-norm is
\[
\left\|\mathcal{L}\right\|_F:=\sup_{\|A\|_F\leq 1}\left\|\mathcal{L}(A)\right\|_F\,.
\]
Now, we are ready to introduce the approximation result, which is summarized in the following theorem:


\begin{thm}[Concentration of a single trajectory around deterministic dynamics]
Assuming our scheme possesses at least first-order accuracy, meaning that there exists a constant $C_\mathcal{F}$ such that
\[
\left\|\mathcal{F}(\tau,A)-\exp(\mathcal{L}_{A}\tau)\right\|_1\leq C_\mathcal{F}\tau^2
\]
almost surely for $A$ drawn from $\Xi_A$. Define $C_\Xi=\sup_{A\in\mathrm{supp}(\Xi_A)}\max\left\{\|\mathcal{L}_A\|_1,\|\mathcal{L}_A\|_F\right\}$. 
If $\tau=\mathcal{O}\left(\min\left\{\frac{1}{T},\frac{1}{C_\Xi}\right\}\right)$, we have
\begin{equation}\label{eqn:half_order_converge}
\mathbb{E}\left\|\rho_M-\mathbb{E}(\rho(T))\right\|_F=\mathcal{O}\left(\exp(C_\Xi T)C_\Xi\sqrt{T\tau}+C_{\mathcal{F}}T\tau\right)
\end{equation}
\end{thm}
\begin{rem} 
The relationship between the $F$-norm and the trace norm in \eqref{eqn:half_order_converge} can be expressed as follows:
\[
\frac{1}{\sqrt{d}}\left\|\rho_M-\mathbb{E}(\rho(T))\right\|_1\leq \left\|\rho_M-\mathbb{E}(\rho(T))\right\|_F\leq \left\|\rho_M-\mathbb{E}(\rho(T))\right\|_1\,.
\]
It is important to highlight that these inequalities are sharp. Therefore, in the worst-case scenario, the upper bound of $\left\|\rho_M-\mathbb{E}(\rho(T))\right\|_1$ depends exponentially on the number of qubits $n$. It remains an intriguing question whether there exists potential for further improvement in the bound of $\left\|\rho_M-\mathbb{E}(\rho(T))\right\|_1$.
\end{rem}
\begin{proof} Define 
\[
\widetilde{\rho}_n=\left(1+\mathbb{E}(\mathcal{L}_{A_n})\tau\right)\widetilde{\rho}_{n-1},\quad \text{where}\ \widetilde{\rho}_0=\rho(0)\,.
\]
Here $\mathcal{L}_{A_n}$ is defined in \eqref{eqn:Lindblad_dynamics}, and the subscript $A_n$ indicates its dependence on the random matrix $A_n$. We notice that
\[
\frac{\mathrm{d}\mathbb{E}(\rho(t))}{\mathrm{d}t}=\mathbb{E}(\mathcal{L}_A)\mathbb{E}(\rho(t))\,.
\]
Because $\exp(\mathbb{E}(\mathcal{L}_A) t)$ is a completely positive trace-preserving map for any $t\geq 0$, 
\begin{equation}\label{eqn:d_small_1}
\left\|\mathbb{E}(\rho(T))-\widetilde{\rho}_M\right\|_F\leq \left\|\mathbb{E}(\rho(T))-\widetilde{\rho}_M\right\|_1=\mathcal{O}\left(C_\Xi^2T\tau\right)\,.
\end{equation}

Next, we rewrite \eqref{eqn:simulation_update} as
\begin{equation}\label{eqn:simulation_update_new}
\begin{aligned}
\rho_n=&\left(1+\mathcal{L}_{A_n}\tau\right)\rho_{n-1}\\
&+\left(\mathcal{F}(\tau,A_n)-\exp(\mathcal{L}_{A_n}\tau)\right)\rho_{n-1}\\
&+\left(\exp(\mathcal{L}_{A_n}\tau)-\left(1+\mathcal{L}_{A_n}\tau\right)\right)\rho_{n-1}\,.
\end{aligned}
\end{equation}
The second and third terms can be bounded by:
\begin{equation}\label{eqn:d_small_2}
\left\|\left(\mathcal{F}(\tau,A_n)-\exp(\mathcal{L}_{A_n}\tau)\right)\rho_{n-1}\right\|_1\leq C_\mathcal{F}\tau^2\,,
\end{equation}
and
\begin{equation}\label{eqn:d_small_3}
\left\|\left(\exp(\mathcal{L}_{A_n}\tau)-\left(1+\mathcal{L}_{A_n}\tau\right)\right)\rho_{n-1}\right\|_1=\mathcal{O}(C_\Xi^2\tau^2)\,.
\end{equation}

Plugging \eqref{eqn:d_small_2}, \eqref{eqn:d_small_3} into \eqref{eqn:simulation_update_new}, when $\tau<\|\mathcal{L}_{A_n}\|_1^{-1}$, $(1+\mathcal{L}_{A_n}\tau)$ are completely positive trace-preserving maps, and we obtain
\begin{equation}\label{eqn:second_key}
\left\|\rho_M-\Pi^N_{n=1}\left(1+\mathcal{L}_{A_n}\tau\right)\rho_{0}\right\|_F\leq \left\|\rho_M-\Pi^N_{n=1}\left(1+\mathcal{L}_{A_n}\tau\right)\rho_{0}\right\|_1=\mathcal{O}((C_\mathcal{F}+C_\Xi^2)T\tau)\,.
\end{equation}
Finally, we borrow idea from \cite[Lemma 3.7]{Chen_2021} to bound $\mathbb{E}\left\|\Pi^N_{n=1}\left(1+\mathcal{L}_{A_n}\tau\right)\rho_{0}-\widetilde{\rho}_M\right\|^2_F$. First, we rewrite
\[
\begin{aligned}
&\Pi^N_{n=1}\left(1+\mathcal{L}_{A_n}\tau\right)\rho_{0}-\widetilde{\rho}_M\\
=&\underbrace{\left(\mathcal{L}_{A_N}\tau-\mathbb{E}\left(\mathcal{L}_{A_N}\right)\tau\right)\Pi^{N-1}_{n=1}\left(1+\mathcal{L}_{A_n}\tau\right)\rho_{0}}_{\mathrm{(I)}}\\
&+\underbrace{\left(1+\mathbb{E}\left(\mathcal{L}_{A_N}\right)\tau\right)\left(\Pi^{N-1}_{n=1}\left(1+\mathcal{L}_{A_n}\tau\right)\rho_{0}-\Pi^{N-1}_{n=1}\left(1+\mathbb{E}\left(\mathcal{L}_{A_n}\right)\tau\right)\rho_{0}\right)}_{\mathrm{(II)}}
\end{aligned}
\]
The key observation is that only (I) contains the random operator $\mathcal{L}_{A_N}$. Thus, 
\[
\mathbb{E}\left(\mathrm{(I)}^\dagger\mathrm{(II)}\right)=\mathbb{E}\left(\mathbb{E}\left(\mathrm{(I)}^\dagger\mathrm{(II)}\middle|A_{1:n-1}\right)\right)=0\,.
\]
and $\mathbb{E}\left(\mathrm{(II)}^\dagger\mathrm{(I)}\right)=0$. These further imply
\[
\begin{aligned}
&\mathbb{E}\left\|\Pi^N_{n=1}\left(1+\mathcal{L}_{A_n}\tau\right)\rho_{0}-\widetilde{\rho}_M\right\|^2_F=\mathbb{E}\mathrm{Tr}\left[(\mathrm{(I)}+\mathrm{(II)})^\dagger(\mathrm{(I)}+\mathrm{(II)})\right]\\
=&\mathbb{E}\left\|\mathrm{(I)}\right\|^2_F+\mathbb{E}\left\|\mathrm{(II)}\right\|^2_F
\leq \mathbb{E}\left\|\mathrm{(I)}\right\|^2_1+(1+M\tau)^2\mathbb{E}\left\|\left(\Pi^{N-1}_{n=1}\left(1+\mathcal{L}_{A_n}\tau\right)\rho_{0}-\Pi^{N-1}_{n=1}\left(1+\mathbb{E}\left(\mathcal{L}_{A_n}\right)\tau\right)\rho_{0}\right)\right\|^2_F\\
=&4M^2\tau^2+(1+M\tau)^2\mathbb{E}\left\|\left(\Pi^{N-1}_{n=1}\left(1+\mathcal{L}_{A_n}\tau\right)\rho_{0}-\Pi^{N-1}_{n=1}\left(1+\mathbb{E}\left(\mathcal{L}_{A_n}\right)\tau\right)\rho_{0}\right)\right\|^2_F\\
\le&4M^2\tau^2+\exp(2M\tau)\mathbb{E}\left\|\left(\Pi^{N-1}_{n=1}\left(1+\mathcal{L}_{A_n}\tau\right)\rho_{0}-\Pi^{N-1}_{n=1}\left(1+\mathbb{E}\left(\mathcal{L}_{A_n}\right)\tau\right)\rho_{0}\right)\right\|^2_F\,.
\end{aligned}
\]
Repeating the above calculation iteratively, we obtain
\[
\mathbb{E}\left\|\Pi^N_{n=1}\left(1+\mathcal{L}_{A_n}\tau\right)\rho_{0}-\widetilde{\rho}_M\right\|^2_F=\mathcal{O}\left(\exp(2C_\Xi T)C_\Xi^2T\tau\right)\,.
\]
Therefore, 
\begin{equation}\label{eqn:d_small_4}
\mathbb{E}\left\|\rho_M-\widetilde{\rho}_M\right\|_F\leq \left(\mathbb{E}\left\|\Pi^N_{n=1}\left(1+\mathcal{L}_{A_n}\tau\right)\rho_{0}-\widetilde{\rho}_M\right\|^2_F\right)^{1/2}=\mathcal{O}\left(\exp(C_\Xi T)C_\Xi\sqrt{T\tau}\right)\,.
\end{equation}
Combining \eqref{eqn:d_small_1}, \eqref{eqn:second_key}, and \eqref{eqn:d_small_4}, we prove \eqref{eqn:half_order_converge}.
\end{proof}

\end{document}